\pdfoutput=1
\documentclass[a4paper,reprint,twocolumn,notitlepage,aip,nofootinbib]{revtex4-2}
\usepackage[left=1.5cm,right=1.5cm,top=1cm,bottom=1.5cm,includeheadfoot]{geometry}

\usepackage{tgtermes} 
\usepackage{tgheros} 
\usepackage[T1]{fontenc}

\usepackage{amssymb}
\usepackage{amsmath}
\usepackage{amsthm}
\usepackage{graphicx}
\usepackage[dvipsnames]{xcolor}
\usepackage{soul}
\usepackage{float}
\usepackage{multirow}
\usepackage{enumitem}

\PassOptionsToPackage{hyphens}{url} 
\usepackage[breaklinks=true]{hyperref}
\bibpunct{[}{]}{;}{n}{}{} 

\newtheorem{theorem}{Theorem}
\newtheorem{proposition}[theorem]{Proposition}

\newcommand{\R}{\mathbb{R}}

\newcommand{\Z}{\mathbb{Z}}
\newcommand{\eps}{\varepsilon}
\renewcommand{\d}{\,\mathrm{d}} 
\renewcommand{\i}{\mathrm{i}}
\newcommand{\rme}{\mathrm{e}}
\newcommand{\rhogs}{{\rho_\mathrm{gs}}}

\newcommand{\FHKpure}{F_\mathrm{HK1,pure}}
\newcommand{\FHKens}{F_\mathrm{HK1,ens}}
\newcommand{\FSD}{F^0_\mathrm{SD}}

\newcommand{\FLL}{\FCSpure}
\newcommand{\FDM}{\FCSens}

\newcommand{\FCSpure}{F_\mathrm{CS,pure}}
\newcommand{\FCSens}{F_\mathrm{CS,ens}}

\newcommand\FD {F_{\mathrm{D}}}
\newcommand\GD {G_{\mathrm{D}}}
\newcommand\qmstate {\Gamma}

\newcommand{\xx}{\mathbf{x}}
\newcommand{\rr}{\mathbf{r}}

\newcommand{\pp}{\mathbf{p}}

\newcommand{\vv}{\mathbf{v}}

\newcommand{\jpara}{\mathbf{j}^\mathrm{p}}
\newcommand{\jtot}{\mathbf{j}}
\newcommand{\jmag}{\mathbf{j}^\mathrm{m}}
\newcommand{\jarb}{\mathbf{k}}

\newcommand\FGH {G}
\newcommand{\A}{\mathbf{A}}
\renewcommand{\a}{\mathbf{a}}
\newcommand{\B}{\mathbf{B}}
\newcommand\aD {\a}
\newcommand\jparaDm {{\jpara_{\mathrm{m}}}}

\newcommand{\Ah}{\mathbf{\hat{A}}}
\newcommand{\kn}{\mathbf{k}_{\mathbf{n}}}
\newcommand{\bn}{\mathbf{n}}
\newcommand{\ah}{\hat{a}}
\newcommand{\be}{\boldsymbol{\epsilon}}
\newcommand{\bj}{\mathbf{j}}

\DeclareMathOperator{\trace}{Tr}

\begin{document}

\author{Markus Penz}
\address{Basic Research Community for Physics, Innsbruck, Austria}

\author{Erik I. Tellgren}
\address{Hylleraas Centre for Quantum Molecular Sciences, University of Oslo, Norway}

\author{Mih\'aly A. Csirik}
\address{Department of Computer Science, Oslo Metropolitan University, Norway}
\address{Hylleraas Centre for Quantum Molecular Sciences, University of Oslo, Norway}

\author{Michael Ruggenthaler}
\address{Max Planck Institute for the Structure and Dynamics of Matter, Hamburg, Germany}

\author{Andre Laestadius}
\email{andre.laestadius@oslomet.no}
\address{Department of Computer Science, Oslo Metropolitan University, Norway}
\address{Hylleraas Centre for Quantum Molecular Sciences, University of Oslo, Norway}

\title{
	The structure of the density-potential mapping\\
	Part~II: Including magnetic fields
}

\begin{abstract}
The Hohenberg--Kohn theorem of density-functional theory (DFT) is broadly considered the conceptual basis for a full characterization of an electronic system in its ground state by just the one-body particle density. In this Part~II of a series of two articles, we aim at clarifying the status of this theorem within different extensions of DFT including magnetic fields. We will in particular discuss current-density-functional theory (CDFT) and review the different formulations known in the literature, including the conventional paramagnetic CDFT and some non-standard alternatives. For the former, it is known that the Hohenberg--Kohn theorem is no longer valid due to counterexamples. Nonetheless, paramagnetic CDFT has the mathematical framework closest to standard DFT and, just like in standard DFT, non-differentiability of the density functional can be mitigated through Moreau--Yosida regularization. Interesting insights can be drawn from both Maxwell--Schr\"odinger DFT and quantum-electrodynamical DFT, which are also discussed here.
\end{abstract}

\maketitle
\tableofcontents

\newpage
\section{Introduction}
\label{sec:intro}

The celebrated and highly successful method of using the one-body particle density to describe quantum systems---density-functional theory (DFT)---has also been extended to include magnetic fields~\cite{Vignale1987,Vignale1988,vignale-rasolt-geldart1990,Diener,Capelle2002}. In this Part~II of a two-part review series, we will explore such formulations. 
Just like the case for Part~I~\cite{PartI}, the scope of this review is limited to topics closely related to the Hohenberg--Kohn (HK) mapping and properties of the exact functional(s). Here, in Part~II, this is done for extended DFTs to account for magnetic fields. Again, many excellent reviews and textbooks are available on the subject of standard DFT \cite{vonBarth2004basic,burke2007abc,burke2012perspective,dreizler2012-book,eschrig2003-book,parr}. We also direct the interested reader to a rather unique round-table structured article \cite{teale2022round-table} that also discusses extended DFTs. In addition, Part~I can be consulted and is also referenced throughout this part.

In Part~I, we discussed the theoretical aspects of the HK theorem as far as the standard DFT is concerned. In this part, we will continue the study of more general DFTs related to more general Hamiltonians including magnetic fields. 
The formulation using a {\it universal} functional in terms of just the density (valid for all systems in an external electric potential) must then be augmented when the Hamiltonians considered include more than scalar potentials (in addition to parts modelling the internal energy). 
The response of atoms and molecules to strong magnetic fields is
of direct interest in astrophysics~\cite{LAI_RMP73_629}. Moreover, magnetic properties, such
as magnetizabilities and nuclear magnetic resonance parameters, are a
major target of quantum chemistry~\cite{HelgakerJaszunskiRuud}. Other
static magnetic properties
include magnetically induced ring currents due to their statistical
association with other chemical
properties~\cite{GOMES_CR101_1349} and higher-order static
properties~\cite{VAARA_PRL86_3268,PAGOLA_PRA72_033401,CAPUTO_JCP126_154103}.
Additionally, there are many time-
or frequency-dependent properties related to the response to external magnetic fields.
However, standard DFT does not fully describe magnetic properties, thereby motivating the schemes studied here.  

An important ingredient in the density-functional approach is to obtain a universal density functional that is appropriate for the underlying Hamiltonian.
A more general Hamiltonian would intuitively require more variables of the corresponding extended DFT. 
For magnetic systems, natural candidates~\cite{Tellgren2012} to use as variables alongside the particle density are the gauge-invariant total current density $\jtot$ (sometimes also called the `physical' current density) and the paramagnetic current $\jpara$. 
The paramagnetic current must be carefully distinguished from the total current, where the latter cannot be determined from the wave function $\psi$ (or density matrix $\Gamma$) alone since it also includes the external vector potential. 
Both of these current densities have been presented in the literature as variables of current-density-functional theory (CDFT)~\cite{Vignale1987,Vignale1988,vignale-rasolt-geldart1990,Diener,Capelle2002}. 
We will here discuss the paramagnetic and physical CDFTs, as well as other formulations for magnetic systems, with the HK1 and HK2 structure in mind. We direct the interested reader to Refs.\citenum{Tellgren2012} and \citenum{LaestadiusBenedicks2014} for more on the choice of basic variables in CDFT.
In addition, there are other options than a theory formulated with a current density, e.g., the magnetic-field DFT of \citeauthor{GRAYCE}~\cite{GRAYCE}. However, this formalism requires that a semi-universal functional is employed, i.e., utilizing a functional that takes the particular magnetic field of interest as a parameter.

The complexity of formulating a DFT for magnetic systems is apparent from the fact that there is no HK theorem yet proven. This means it is unknown whether the particle density and a current determine the scalar and vector potential of the system~\cite{Capelle2002,Tellgren2012,LaestadiusBenedicks2014,Laestadius2021}. In fact, the theory 
with the paramagnetic current density cannot be used to establish a one-to-one correspondence between the densities and the potentials as first demonstrated by \citeauthor{Capelle2002}~\cite{Capelle2002}. 
Concerning a HK theorem that uses the total current density, the best attempt so far (due to \citeauthor{Diener}~\cite{Diener}) 
is irreparably in error as was recently shown in Ref.~\citenum{Laestadius2021}.
Moreover, even if a HK result could be proven for the total current density, there are serious issues with the HK variational principle and its extension to $N$-representable density pairs~\cite{LaestadiusBenedicks2015}.  Simply put, the total current density is not a suitable, independent variational parameter. However, this difficulty can be circumvented in models that modify the usual Rayleigh--Ritz variational principle. Specifically, as will be discussed below, the introduction of an induced magnetic field as an independent variational parameter in the Maxwell--Schr\"odinger model does permit a DFT formulated using the total current density.~\cite{TellgrenSolo2018}

This review is structured as follows: In Section~\ref{sec:split} we first repeat the restructuring of the HK theorem into two parts, namely HK1 and HK2. Just as in Part~I, this will be a feature of our presentation here for extended DFTs. We thereafter discuss preliminaries in Section~\ref{sec:Prel} where, e.g., different densities (in addition to the one-body particle density) are introduced together with our typical Hamiltonian in Eq.~\eqref{eq:CDFT-Ham-general}. 

In Section~\ref{sec:pCDFT} we then investigate a formulation of CDFT using the paramagnetic current density. For this theory HK1 holds, whereas HK2 does not. The latter fact is also discussed by considering the structure of known counterexamples. 
Paramagnetic CDFT is the formulation that comes closest to the mathematical framework developed by Lieb and others for standard DFT, which is outlined in the section. We also review Moreau--Yosida regularization which can be achieved in these variables (on a reflexive ``density space''). 
The restriction to uniform magnetic fields is also discussed, where the function space of paramagnetic current densities can be reduced to a finite-dimensional vector space. In Section~\ref{sec:mUCP} we briefly discuss the unique-continuation property (from sets of positive measure) for magnetic Schr\"odinger operators that can be used to establish (for an eigenstate) $\psi\neq 0$ almost everywhere (a.e.).
If one treats the magnetic field (or the magnetic vector potential) as a parameter of the system, a semi-universal formulation of CDFT becomes available. This treatment, sometimes referred to as B-DFT, is discussed in Section~\ref{sec:B-DFT}. 

We then continue, in Section~\ref{sec:tCDFT}, with a discussion on CDFT using the total (physical) current density. Here, already HK1 fails and we therefore discuss alternatives, one being the approach of Diener. We also look into partial HK results. Furthermore, in Section~\ref{sec-MDFT}, we discuss how the introduction of an induced classical magnetic field circumvents the difficulties with using the total current density as a variational parameter. We provide HK1 and HK2 in this setting. In Section~\ref{sec:QEDFT}, we consider the natural generalization to a quantized electromagnetic field induced by the electrons.
We conclude our review with a summary in Section~\ref{sec:summary}.

\section{Restructuring the Hohenberg--Kohn theorem}\label{sec:split}

In Part~I of this review, we introduced a convenient and beneficial split of the seminal HK theorem of standard DFT into two separate results:
\begin{itemize}
    \item (\textbf{HK1}) If two potentials share a common ground-state density then they also share a common ground-state wave function or density matrix.
    \item (\textbf{HK2}) If two potentials share any common eigenstate and if this eigenstate further is non-zero almost everywhere then they are equal up to a constant.
\end{itemize}
The combination of both results then gives the full HK theorem that allows a well-defined density-potential mapping in standard DFT. It was shown in Part~I that while HK1 is uncontroversial and in fact follows solely from how the energy functional is defined, HK2 requires a bit more technicality if we want to guarantee that the eigenstate is in fact non-zero almost everywhere. The split of the HK theorem allows to study the status of HK1 and HK2 separately for different versions of DFT for magnetic systems. We will see that automatically HK1 holds, as was already pointed out in Section~X of Part~I, for any variant of DFT that has a \emph{universal} constrained-search functional, i.e., one that varies over the density quantities independent of the external potentials. This will be the case in paramagnetic CDFT (Section~\ref{sec:pCDFT}), yet not in total CDFT (Section~\ref{sec:tCDFT}), where the current density depends on the vector potential. While this raises doubts about the possibility of a full HK theorem in total CDFT, a formulation beyond the proposed split can still be feasible. At present its status is open, and we will summarize the most relevant attempts in Section~\ref{sec:tCDFT}.
The strategy for proving HK2, on the other hand, cannot be generalized from standard DFT to variants involving magnetic fields. For paramagnetic CDFT even a general condition can be derived that facilitates counterexamples (Section~\ref{sec:cond-count}). Does such a failure of the HK theorem, and with it of a unique density-potential mapping, deliver a final blow to these versions of DFT for magnetic systems? Not necessarily, since parts of the theory, like the availability of a density functional, still survive. Further, within a dual setup of density-potential variables, a regularization technique can be used to reinstate a unique \emph{quasi}density-potential mapping (Section~\ref{sec:para-MY}).
An exception is finally given by the most elaborate theory discussed here, quantum-electrodynamical DFT (QEDFT). In QEDFT the expectation value of the quantized electromagnetic field operator enters as another density variable and not only the HK1/HK2 split works out but both theorems can be successfully established (Section~\ref{sec:QEDFT}).

\section{Preliminaries}
\label{sec:Prel}

As pointed out in the previous section, the conceptual centerpiece of DFT is the availability of a density-potential mapping. A set of reduced quantities (densities) then suffices to determine the external parameters (potentials) acting on the system. Conversely, this allows to fully control the densities by adjusting the potentials, most importantly in order to induce the effect of particle interactions into a non-interacting system with the help of the so-called exchange-correlation potential. The mediating tool to establish such a mapping in the ground-state theory is the energy: A functional is set up that takes the external potentials $\vv$ as arguments and minimizes the total energy of the system by varying over all possible densities $\xx$. In Section~X of Part~I we already introduced such an `abstract' formulation of DFT, where the ground-state energy $E[\vv]$ is determined by
\begin{align} 
    \label{eq:Fxx}
    \tilde F[\xx] &= \inf_{\psi\mapsto\xx}\{\langle \psi | H_0 | \psi\rangle\}, \\
    E[\vv] &= \inf_\xx \{ \tilde F[\xx] + f[\vv,\xx]  \} . \nonumber
\end{align}
Here, $H_0$ in the definition of the universal constrained-search functional is such that it contains only internal contributions, i.e., no reference to the external potentials $\vv$ is included. Variation in the definition of $\tilde F[\xx]$ is over all states $\psi$ that yield the given densities $\xx$, indicated by the notation $\psi\mapsto\xx$. What is then missing for the total energy is given by the coupling of the densities to the potentials, summarized by the term $f[\vv,\xx]$. Then this formulation alone already facilitates the HK1 theorem, since the density quantity $\xx$ alone already determines the possible ground-states in Eq.~\eqref{eq:Fxx}. To achieve such a formulation, we must thus separate out all quantities that directly couple to the potentials included in the system. In the case of standard DFT the scalar potential just couples linearly to the one-particle density in the form of a dual pairing, $f[v,\rho] = \langle v,\rho \rangle$. In the presence of a magnetic field it is clear that the one-particle density needs to be complemented with another density quantity that allows to determine the magnetic energy contribution. If one adheres to the dual-pairing structure of standard DFT, then this must be a vector-field quantity, the current density. But it is not strictly imposed that $\vv$ and $\xx$ are dual variables, meaning the potential space is dual to the density space. When this is not the case, the density variable might be redundant, i.e., has more information than needed in order to determine the external potential. A concrete example is found in certain  formulations of collinear spin-DFT, where the external potential is just the scalar potential $v$  and the density variables, $\xx=(\rho_{\uparrow},\rho_{\downarrow})$, are the spin up and down densities. Although it is easier to devise practical approximate functionals with direct access to both spin densities, this leads to the said redundancy and breaks the dual setting. Duality can be restored by including a scalar $B_z'$ (see Table~\ref{table:sum-dfts}) or, equivalently, spin-resolved potentials~\cite{AYERS_JCP124_224108}. Another example is provided by different treatments of spin and orbital effects in the presence of magnetic fields~\cite{CAPELLE_PRL78_1872}. Some formulations rely on the pair $(\rho,\mathbf{m})$, where $\mathbf{m}$ is a possibly noncollinear spin density~\cite{ESCHRIG_JCC20_23,GONTIER_PRL111_153001,EICH_PRL111_156401}. Other formulations rely on a density triple $\xx=(\rho,\mathbf{m},\jpara)$ that gives rise to a spin Zeeman term $\langle \mathbf{m}, \nabla\times\B\rangle$ and an orbital Zeeman term $\langle \jpara, \A \rangle$. Note that, if partial integration can be performed without boundary terms, then $\langle \mathbf{m}, \nabla\times\B\rangle = \langle \nabla\times\mathbf{m}, \B\rangle$. Alternative formulations thus introduce the magnetization current $\jmag = \jpara + \nabla\times\mathbf{m}$ instead, which results in a combined spin and orbital term $\langle \jmag, \A \rangle$.

It is our task now, to determine the appropriate density quantities for DFT including magnetic fields. As we have seen, this goes by writing down the ground-state energy of the system, so the starting point will naturally be the system Hamiltonian.
Throughout this review we will employ atomic units, which only leaves the speed of light $c$ and the vacuum magnetic permeability $\mu_0$ as fundamental constants. The factor $1/c$ that usually still appears in front of vector potentials and magnetic fields can further be absorbed into the corresponding units.
The most general Hamiltonian considered here is the Pauli Hamiltonian,
\begin{equation}\label{eq:CDFT-Ham-general}
\begin{aligned}
    & H[v,w,\B',\A] = \frac{1}{2} \sum_k (-\i\nabla_k + \A(\rr_k))^2 \\
     &\quad + \sum_k v(\rr_k) + \sum_{k<l} w(r_{kl})
     + \sum_k \B'(\rr_k)\cdot\mathbf{S}_k.
\end{aligned}
\end{equation}
We allow for a general interaction term $w$ that depends on the particle distance $r_{kl}=|\rr_k-\rr_l|$, next to the scalar potential $v$, as well as a vector potential $\A$ and a magnetic field $\B'$ that couples to the Pauli matrices.
Although one is used to think that the vector potential and the magnetic field are coupled via $\B = \nabla \times \A$, for the purpose of DFT they can be assumed independent. We thus write $\B'$ for this independent magnetic field. This is especially useful for constructing the Kohn--Sham system, where non-interacting particles are steered by choosing the appropriate external fields. In most cases though, only the vector potential $\A$ is present and we set $\B'=0$. Further, for molecular systems we fix the interaction to a Coulomb potential $w(r_{kl}) = \lambda r_{kl}^{-1}$, where $\lambda=1$ corresponds to full interaction and $\lambda = 0$ to the non-interacting Kohn--Sham system. Then the Hamiltonian from Eq.~\eqref{eq:CDFT-Ham-general} is reduced to
\begin{equation}\label{eq:CDFT-Ham-basic}
  H[v,\A] = \frac{1}{2} \sum_k (-\i\nabla_k + \A(\rr_k))^2
  + \lambda \sum_{k<l} \frac{1}{r_{kl}} 
  + \sum_k v(\rr_k),
\end{equation}
or, written down in its basic components, $H[v,\A] = T_\A + W +  V[v]$. Here the kinetic operator in the presence of a vector potential is $T_\A = \frac{1}{2}\sum_k(-\i\nabla_k + \A(\rr_k))^2$ (sometimes with a minus sign instead in front of $\A$, that comes from the assumed negative charge of the particles, but that can be absorbed into $\A$).
The Hamiltonian without external potentials is then as usual $H_0 = T + W$, where it holds $T_\mathbf{0} = T$. When we talk about a ``non-interacting'' Hamiltonian, this means that $H[v,\A] = T_\A + V[v]$, without the interaction term $W$.

Now, investigating the energy expectation value $\langle \psi| H[v,\A]| \psi \rangle$, we see that we have to rewrite the mixed term 
$ \sum_k \A(\rr_k)\cdot (\langle \psi |  (-\i \nabla_k \psi) \rangle - \langle (\i \nabla_k\psi) | \psi \rangle)$ that arises from squaring out the kinetic term in a form suitable for a density-functional formulation. To that end, next to the one-particle density $\rho_\psi$ (Eq.~(2) in Part~I),
we define the paramagnetic current density of a given $N$-particle pure state $\psi$ in terms of the spin-summed one-particle reduced density matrix,
\begin{equation*}
    \begin{aligned}
        \jpara_\psi(\rr) & = \mathrm{Im} \{ \nabla \gamma_\psi(\rr,\rr')\vert_{\rr'=\rr} \} \\
        &\;= N \sum_{\underline{\sigma}} \int_{\mathbb R^{3(N-1)}} \mathrm{Im}\left\{ \psi^*\nabla_{\rr}\psi \right\} \d \rr_2 \ldots \d \rr_N.
    \end{aligned}
\end{equation*}
Here, $\rr=\rr_1$ and $\underline\sigma = (\sigma_1,\ldots,\sigma_N)$ are the spin degrees-of-freedom. 
If we have an ensemble state given by a density matrix $\Gamma = \sum_j \lambda_j \vert \psi_j\rangle \langle \psi_j \vert$, $\lambda_j \in [0,1]$, $\sum_j \lambda_j = 1$,  then the paramagnetic current is
\begin{equation*}
    \begin{aligned}
        \jpara_\Gamma (\rr) & = \mathrm{Im} \{ \nabla \gamma_\Gamma(\rr,\rr')\vert_{\rr'=\rr} \} = \sum_j \lambda_j \jpara_{\psi_j}(\rr).
    \end{aligned}
\end{equation*}
We also define the total (or `physical') current density for a given density matrix $\Gamma$ and a vector potential $\A$,
\begin{equation*}
\jtot = \jpara_{\Gamma} + \rho_{\Gamma} \A, 
\end{equation*} 
where {\it both} the state {\it and} the vector potential are explicitly needed in the definition. Of course, in the case $\Gamma = \vert \psi\rangle \langle \psi \vert$, we have $\jtot = \jpara_{\psi} + \rho_{\psi} \A$. 

The utility of introducing the paramagnetic current $\jpara$ 
(and also $\jtot$) becomes apparent when we compute the energy expectation value for some pure state $\psi$,
\begin{equation} \label{eq:E-def-jpara}
    \begin{aligned}
    &\langle \psi| H[v,\A] | \psi \rangle  
     =  \langle \psi|H_0|\psi \rangle \\
     &+ \int_{\mathbb R^3} \left(v(\rr) + \tfrac 1 2 |\A(\rr)|^2 \right) \rho_\psi(\rr) \d \rr
     + \int_{\mathbb R^3} \A(\rr) \cdot \jpara_\psi(\rr) \d \rr \\
    &= \langle \psi|H_0|\psi \rangle + \langle v + \tfrac 1 2 |\A|^2 , \rho_\psi \rangle + \langle \A, \jpara_\psi \rangle.
    \end{aligned}
\end{equation}
We have made use of the notation $\langle f, g\rangle = \int f(\rr) g(\rr) \d \rr$ for the dual pairing between potential and density quantities (with obvious extension to vector fields, see Part~I for further details). 
Similarly, for ensemble states given by density matrices we have
\begin{equation} \label{eq:E-def-jpara2}
    \trace (H[v,\A] \Gamma) 
    = \trace (H_0\Gamma) + \langle v + \tfrac 1 2 |\A|^2, \rho_\Gamma \rangle  + \langle \A, \jpara_\Gamma \rangle.
\end{equation}
Note that because of
\begin{equation*}
E[v,\A] = \inf_{\psi} \langle \psi| H[v,\A] | \psi \rangle  = \inf_\Gamma \trace (H[v,\A] \Gamma)
\end{equation*}
it makes no difference for the energy if we minimize over pure states or density matrices since every component (eigenvector) of a density matrix already realizes the ground-state energy as a degenerate ground state.
For the dual pairing $\langle v + \tfrac 1 2 |\A|^2, \rho \rangle$, the definition
\begin{equation}\label{eq:Aabs-u}
u=u[v,\A] = v + \tfrac 1 2 |\A|^2
\end{equation}
allows us to formulate the theory with an effective potential. Here, for the mathematical formulation it becomes necessary that $v$ and $\vert \A \vert^2$ are elements of the same function space that is also dual to the density space. If this holds, and one additionally has that $\rho\A$ is from the same space as the current $\jpara$, we call those spaces \emph{compatible}~\cite{MY-CDFTpaper2019}. This property is important for the convex formulation of 
paramagnetic CDFT and will be described further in Section~\ref{sec:cdft-functionals}.

We define $N$- and $v$-representability for a pair $(\rho,\jpara)$ in an equivalent fashion as in standard (density-only) DFT (see Section~III of Part~I).
Note that we stick to the denomination ``$v$-representable'' from standard DFT instead of saying ``$(v,\A)$-representable''.
The density pair $(\rho,\jpara)$ is said to be

\begin{enumerate}[label=(\roman*)]
	\item pure-state $N$-representable if there is a wave function $\psi$ that has finite kinetic energy such that $\rho_\psi = \rho$ and $\jpara_\psi = \jpara$,
	\item ensemble $N$-representable if there is a density matrix $\Gamma$ that has finite kinetic energy such that $\rho_\Gamma = \rho$ and $\jpara_\Gamma = \jpara$,
	\item pure-state $v$-representable if there exists a potential pair $(v,\A)$, such that the Hamiltonian $H[v,\A]$ has a ground-state wave function $\psi$ with $\rho_\psi = \rho$ and $\jpara_\psi = \jpara$, and
	\item ensemble $v$-representable if there exists a potential pair $(v,\A)$, such that the Hamiltonian $H[v,\A]$ has a ground-state density matrix $\Gamma$ with $\rho_\Gamma = \rho$ and $\jpara_\Gamma = \jpara$.
\end{enumerate}

Contrary to the situation in standard DFT, pure-state and ensemble $N$-representability have to be differentiated, since different results apply. Additionally to the condition $\int |\nabla\sqrt{\rho(\rr)}|^2 \d\rr < \infty$ already known from standard DFT, different conditions involving the paramagnetic current $\jpara$ and the vorticity  $\pmb{\nu}=\nabla\times (\jpara/\rho)$ must hold.
In the construction of \citeauthor{LiebSchrader}~\cite{LiebSchrader}, the velocity field $\jpara/\rho$ must either be curl-free or the number of particles must be $N\geq 4$ and additional decay properties on $\pmb{\nu}$ must hold. This not only allows for a pure state with the required densities, but even in the form of a Slater determinant.
Yet, contrary to standard DFT, it does not give an upper bound on the kinetic energy of the representing determinant.
A different result for ensemble $N$-representability is that of \citeauthor{TellgrenNrep}~\cite{TellgrenNrep}. Here, the integrals $\int |\jpara|^2 / \rho\d\rr$ and $\int (1+r^2) \rho(\partial_\alpha(\jpara_\beta/\rho))^2 \d\rr$ (for all $\alpha,\beta$ that describe the different components of a 3-vector) must be finite. The proof is by direct construction of a one-particle reduced density matrix and a kinetic-energy bound is available as well.
The problem of $v$-representability must be marked as mostly unsolved, just as in standard DFT. Still, these notions are important and ubiquitous in DFT, since the formulation often depends on constraints like $\psi\mapsto (\rho,\jpara)$, which means ``all wave functions $\psi$ that yield the densities $(\rho,\jpara)$'' and that consequently only makes sense if $(\rho,\jpara)$ is pure-state $N$-representable. For $\Gamma\mapsto (\rho,\jpara)$ ensemble $N$-representability would be sufficient. The $v$-representability naturally shows up in connection to the HK theorem: For which density pairs can a unique mapping to potentials be established?

Note that this definition for $v$-representability of the density pair $(\rho,\jpara)$ involving the paramagnetic current directly carries over to the total current: If $(\rho,\jpara)$ is $v$-representable using $(v,\A)$, then also $(\rho,\jtot) = (\rho,\jpara + \rho\A)$ is. On the other hand, it does not really make sense to ask for $N$-representability of a total current in the presence of a vector potential, only for its paramagnetic part. Other realizations of DFT including magnetic fields will include different density and potential quantities, so these notions have to be adopted accordingly.

\section{Paramagnetic CDFT}
\label{sec:pCDFT}

We begin by addressing the status of HK1 and HK2. The energy expressions given in Eqs.~\eqref{eq:E-def-jpara} and~\eqref{eq:E-def-jpara2}, give  
the ground-state energy for a given potential pair $(v,\A)$   
\begin{equation}\label{eq:E-cdft}
\begin{aligned}
E[v,\A] &= \inf_{\psi} \{ \langle \psi |H_0 |\psi \rangle + \langle u[v,\A],\rho_\psi\rangle + \langle \A, \jpara_\psi \rangle \} \\
&= \inf_{\Gamma}\left\{ \trace (H_0 \Gamma) + \langle u[v,\A],\rho_\Gamma\rangle + \langle \A, \jpara_\Gamma \rangle \right\} ,
\end{aligned}
\end{equation}
where we recall the effective potential $u[v,\A] = v + \frac 1 2 \vert \A \vert^2$ from Eq.~\eqref{eq:Aabs-u}. 
Again, just like in the density-only setting with $E[v]$, the structure of $E[v,\A]$ is such that for fixed densities $(\rho,\jpara)$ the terms $\langle u[v,\A],\rho \rangle$ and $\langle \A, \jpara \rangle$ are already fully determined and do not need explicit reference to the wave function or the density matrix. This allows us to establish HK1 as follows.

\begin{theorem}[HK1 for paramagnetic CDFT]\label{th:paraCDFT-HK1}
Let $\Gamma_1$ be a (mixed) ground state of $H[v_1,\A_1]$ and $\Gamma_2$ a (mixed) ground state of $H[v_2,\A_2]$. If $\Gamma_1,\Gamma_2 \mapsto (\rho,\jpara)$, i.e., if these states share the same density pair, then $\Gamma_1$ is also a ground state of $H[v_2,\A_2]$ and $\Gamma_2$ is also a ground state $H[v_1,\A_1]$.
%
\end{theorem}

\begin{proof}
Since we assumed the existence of ground states $\Gamma_1,\Gamma_2$ for the respective potentials, the infimum in Eq.~\eqref{eq:E-cdft}, when varied over density matrices, is actually a minimum.
Further, for $i=1,2$ the energy contributions $\langle \A_i, \jpara \rangle$ and $\langle u[v_i,\A_i],\rho \rangle$ are fixed because $(\rho,\jpara)$ is given and can be taken out of the minimum,
\begin{equation*}
\begin{aligned}
E[v_i,\A_i] = \min_{\Gamma' \mapsto (\rho,\jpara)}\trace (H_0 \Gamma') &+ \langle u[v_i,\A_i],\rho \rangle  \\
&+ \langle \A_i, \jpara \rangle.
\end{aligned}
\end{equation*}
We now note that the remaining minimum includes no reference to the potentials $(v_i,\A_i)$ and is thus determined by the density pair $(\rho,\jpara)$ alone.
This means that $\Gamma_1,\Gamma_2$ are both valid ground states for both Hamiltonians, $H[v_1,\A_1]$ and $H[v_2,\A_2]$.
\end{proof}

The above proof followed precisely the proof structure of Theorem~1 in Part~I, where also an alternative proof was given that follows the more traditional route using energy inequalities. This alternative proof can just as easily be adapted to the paramagnetic CDFT setting.

One has to use a bit of caution in case of degeneracy. This means that there are potentials $(v,\A)$ that lead to a full set of degenerate ground-state wave functions $\{\psi_j\}_j$ that in turn can be combined into mixed states $\Gamma$ and lead to very different density pairs.
For such cases it was shown that the density pair $(\rho,\jpara)$ is not sufficient to determine the full set of degenerate ground-state wave functions $\{\psi_j\}_j$~\cite{LaestadiusTellgren2018}. 
So the usual statement in DFT that ``the density determines the ground state'' cannot be taken for granted if one means to say ``\emph{all} ground states'', after all we do not have a full HK result for paramagnetic CDFT as we will see below.
This case arises, for example, as a general feature of degenerate systems where the degenerate eigenstates have different angular momenta.
What is still true is that, by Theorem~\ref{th:paraCDFT-HK1} above, the density determines some ground state.
Moreover, when $(\rho,\jpara)$ is ensemble $v$-representable from $H[v,\A]$ by a mixed state formed from $r$ degenerate ground states, then any Hamiltonian $H[v',\A']$ that shares this ground-state density pair must have at least $r$ degenerate ground states in common with $H[v,\A]$.\cite{LaestadiusTellgren2018} Thus, any set of Hamiltonians that shares a ground-state density pair $(\rho,\jpara)$ by necessity has to have at least one joint ground state. 
The non-degenerate case was already noted in the case of paramagnetic CDFT by \citeauthor{Vignale1987}~\cite{Vignale1987}.

Is it possible to proceed to the next step and obtain a HK2? Unfortunately not. 
If the external scalar potential $v$ is supplemented by an external vector potential $\A$ that can give rise to magnetic fields, then the HK theorem in general does not hold any more. The reason is that (infinitely) many combinations of scalar and vector potentials could be linked to the same ground state, i.e., the ground state does not uniquely determine its potentials. This even holds when gauge transformations are taken into account that equate equivalent potentials. In the context of CDFT, this was first noted by \citeauthor{Capelle2002}~\cite{Capelle2002}.

The argument in Ref.~\citenum{LaestadiusBenedicks2014} (see also \citeauthor{Tellgren2012}~\cite{Tellgren2012} on the topic of non-uniqueness in paramagnetic CDFT), in a condensed form, is the following: Assume that a one-electron system without a vector potential supports a ground state $\psi_0$.  We can consider, for example, a Hydrogen-like system. The Schr\"odinger equation is then $H[v,\mathbf{0}]\psi_0 = E \psi_0$, with $H[v,\mathbf{0}] = -\tfrac 1 2\nabla^2 + v$, and where 
we assume that $v$ is locally bounded from above. Apart from this we keep $v$ arbitrary.
We then know \cite[Section~11.8]{LiebLoss} that in such a case $\psi_0$ is {\it unique}, real and everywhere greater than zero. 
Now, introduce another system that includes a vector potential in its Hamiltonian. 
Set $\A= \nabla \phi \times \nabla \psi_0$, where the choice of $\phi$ is kept open, and let $\mu\in\R$. We can observe the following facts:
\begin{itemize}
 \item $\nabla \cdot \A$, the divergence of $\A$, equals zero,
 \item the magnetic field, $\B = \nabla \times \A$, is not identically zero (except possibly for some particular 
choices of $\phi$), and therefore $\A$ is not a gradient field,
 \item $\A\cdot \nabla \psi_0 = 0$.
\end{itemize}
Now, consider the Schr\"odinger operator 
\begin{equation*}
\begin{aligned}
H[v - \tfrac 1 2 |\mu \A |^2,\mu \A] &=  \tfrac 1 2(-\i\nabla + \mu \A)^2 + v - \tfrac 1 2 |\mu \A |^2 \\
&=: H_{\mu \A}.
\end{aligned}
\end{equation*}
We notice that $H_{\mu \A}\psi_0 = H[v,\mathbf{0}]\psi_0 = E \psi_0$, because of the facts above. Thus, for any $\mu$ and an arbitrary choice of $\phi$, we have that $\psi_0$ is an eigenstate (not yet the \emph{ground} state) of $H_{\mu \A}$. Consequently, the density $\rho= \psi_0^2$ and the paramagnetic 
current $\jpara_{\psi_0} = \textrm{Im}\{ \psi_0^*\nabla\psi_0\} =0$ (which is zero since $\psi_0$ is real, as noted above) are independent of $\mu$ and $\phi$. Nevertheless, the potential 
$v - \tfrac 1 2| \mu \A |^2$ and vector potential $ \mu \A$ of course depend crucially on $\mu$ and $\phi$. 

\begin{figure}
    \centering
    \includegraphics[width = 0.48\textwidth]{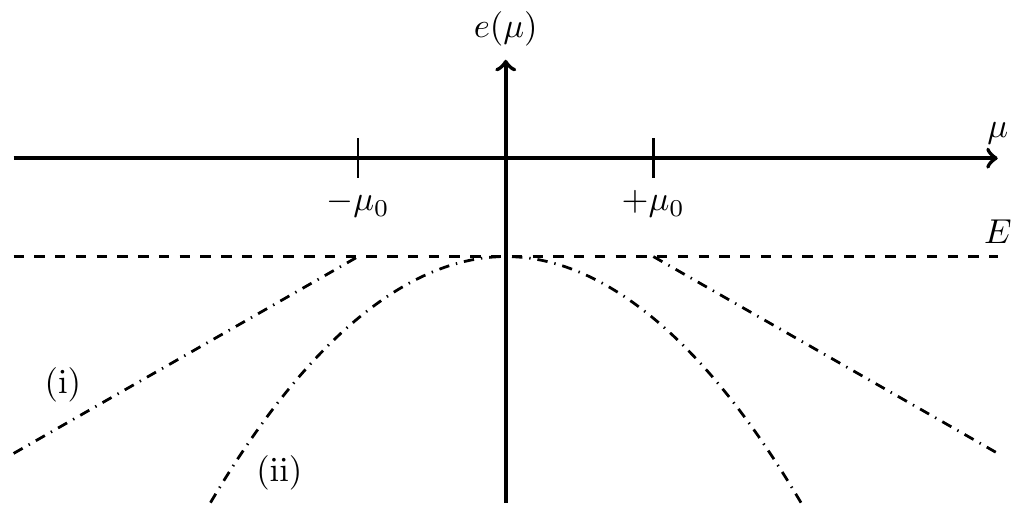}
    \caption{Illustration of the concave $e(\mu)$ (defined in the text) for the two different cases: (i) $e(\mu) = E$ for all $|\mu|<\mu_0$, for some $\mu_0>0$, or (ii) $e(\mu)< E$ for $\mu \neq 0$ and 
$\lim_{\mu \to 0} e(\mu) =E$.}
    \label{Fig:e-mu}
\end{figure}

In the next step it will become clear why we introduced the seemingly unnecessary parameter $\mu$ (since $\phi$ was arbitrary anyway). This is because $\mu$ will be used in proving that $\psi_0$ really is the ground state of $H_{\mu \A}$, at least for small enough $\mu$.
We give an outline of the proof, the full proof can be found in Ref.~\citenum{LaestadiusBenedicks2014} (proof of Theorem~2). We also refer the reader to Ref.~\citenum{Avron1981} as well as to Theorem~4 in the aforementioned Ref.~\citenum{LaestadiusBenedicks2014} for more details on this counterexample. Let $e(\mu) \leq E$ denote the ground-state energy. This is a continuous and even function of $\mu$, i.e., $e(\mu )= e(-\mu)$ and has $e(0)=E$. 
Moreover, since $H_{\mu \A}$ is linear in $\mu$ (not quadratic, since the quadratic term gets canceled), $e(\mu)$ is a concave function.
There are now two possibilities: (i) $e(\mu) = E$ for all $|\mu|<\mu_0$, for some $\mu_0>0$, or (ii) $e(\mu)< E$ for $\mu \neq 0$, and 
$\lim_{\mu \to 0} e(\mu) =E$. Both cases are illustrated in Fig.~\ref{Fig:e-mu}. 
In case (i), the ground-state energy equals $E$ for all $|\mu|<\mu_0$ and $\psi_0$ is the ground state 
for all these $\mu$. In case (ii), let us consider  a ground state $\psi_\mu$. It is then not difficult to prove~\cite{Laestadius2014} that the 
limit wave function  $\lim_{\mu\to 0} \psi_\mu$ is nonzero (i.e., not the zero function) and therefore is an additional ground state of $H[v,\mathbf{0}]$ (but $H[v,\mathbf{0}]$ has a unique ground state). The reason that it is not 
the zero function is that $\int v |\psi_\mu|^2 \d\rr < -C$, for some fixed positive constant $C$. 
The reason that the limit function 
is not $\psi_0$ itself is that $\psi_\mu$ is orthogonal to $\psi_0$ for all $\mu$ (that is, $\langle \psi_\mu |\psi_0 \rangle =0$ for all $\mu$, which would then give a contradiction in the limit $\mu\to 0$).
At any rate, we can conclude that $\psi_0$ is another ground state of the magnetic system for sufficiently small $\mu$. Thus, for magnetic Schr\"odinger operators, the (ground-state) solution does not uniquely determine the potentials. 
In fact, the above argument shows that there are {\it infinitely} many systems that share the same ground state if magnetic fields are included in the formulation. 
The above discussion is a more mathematical construction of the situation first demonstrated by Capelle and Vignale~\cite{Capelle2002} summarized in the following theorem.

\begin{theorem}[Capelle and Vignale~\cite{Capelle2002}]\label{thm:no-HK-p-cdft}
For CDFT formulated with the paramagnetic current density $\jpara$, HK2 does not hold and consequently there cannot be a HK result.
\end{theorem}

We shall further explore counterexamples to the HK theorem in paramagnetic CDFT in the next section. 
However, we first discuss a further subtlety from Ref.~\citenum{LaestadiusTellgren2018}. Suppose now that a given pair $(\rho,\jpara)$ is associated with two different Hamiltonians. Is it then true that the level of degeneracy for these potential pairs needs to be the same for the ground state? This turns out to {\it not} be the case. 
Indeed, we can pick a ground state $\psi_0$ for some system without a magnetic field. Then we can construct a magnetic system
that has a degenerate ground state that includes $\psi_0$. We will return to this matter later in Section~\ref{sec:para-MY} when we discuss Kohn--Sham theory for paramagnetic CDFT. 

\subsection{Discussion of HK2 counterexamples} 
\label{sec:cond-count}

As known in the literature, and here summarized in Theorem~\ref{thm:no-HK-p-cdft}, a full HK theorem for paramagnetic CDFT is {\it not} possible. 
Phrased somewhat differently, \citeauthor{Vignale1987}'s attempted proof~\cite{Vignale1987} of a HK theorem for paramagnetic CDFT suffers from a loophole, since it does not exclude the possibility that two (or more) sets of different potentials share the same ground-state wave function. Explicit counterexamples have been constructed by exploiting that angular momentum is quantized in cylindrically symmetric systems and with very special choices of the magnetic vector potential. Despite these counterexamples, our intuition is that these are exceptions connected to high symmetry. In typical cases, lacking both symmetry and unlikely coincidences, it might hold that no (further) counterexamples exist. Yet, until now a general result along these lines has not been proved. However, we can make the intuition more precise for non-interacting systems.

Consider the non-interacting ($\lambda=0$) $N$-electron Hamiltonian from Eq.~\eqref{eq:CDFT-Ham-basic}, with the substitution $v\mapsto u$ from Eq.~\eqref{eq:Aabs-u}, resulting in $\bar{H}[u,\A] = H[v,\A]$.  Furthermore, we use the decomposition
\begin{equation*}
\begin{aligned}
    \bar H[u,\A] &= \sum_{j=1}^N \bar{h}_j[u,\A] \\
    &= \sum_{j=1}^N \left( -\tfrac{1}{2} \nabla_j^2 - \tfrac{\i}{2} \{\nabla_j, \A(\rr_j)\} + u(\rr_j) \right),
\end{aligned}
\end{equation*}
which is a sum of one-particle terms of the form
\begin{equation*}
\begin{aligned}
    \bar{h}_j[u,\A] &= -\tfrac{1}{2} \nabla_j^2 - \tfrac{\i}{2} \{\nabla_j, \A(\rr_j)\} + u(\rr_j) \\
    &= -\tfrac{1}{2} \nabla_j^2 - \i\A(\rr_j)\cdot\nabla_j - \tfrac{\i}{2} (\nabla_j \cdot \A(\rr_j)) + u(\rr_j).
\end{aligned}
\end{equation*}
Here, $\left\{ \cdot,\cdot \right\}$ denotes the anti-commutator.
When particle indices are superfluous we write simply $\bar{h}[u,\A]$ and we sometimes let a tilde indicate that the divergence term is absorbed into the scalar potential, i.e., $\widetilde{u}(\rr) = u(\rr) - \tfrac{\i}{2} \nabla \cdot \A(\rr)$.
Suppose now that a Slater determinant $\psi_{\mathrm{SD}} = |\phi_1 \phi_2 \ldots \phi_N|$, formed from orthonormal orbitals, is a shared ground state of two such Hamiltonians,
\begin{align*}
    \bar{H}[u_1,\A_1] \psi_{\mathrm{SD}} & = E_1 \psi_{\mathrm{SD}}, \\
    \bar{H}[u_2,\A_2] \psi_{\mathrm{SD}} & = E_2 \psi_{\mathrm{SD}}.
\end{align*}
These $N$-electron equations can also equivalently be written as one-electron equations,
\begin{align*}
    \bar{h}[u_1,\A_1] \phi_k & = \sum_{l=1}^N \varepsilon^{(1)}_{kl} \phi_l, \\
    \bar{h}[u_2,\A_2] \phi_k & = \sum_{l=1}^N \varepsilon^{(2)}_{kl} \phi_l.
\end{align*}
Letting $U = u_2 - u_1$, $\a = \A_2 - \A_1$, $\omega_{kl} = \varepsilon^{(2)}_{kl} - \varepsilon^{(1)}_{kl}$, the difference of the above two equations can be written
\begin{equation*}
  \left( -\tfrac{\i}{2} \{\nabla,\a\} + U \right) \phi_k = \left( -\i\a\cdot\nabla + \widetilde{U} \right) \phi_k = \sum_l \omega_{kl} \phi_l.
\end{equation*}
Noting that unitary transformations within the space of occupied orbitals do not change the total energies $E_1$ and $E_2$, it is always possible to choose the orbitals such that \emph{one} of the $\varepsilon^{(1)}_{kl}$, $\varepsilon^{(2)}_{kl}$, and $\omega_{kl}$ is diagonal. For our purposes, it is convenient to choose orbitals such that $\omega_{kl}$ is diagonal. Hence,
\begin{equation}
  \label{eqVRHKloopholeNecEigProb}
  \left( -\tfrac{\i}{2} \{\nabla,\a\} + U \right) \phi_k = \left( -\i\a\cdot\nabla + \widetilde{U} \right) \phi_k = \omega_{kk} \phi_k.
\end{equation}
Division by $\phi_k$ now yields (assuming $\phi_k\neq 0$ almost everywhere)
\begin{equation*}
  -\i\frac{\a\cdot\nabla \phi_k}{\phi_k} + \widetilde{U} = -\i\a\cdot\nabla \log(\phi_k) + \widetilde{U} = \omega_{kk}.
\end{equation*}
The real part of this expression is somewhat subtle to work with, since $\omega_{kk}$ depends on $(U,\a)$. The imaginary part takes the simple form
\begin{equation}
    \label{eqVRHKloopholeNecCond}
    -\a\cdot\nabla \log(|\phi_k|) - \tfrac{1}{2} \nabla\cdot\a = 0,
\end{equation}
which is equivalent to
\begin{equation*}
    \nabla\cdot |\phi_k|^2 \a = |\phi_k|^2 \nabla\cdot\a + \a\cdot\nabla |\phi_k|^2 = 0.
\end{equation*}
From this divergence condition for the density contribution $\rho_k = |\phi_k|^2$ of each individual orbital also $\nabla\cdot(\rho\a) = 0$ follows since $\rho=\sum_k\rho_k$.
It is instructive to see how Eq.~\eqref{eqVRHKloopholeNecCond} is satisfied in each of the known counterexamples that prevents a full HK result in paramagnetic CDFT:
\begin{itemize}
    \item[(C1)] {\bf Cylindrical symmetry}: For $u_1 = u_2$ that are cylindrically symmetric about the $z$-axis, two vector potentials $\A_1(\rr) = B_1 \mathbf{e}_z \times \rr$ and $\A_2(\rr) = B_2 \mathbf{e}_z \times \rr$ preserve the symmetry and lead to quantized angular momentum. As long as the difference $B_2-B_1$ is not large enough to lead to a level-crossing, the ground-state wave function is therefore the same, and the energies differ by a trivial shift $E_2-E_1=\tfrac{1}{2} (B_2-B_1) L_z$. This holds because the orbitals of both ground states are eigenfunctions of the operator $\{\nabla,\a\}=\tfrac{\i}{2} (B_2-B_1) \hat{L}_z$ and Eq.~\eqref{eqVRHKloopholeNecEigProb} is therefore satisfied. Finally, Eq.~\eqref{eqVRHKloopholeNecCond} is satisfied because $\a(\rr)$ is parallel to the angular direction $\mathbf{e}_z\times\rr$, whereas the gradient $\nabla |\phi_k(\rr)|$ is always contained in the two-dimensional plane spanned by $\mathbf{e}_z$ and $\rr$.
    \item[(C2)] {\bf Real-valued one-electron ground states}: Given a one-electron ground state $\psi$ of $\bar H[u_1,\A_1]$, with $\A_1 = 0$, one can then always construct another Hamiltonian by setting $U(\rr) = u_2(\rr)-u_1(\rr) = 0$ and $\a(\rr) = \A_2(\rr)-\A_1(\rr) = \A_2(\rr) = \nabla\times ( g(\rr) \nabla\psi(\rr) )$. This is possible since the ground state can be chosen as a real-valued function in the absence of a vector potential. In this case, $E_2 = E_1$, $\a$ is divergence free by construction, and $\a(\rr)\cdot\nabla \psi(\rr) = 0$. Because, in the $N=1$ case, there is no distinction between a Slater determinant $\psi_\mathrm{SD}$ and its orbital $\phi$, this also verifies that the necessary condition Eq.~\eqref{eqVRHKloopholeNecEigProb} is satisfied.
    \item[(C3)] {\bf One-electron ground states}: Given a one-electron ground state $\psi$ of $\bar H[u_1,\A_1]$, with $\A_1 \neq \mathbf{0}$, one can then choose $U(\rr) = u_2(\rr)-u_1(\rr) = 0$ and $\a(\rr) = \A_2(\rr)-\A_1(\rr) = \i C \nabla\psi(\rr)^* \times  \nabla \psi(\rr)$, which is real-valued. The constant $C>0$ needs to be chosen sufficiently small not to result in a level crossing. From this choice, it follows that $E_2 = E_1$, $\a$ is divergence free by construction, and $\a(\rr)\cdot\nabla \psi(\rr) = 0$. Identifying $\psi$ and its orbital $\phi$, this also verifies the necessary condition Eq.~\eqref{eqVRHKloopholeNecEigProb}.
    \item[(C4)] {\bf Non-interacting real-valued two-orbital systems}: Let the Slater determinant $\psi_{\mathrm{SD}} = |\phi_1 \phi_2|$ be the ground state of a non-interacting Hamiltonian $\bar H[u_1,\A_1]$, with $\A_1=\mathbf{0}$. The orbitals $\phi_1$ and $\phi_2$ can in this case always be chosen real. Another Hamiltonian sharing the ground state $\psi_{\mathrm{SD}}$ can now be constructed by setting $U(\rr) = u_2(\rr) - u_1(\rr) = 0$ and $\a(\rr) = \A_2(\rr) - \A_1(\rr) = C \nabla\phi_1(\rr) \times \nabla\phi_2(\rr)$, where $C$ is a constant sufficiently small not to result in a level crossing. It follows that $E_2 = E_1$, $\a$ is divergence free by construction, and the necessary condition Eq.~\eqref{eqVRHKloopholeNecEigProb} is satisfied because $\a$ is orthogonal to both $\nabla \phi_1$ and $\nabla\phi_2$.
\end{itemize}

In absence of special symmetries, satisfaction of the necessary condition Eq.~\eqref{eqVRHKloopholeNecCond} becomes increasingly implausible with increasing $N$.
Note that Eq.~\eqref{eqVRHKloopholeNecCond} is of type $\a\cdot\mathbf{x}_k=d$ for all $\mathbf{x}_k = \nabla \log(|\phi_k|)$. Consequently, all $\nabla \log(|\phi_k|)$ must lie in the same affine plane orthogonal to $\a$.
In the absence of special symmetries and for large enough $N$, the fact that orbitals are orthonormal typically leads to orbital gradients that are not contained in the same plane. Even in the presence of a few discrete symmetries, such as $90$ degree rotations or inversion, we would expect special points $\rr$, e.g., symmetry axes or planes, where gradients $\nabla \log(|\phi_k|)$ are confined to a plane to make up a set of measure zero. In summary, we expect the detailed features of a typical ground state $\psi_{\mathrm{SD}}$ to force the conclusion that $\a = \mathbf{0}$.
Interestingly, this critical part that symmetries play for the presence of counterexamples to a possible full HK result in paramagnetic CDFT reminds a lot on a comparable statement in linear-response time-dependent one-body density-matrix-functional theory where no HK-like result is available.\cite{giesbertz2016invertibility}.

The interacting case, $\lambda =1$ in Eq.~\eqref{eq:CDFT-Ham-basic}, is considerably harder to analyze, although the corresponding necessary condition does not appear to be any less restrictive. For a shared $N$-electron ground state $\psi(\rr_1,\dots,\rr_N)$ of two interacting Hamiltonians  $\bar H[u_1,\A_1]$ and $\bar H[u_2,\A_2]$, we obtain
\begin{equation*}
\begin{aligned}
    &\left( \bar H[u_2,\A_2] - \bar H[u_1,\A_1] \right) \psi \\
    &= \sum_{j=1}^N \left( -\i\a(\rr_j)\cdot\nabla_j + \widetilde{U}(\rr_j) \right) \psi = (E_2 - E_1) \psi.
\end{aligned}
\end{equation*}
Division by $\psi$ yields 
\begin{equation*}
    -\i \sum_j \a(\rr_j) \cdot \nabla_j \log(\psi) + \sum_{j=1}^N  \widetilde{U}(\rr_j) = E_2-E_1,
\end{equation*}
which is a highly restrictive condition since a typical wave function is a highly nontrivial function of all particle coordinates simultaneously. All other terms are additive over particle coordinates. In particular, the imaginary part becomes
\begin{equation*}
    \sum_j \a(\rr_j) \cdot \nabla_j \log(|\psi|) = -\frac{1}{2} \sum_{j=1}^N  \nabla_j\cdot\a(\rr_j),
\end{equation*}
which is a highly restrictive condition on the joint many-electron probability distribution $|\psi|^2$.

So far we have focused on necessary conditions that a counterexample must satisfy. It is also possible to derive a near sufficient condition from the requirement that two Hamiltonians commute. We write ``near sufficient" because this only guarantees that they share all eigenstates, not that the energy ordering and therefore the ground states are the same. However, if the difference between the Hamiltonians is made small enough to not induce a level crossing with the ground state, the condition becomes sufficient. Returning to the non-interacting case, we note that two one-electron Hamiltonians commute if and only if
\begin{equation*}
\begin{aligned}
    &[\bar{h}[u_1,\A_1], \bar{h}[u_2,\A_2]] \\
    &=[\bar{h}[u_1,\A_1], \bar{h}[u_2,\A_2] - \bar{h}[u_1,\A_1]] \\
    &= [-\tfrac{1}{2} \nabla^2 - \i\A_1\cdot\nabla + \widetilde{u}_1, -\i\a\cdot\nabla + \widetilde{U}] = 0.
\end{aligned}
\end{equation*}
We would like to check this condition with respect to (C1)--(C4) from above. For this we write out the commutator into its separate parts, substitute $\tilde u_1$ and $\tilde U$, and use that in all counterexamples $U=0$.
\begin{align*}
  &-[\A_1\cdot\nabla + \tfrac{1}{2}(\nabla\cdot\A_1),\a\cdot\nabla + \tfrac{1}{2}(\nabla\cdot\a)]\\
  &+[\tfrac{1}{2} \nabla^2, \i\a\cdot\nabla + \tfrac{\i}{2}(\nabla\cdot\a)] - [u_1,\i\a\cdot\nabla]=0.
\end{align*}
This equation is satisfied by (C1). The other counterexamples of the forms (C2)--(C4) in general all have $[u_1,\i\a\cdot\nabla]\psi \neq 0$ and, since $u_1$ can be chosen independently, this means that the condition is not satisfied.
Hence, these counterexamples are interesting since they involve non-commuting Hamiltonians which do not share all eigenvectors, but nonetheless share ground states.

Finally, in the interacting case, the commutator
\begin{equation*}
\begin{aligned}
    &[\bar{H}[u_1,\A_1], \bar{H}[u_2,\A_2]] \\
    &= [\bar{H}[u_1,\A_1], \bar{H}[u_2,\A_2]-\bar{H}[u_1,\A_1]]
\end{aligned}
\end{equation*}
contains the additional contribution $\sum_j [w, -\i\a(\rr_j)\cdot\nabla_j]$, where we recall that $w(r_{12}) = \lambda r_{12}^{-1}$. For these terms to give vanishing total contribution, we must have
\begin{equation*}
\begin{aligned}
    &\a(\rr_1)\cdot\nabla_1\frac{1}{r_{12}} + \a(\rr_2)\cdot\nabla_2\frac{1}{r_{12}} \\
    &= \left( \a(\rr_1)-\a(\rr_2) \right) \cdot \frac{\rr_1-\rr_2}{r_{12}^3} = 0.
\end{aligned}
\end{equation*}
Hence, for interacting systems, the near-sufficient condition is satisfied for linear vector potentials $\a(\rr) = \mathbf{b}\times\rr + \mathbf{q}$, with $\mathbf{b}$ and $\mathbf{q}$ constant, and excludes all other forms.


\subsection{Paramagnetic CDFT functionals and convex formulation}
\label{sec:cdft-functionals}

Although usually described as the theoretical foundation of DFT, the lack of a (full) HK result for the paramagnetic current density 
does not prevent a mathematical formulation that is very close to the corresponding one in standard DFT (described in detail in Part~I). In fact, HK1 alone is enough to set up a similar hierarchy of functionals for paramagnetic CDFT just as in standard DFT. \citeauthor{Vignale1987}~\cite{Vignale1987} first introduced the correspondence of a HK functional (here denoted $\FHKpure$) and the first mathematical formulation (including a paramagnetic Lieb functional) was done in \citeauthor{Laestadius2014}~\cite{Laestadius2014} for the current vector space $\vec L^1= L^1 \times L^1 \times L^1$ (later refined in Ref.~\citenum{MY-CDFTpaper2019}, see below).  

Let $(\rho,\jpara)$ be associated with a ground state $\psi_{\rho,\jpara}$ that could potentially come from many different potential pairs (due to lack of HK2), but here we suppose that at least one such pair $(v,\A)$ exists which makes $\psi_{\rho,\jpara}$ pure-state $v$-representable. Then
\begin{align*}
    \FHKpure[\rho,\jpara]= \langle\psi_{\rho,\jpara}| H_0 |\psi_{\rho,\jpara} \rangle
\end{align*}
is well-defined due to the availability of HK1 that maps the density pair to a ground state.  A slightly less severe constraint is to instead rely on ensemble $v$-representability (of the density pair) and introduce the functional 
\begin{align*}
    \FHKens[\rho,\jpara]= \trace (H_0 \Gamma_{\rho,\jpara}),
\end{align*}
where $\Gamma_{\rho,\jpara}$ is a ground-state density matrix for at least one $H[v,\A]$ and 
$\Gamma_{\rho,\jpara} \mapsto (\rho,\jpara)$. $\FHKens$ extends $\FHKpure$ to density pairs that are not pure-state $v$-representable but are ensemble $v$-representable. 

Since paramagnetic CDFT inherits (from standard DFT) the fact that not all $(\rho,\jpara)$ are (ensemble) $v$-representable, 
the corresponding constrained-search functionals are useful extensions to all $N$-representable density pairs. They are defined by
\begin{align}
    \label{eq:FCSpure}
    \FCSpure[\rho,\jpara] &= \inf_{\psi \mapsto (\rho,\jpara)} \langle \psi |H_0| \psi\rangle  ,\\
    \FCSens[\rho,\jpara] &= \inf_{\Gamma \mapsto (\rho,\jpara)} \trace( H_0 \Gamma ) . \nonumber
\end{align}
Note that these two functionals are different. The pure-state version was first introduced by \citeauthor{Vignale1987}~\cite{Vignale1987} and the density-matrix version is (in DFT) due to \citeauthor{Valone80}~\cite{Valone80}. 
In Ref.~\citenum{Laestadius2014} (Proposition~8) it was demonstrated that $\FLL[\rho,\jpara]$ is non-convex using the non-convexity of $\FLL[\rho]$ of standard DFT~\cite{penz-DFT-graphs}.
Since $\Gamma \mapsto (\rho_\Gamma,\jpara_\Gamma)$ is linear, it follows that $\FDM[\rho,\jpara]$ \emph{is} convex. 

Before continuing through the hierarchy of paramagnetic current-density functionals we will make some more technical remarks. 
A suitable density space for paramagnetic CDFT was established in Ref.~\citenum{MY-CDFTpaper2019} as
\begin{equation*}
(\rho,\jpara) \in (L^1\cap L^3) \times (\vec L^1\cap \vec L^{3/2}) =: X \times \vec Y,
\end{equation*}
where we use the notation $\vec L^p = L^p \times L^p \times L^p$. Finite kinetic energy of $\psi$ is used for both the $\vec L^1$- and $\vec L^{3/2}$-constraint for $\jpara_\psi$ (recall that $\rho_\psi\in L^1$ for normalized $\psi$ even if $\langle \psi| T | \psi \rangle =+\infty$). Of course, there are also other constraints that could be used to characterize the set of $N$-representable density pairs, such as $\int \vert \jpara \vert^2 / \rho \d \rr < +\infty$ (the current-correction to the von Weiz\"acker term). Both $\FLL$ and $\FDM$ can be defined on the whole of $X \times \vec Y$ simply by setting the values to $+\infty$ when no states exist satisfying the density constraint. The functional $\FLL$ is expectation-valued \cite[Theorem~5]{Laestadius2014}, i.e., there exists a wave function $\psi_0$ such that
\begin{equation*}
\FLL[\rho,\jpara] = \langle \psi_0 | H_0| \psi_0 \rangle, \quad  \psi_0 \mapsto (\rho,\jpara) .
\end{equation*}
This follows from the fact that the set of all wave functions yielding a fixed paramagnetic current density $\jpara$ is weakly closed~\cite{Laestadius2014}. 
For $\FDM$, the fact that there exists a $\Gamma_0$ such that
\[
\FDM[\rho,\jpara] = \trace  (H_0 \Gamma_0) , \quad  \Gamma_0 \mapsto (\rho,\jpara), 
\]
was proven only fairly recently by \citeauthor{Kvaal2021}~\cite{Kvaal2021}.

Just as in standard DFT, we can define a Lieb functional. The convex formulation of paramagnetic CDFT requires a change of variables already mentioned above in Eq.~\eqref{eq:Aabs-u}, i.e., we set $u=u[v,\A] = v + \vert \A\vert^2/2$. Recall the notion of {\it compatibility} of the function spaces that is fulfilled here and requires that $\vert \A\vert^2$ is an element of the dual space of the space of densities. 
(Since we have $\jpara \in \vec L^1\cap\vec L^{3/2}$, the space for vector potentials is $\vec L^3 + \vec L^\infty$ and $|\A|^2 \in L^{3/2} + L^\infty$ giving $u \in L^{3/2} + L^\infty$ as well, i.e., the potential space is the dual of the density space $X=L^1\cap L^3$, as required.) 
We then let $\bar E[u,\A] = E[v,\A]$, which is a jointly concave energy function (this is the reason we call this a convex formulation). 
Now, we can define on $X\times \vec Y$
\begin{equation}\label{eq:paraCDFT-dens-functional}
    F[\rho,\jpara]= \sup_{u,\A} \{ \bar E[u,\A] - \langle u,\rho \rangle - \langle \A, \jpara \rangle  \} .
\end{equation}
This expresses the link between a universal functional of the density pair and the ground-state energy through a Legendre--Fenchel transformation---just as in standard DFT. 
%
Conversely, the Legendre--Fenchel transformation can also be utilized to go back from $F[\rho,\jpara]$ to $\bar E[u,\A]$,
\begin{equation*}
\bar E[u,\A] = \inf_{\rho,\jpara} \{ F[\rho,\jpara] + \langle u,\rho \rangle 
+ \langle \A,\jpara \rangle \}.
\end{equation*}
In analogy with the presentation of standard DFT given in Part~I, 
the HK variational principle for paramagnetic CDFT can now be formulated as
\begin{align*}
    E[v,\A] = \inf_{\rho,\jpara} \left\{ F_\bullet[\rho,\jpara] +  \langle u[v,\A],\rho \rangle + \langle \A, \jpara \rangle \right\}.
\end{align*}
Here, $F_\bullet$ is any of the \emph{admissible} paramagnetic functionals (i.e., any of the above; see Part~I, especially Table I for the full hierarchy of such functionals).

It has recently been proven~\cite{Kvaal2021} that $F$ is lower semicontinuous and that $F = \FDM$. Thus, although paramagnetic CDFT lacks a HK theorem, the equality of the Lieb functional $F$ and the density-matrix constrained-search functional $\FDM$ is carried over to CDFT. This means that $\FDM$ contains the same information as the energy functional $\bar E$ (or $E$).

Finally, we discuss one more paramagnetic current-density functional that connects to the non-interacting reference system used in the Kohn--Sham scheme. For this we take $H_0 = T$ in Eq.~\eqref{eq:FCSpure}, i.e., no interactions are involved, and one restricts the wave functions to single Slater determinants,
\begin{equation*}
\begin{aligned}
    &\FSD[\rho,\jpara] \\
    &= \inf_{\phi \mapsto (\rho,\jpara)}\{ \langle \phi |T | \phi  \rangle \mid \text{$\phi$ is a single Slater determinant} \}.
\end{aligned}
\end{equation*}
Just as before, the zero superscript in the notation indicates that non-interacting systems are considered.
As we described in Section~\ref{sec:Prel}, it has been proven by \citeauthor{LiebSchrader}~\cite{LiebSchrader} that for $N\geq 4$ under mild conditions on $(\rho,\jpara)$ there is always a determinant $\phi$ such that 
$\phi \mapsto (\rho,\jpara)$ is (pure-state) $N$-representable. Just as in density-only DFT, we have in the absence of degeneracy $\FSD[\rho,\jpara]=\FLL^0[\rho,\jpara]$. 
In general, for non-interacting systems, we have two {\it different energies},
\begin{equation*}
\begin{aligned}
    &E^0[v,\A] \\
    &= \inf_{\rho,\jpara} \{  \FLL^0[\rho,\jpara] + \langle v + \tfrac{1}{2}|\A|^2, \rho \rangle + \langle \A, \jpara \rangle  \} \\
    &=\inf_{\rho,\jpara} \{  \FDM^0[\rho,\jpara] + \langle v + \tfrac{1}{2}|\A|^2, \rho \rangle + \langle \A, \jpara \rangle  \}
\end{aligned}
\end{equation*}
and
\begin{align*}
    \tilde E^0[v,\A] = \inf_{\rho,\jpara} \{  \FSD[\rho,\jpara] + \langle v + \tfrac{1}{2}|\A|^2, \rho \rangle + \langle \A, \jpara \rangle  \} .
\end{align*}
The latter one forms the basis of what we could describe as standard Kohn--Sham theory.

\subsection{Regularization and the Kohn--Sham scheme in paramagnetic CDFT}
\label{sec:para-MY}

Besides the counterexamples that make a full HK result in paramagnetic CDFT impossible, the same non-differentiability issues for the density functional $F[\rho,\jpara]$ as in standard DFT~\cite[Section~VII]{PartI} can be expected to arise. 
In Part~I of this review we pointed out the possibility of density-potential mixing that is equivalent to Moreau--Yosida regularization of the functional to circumvent this problem~\cite[Section~IX]{PartI}. In this section, we will show how this technique is applicable to paramagnetic CDFT, for which a detailed account can be found in Ref.~\citenum{MY-CDFTpaper2019}.
Just like in (density-only) standard DFT~\cite{KSpaper2018} this requires   the potential and density spaces to be reflexive and strictly convex. Previously, in Section~\ref{sec:cdft-functionals}, we have chosen the density-current space $X \times \vec Y = (L^1\cap L^3) \times (\vec L^1\cap \vec L^{3/2})$, which is \emph{not} reflexive due to the occurrence of the non-reflexive $L^1$. A possible alternative choice is now the extended space $L^3 \times \vec L^{3/2}$ that we will rely on henceforth in this section. The dual space of potentials is then $L^{3/2} \times \vec L^3$, so every scalar potential is chosen as $v\in L^{3/2}$ and the vector potential as $\A \in \vec L^3$, and both spaces are reflexive and strictly convex. This choice of spaces is still \emph{compatible} in the sense that was given before (i.e., $|\A|^2 \in L^{3/2}$), so we are able to set up the same convex formulation with $F[\rho,\jpara]$ being the Legendre--Fenchel transformation of $\bar E[u,\A]$ from Eq.~\eqref{eq:paraCDFT-dens-functional}.

Now, in order to achieve a unique density-potential mapping, we switch from the densities $(\rho,\jpara)$ to the \emph{quasidensities}
\begin{equation}\label{eq:quasidens-relation}
    (\rho_\eps,\jpara_\eps) = (\rho,\jpara) - \eps J^{-1}(u,\A),
\end{equation}
where the potentials $(u,\A)$ are thought to map to $(\rho,\jpara)$ in the ground state. (Note that this notation with a subscript $\eps$ is exactly opposite to the one chosen in Ref.~\citenum{MY-CDFTpaper2019}, but fits to the one used in Part~I.) Here, $J^{-1}$ is the inverse of the duality map $J:L^3 \times \vec L^{3/2} \to L^{3/2} \times \vec L^3$ that canonically maps the density space to the potential space. The duality map $J$ is just the subdifferential of $\frac{1}{2}\|\cdot\|^2$ on the density space, while $J^{-1}$ is the same on the potential space~\cite[Section~F]{KSpaper2018}. Translated to the language of optimization that we adhere to here, this means determining the minimizer in
\begin{equation*}
    \inf_{u,\A} \{ \tfrac{1}{2}\|u\|^2 + \tfrac{1}{2}\|\A\|^2 - \langle u,\rho \rangle - \langle \A,\jpara \rangle \}
\end{equation*}
to get $J(\rho,\jpara)$ and
\begin{equation*}
    \inf_{\rho,\jpara} \{ \tfrac{1}{2}\|\rho\|^2 + \tfrac{1}{2}\|\jpara\|^2 - \langle u,\rho \rangle - \langle \A,\jpara \rangle \}
\end{equation*}
to get $J^{-1}(u,\A)$. In both cases the minimizer is unique since $\|\cdot\|^2$ is strictly convex. Now the aim is to make a connection between the quasidensities and a new, regularized density functional and to achieve the same kind of uniqueness in the density-potential mapping. For this purpose, add the strictly concave term $-\tfrac{\eps}{2}\|(u,\A)\|^2$ inside the supremum of Eq.~\eqref{eq:paraCDFT-dens-functional} and get a unique maximizer and a regularized functional $F_\eps[\rho,\jpara]$ as a result. It can be shown that this functional is G\^ateaux differentiable (even Fr\'echet differentiable for uniformly convex spaces, which is the case here)~\cite[Th.~9]{KSpaper2018}.

Functional differentiability is needed to set up the usual Kohn--Sham scheme from the expression for the ground-state energy of the non-interacting reference system,
\begin{equation*}
        E^0[v_s,\A_s] 
    = \inf_{\rho,\jpara} \big\{ F^0(\rho,\jpara) + \langle u_s , \rho \rangle 
    + \langle  \A_s, \jpara \rangle  \big\}.
\end{equation*}
Here, $F^0$ is defined with the non-interacting, purely kinetic $H_0=T$ and we have of course $u_s = v_s + \tfrac 1 2 \vert \A_s \vert^2$. A minimizing (ground-state) density pair $(\rho,\jpara)$ of $E^0[v_s,\A_s]$ is then assumed to be the same as for $E[v,\A]$, which represents the interacting system.
In a regularized setting, where the differentials $\nabla F_\eps$ and $\nabla F^0_\eps$ are well-defined, it is then possible to set up the relation
\begin{align*}
    (u_s, \A_s) &=
    (u, \A) + \nabla F_\eps[\rho_\eps,\jpara_\eps] - \nabla F^0_\eps[\rho_\eps,\jpara_\eps] \\
    &= (u, \A) + \nabla (F_\eps[\rho_\eps,\jpara_\eps] - F^0_\eps[\rho_\eps,\jpara_\eps]) \\
    &= (u, \A) + (u_\mathrm{Hxc}, \A_\mathrm{xc}),
\end{align*}
at the ground-state quasidensity pair $(\rho_\eps,\jpara_\eps)$ that relates to the ground-state density pair via Eq.~\eqref{eq:quasidens-relation}. This relation defines the (Hartree-)exchange-correlation potentials $(u_\mathrm{Hxc}, \A_\mathrm{xc})$ that need to be added to the external, given potentials to achieve the same ground-state density pair for the non-interacting system as in the interacting system.
By substituting back to $v_s=u_s-\tfrac 1 2 \vert \A_s \vert^2$, we can write down the corresponding Kohn--Sham equation,
\begin{align*}
 &\left( \frac{1}{2} (-\i\nabla + \A(\rr) + \A_\mathrm{xc}(\rr))^2 + v(\rr) + u_\mathrm{Hxc}(\rr) \right. \\
 &\left. + \frac{1}{2}\left( |\A(\rr)|^2 - |\A(\rr)+\A_\mathrm{xc}(\rr)|^2 \right) \right) \varphi_i(\rr\sigma) = \varepsilon_i \varphi_i(\rr\sigma).
\end{align*}
In order to be able to define the exchange-correlation potentials without depending on differentiability, a different approach has been introduced that defines them just in terms of forces~\cite{tchenkoue2019force}.

When it comes to the choice of spaces, selecting $L^2 \times \vec L^2$ for the density \emph{and} potential spaces (since those spaces are self-dual) would have the benefit that the duality map needed for passing from quasidensities to densities is just the identity map. But this clashes with the already mentioned requirement of \emph{compatibility} since then in general $\vert \A \vert^2 \notin L^2$. Thus, $v=u-\tfrac 1 2 \vert \A \vert^2 \in L^2$ cannot be guaranteed to be from the potential space and the functional derivatives can no longer be decomposed into a scalar and vector potential, i.e., only a $(u,\A)$-formulation (without reference to $v$) is in general possible.

We briefly demonstrated in this section that the regularization strategy that was before worked out in standard DFT~\cite{Kvaal2014,KSpaper2018,Kvaal2022-MY} can also be applied to paramagnetic CDFT. 
Yet, although the strategy is very beneficial in order to get differentiable functionals, a unique quasidensity-potential mapping, and for setting up a well-defined Kohn--Sham scheme, it has not evolved into a practical method as of yet.
On the other hand this form of regularization relates closely to the Zhao--Morrison--Parr method~\cite{ZMP1994} for density-potential inversion which clearly \emph{has} a practical purpose~\cite{Penz2022ZMP}.

In addition to the above outlined approach of achieving functional differentiation in CDFT, Ref.~\citenum{MY-CDFTpaper2019} also demonstrated the construction
of a well-defined Kohn--Sham iteration scheme labeled `MYKSODA'. Although implemented only for a toy model (a one-dimensional quantum ring), the presented MYKSODA is an algorithm for  calculations in the full setting of ground-state CDFT employing a Moreau--Yosida-regularized functional.

\subsection{Uniform magnetic fields in DFT}

The two most commonly calculated static magnetic properties are magnetizabilities and
nuclear shielding constants. For the former, it is sufficient to
restrict attention to uniform magnetic fields. 
As uniform fields are often represented by a linear vector potential
in the cylindrical gauge,
\begin{equation*}
\A(\rr) =\tfrac{1}{2} \B\times(\rr-\mathbf{G}) + \a,
\end{equation*}
where $\B,\a,\mathbf{G} \in \mathbb{R}^3$ are constants, it is
convenient to introduce this as a restriction on the vector
potentials. This enables specialization and simplification of
paramagnetic CDFT, which is formulated above as a theory for general,
nonuniform magnetic field. The resulting theory offers a simplified
framework that retains many of the interesting features of the full
CDFT, such as the gauge dependent basic variables and the choice about
how to incorporate spin.~\cite{Tellgren2018} In particular, the status of the
HK theorem turns out to be intermediate between standard
DFT and paramagnetic CDFT.

With the vector potential determined by a
magnetic field $\B$ and gauge shift $\a$, one finds
that the paramagnetic term is given by
\begin{equation*}
  \langle \A, \jpara \rangle = \frac{1}{2}
  \B\cdot\mathbf{L}_{\mathbf{G}} + \a\cdot\pp,
\end{equation*}
where $\pp = \int \jpara \d\rr$ is the canonical momentum
and $\mathbf{L}_{\mathbf{G}} = \int (\rr-\mathbf{G})\times \jpara \d\rr$ is the
canonical angular momentum relative to $\mathbf{G}$. The reference
point $\mathbf{G}$ is in fact redundant in the sense that we can absorb $-\tfrac{1}{2} \B\times\mathbf{G}$ into the constant $\a$, but it is still a very tangible degree of freedom in actual calculations. It is not only constant in the
sense that it does not vary over space, but also in the sense that we
take it to be fixed even when $\B$ and $\a$ are
varied. Note that being defined with the paramagnetic current, both $\pp$ and $\mathbf{L}_{\mathbf{G}}$ are
gauge dependent, unlike the physical momentum $\boldsymbol{\pi} = \int
\mathbf{j} \d\rr$ and the physical angular momentum $\mathbf{J}_{\mathbf{G}} = \int (\rr-\mathbf{G})\times \mathbf{j} \d\rr$.

That $\pp$ is well-defined is guaranteed by the restriction
$\jpara \in \vec L^1$ required to formulate a paramagnetic CDFT. However,
to guarantee that $\mathbf{L}_{\mathbf{G}}$ is well-defined we make the assumption
that $|\rr | \, \jpara \in \vec L^1$ also, for reasons of
compatibility (see Section~\ref{sec:Prel}) between the density and current density, that
$|\rr|^2 \rho \in L^1$. Hence, we only allow wave functions
with finite second-order moments. Under these conditions, we may specialize the Hamiltonian to uniform fields
\begin{equation*}
  \bar{H}[u,\mathbf{a},\mathbf{B}] = \bar{H}[u, \tfrac{1}{2} \B\times(\rr-\mathbf{G}) + \a],
\end{equation*}
and define the ground state energy functional
\begin{equation*}
\begin{aligned}
  &E[u,\a,\B]\\
  &=
  \inf_{\rho,\pp,\mathbf{L}_{\mathbf{G}}} \Big( F_\mathrm{LDFT}[\rho,\pp,\mathbf{L}_{\mathbf{G}}] + \langle u,\rho
    \rangle + \a\cdot\pp + \frac{1}{2}
    \B\cdot\mathbf{L}_{\mathbf{G}}  \Big)
\end{aligned}
\end{equation*}
with the functional
\begin{equation*}
  F_\mathrm{LDFT}[\rho,\pp,\mathbf{L}_{\mathbf{G}}]= \inf_{\Gamma \mapsto
    \rho,\pp,\mathbf{L}_{\mathbf{G}}} \trace(H_0 \Gamma).
\end{equation*}
This framework is termed LDFT~\cite{Tellgren2018}, with `L' standing for linear vector
potentials or the angular momentum~$\mathbf{L}_{\mathbf{G}}$.

As the triple $(\rho,\pp,\mathbf{L}_{\mathbf{G}})$ is linear in
the density matrix $\Gamma$, the analogue of HK1 holds
automatically. The theory also has a convex structure that immediately
leads to a mapping between supergradients of
$E[u,\a,\B]$ and subgradients of
$F_\mathrm{LDFT}[\rho,\pp,\mathbf{L}_{\mathbf{G}}]$.
However, it is subject to some of the same counterexamples to a full
HK theorem as paramagnetic CDFT: In a cylindrically
symmetric system, the ground-state wave function is piecewise constant
as a function of a magnetic field directed along the symmetry axis and
the energy is piecewise linear. Nonetheless, a stronger result is known
for LDFT than CDFT, because all LDFT counterexamples feature cylindrical
symmetry. Excluding the cylindrically symmetric densities, a HK2-type result is available~\cite{Tellgren2018} if we can take the unique-continuation property (see Section~\ref{sec:mUCP}) for the respective Hamiltonian for granted.

\begin{theorem}
Let $\bar{H}[u_1,\a_1,\B_1]$ and $\bar{H}[u_2,\a_2,\B_2]$ be two
Hamiltonians with non-degenerate ground states $\psi_1$ and $\psi_2$, respectively. Suppose these ground states share the same density triple, i.e., $\psi_1, \psi_2 \mapsto (\rho,\pp,\mathbf{L}_{\mathbf{G}})$. Suppose further that $\rho$ is not cylindrically symmetric about any axis. Then (a) $\psi_1$ and $\psi_2$ are equal up to a global phase, (b) the potentials are equal up to a constant shift $u_1 = u_2 + \text{constant}$, and (c) the vector potentials are equal $(\a_1,\B_1) = (\a_2,\B_2)$.
\end{theorem}

\section{The unique-continuation property for magnetic Hamiltonians}
\label{sec:mUCP}

Generally, the unique-continuation property (UCP) for solutions of the Schr\"odinger equation gives conditions on the involved potentials such that if a (distributional) solution vanishes on a set of positive measure it must vanish everywhere.
The question if the UCP holds when the effect of a magnetic field is taken into account was studied on numerous occasions~\cite{BKRS,Wolff1992,Kurata1,Kurata2,Regbaoui,LaestadiusBenedicksPenz}.
The best result for a Hamiltonian of the type of Eq.~\eqref{eq:CDFT-Ham-general} was established in \citeauthor{Garrigue2020-magneticHK}~\cite{Garrigue2020-magneticHK} and we will repeat it here. The restrictions on the involved potentials is in the form of $L_{\mathrm{loc}}^p$ spaces on the space domain $\R^3$ (the reference gives the more general $\R^d$ but our treatment is for simplicity restricted to $\R^3$). For vector fields this means the space is of the form $\vec L_{\mathrm{loc}}^p = L_{\mathrm{loc}}^p \times L_{\mathrm{loc}}^p \times L_{\mathrm{loc}}^p$.

\begin{theorem}[magnetic UCP]
 Let $\A\in \vec L_{\mathrm{loc}}^q$ and  $|\B'|,\mathrm{div}\,\A, v,w\in L^p_{\mathrm{loc}}$,
 where  $p>2$ and $q>6$. Suppose that $\psi$ is a solution to $H[v,w,\B',\A]\psi=E\psi$. If $\psi$ vanishes on a set of positive measure (or if it vanishes to infinite order at a point), then $\psi=0$. 
\end{theorem}

This result can then be directly used to derive a HK-type result for a given magnetic field $\B'$ and vector potential $\A$ just like in the standard DFT case, see Section~IV of Part~I. Note that this is \emph{not} what one would call a HK-result for CDFT, where it should be possible to determine the magnetic field and/or the vector potential from the given density, maybe including other quantities like the current density. The possibility of such results will be studied in detail in the following sections. To summarize, we give the HK-result that is presented in Theorem~1.5 of \citeauthor{Garrigue2020-magneticHK}~\cite{Garrigue2020-magneticHK}.

\begin{theorem}[magnetic HK]
Let $\A\in \vec L^q + \vec L^\infty$, $\B'\in \vec L^p + \vec L^\infty$ and $v_1,v_2,w\in L^p + L^\infty$ with $p>2$ and $q>6$. If there are two normalized ground states $\psi_1$ and $\psi_2$ of 
$H[v_1,w,\B',\A]$ and $H[v_2,w,\B',\A]$, respectively, such that $\rho_{\psi_1}=\rho_{\psi_2}$, then the potentials $v_1,v_2$ are equal up to a constant.
\end{theorem}

\section{Magnetic-field DFT}
\label{sec:B-DFT}

As an alternative to paramagnetic CDFT, it is possible to construct a theory more like standard DFT but parametrized by the magnetic field.
Such a theory 
is commonly referred to as magnetic-field DFT (BDFT for short) and is due to \citeauthor{GRAYCE}~\cite{GRAYCE}.
We denote the Grayce--Harris semiuniversal density functional by
\begin{equation*}
\begin{aligned}
&\FGH[\rho,\A]  = \inf_{\Gamma \mapsto \rho} \trace(H[0,\A]\Gamma),
\end{aligned}
\end{equation*}
such that the ground-state energy can be written 
\begin{equation}\label{eq:GHenergy}
    E[v,\A]= \inf_{\rho} \left\{  \FGH[\rho,\A]  + \langle v ,\rho \rangle  \right\}.
\end{equation}
The Grayce--Harris functional with the diamagnetic term removed is related to other functionals through partial Legendre--Fenchel transformations of its arguments~\cite{TellgrenSolo2018,REIMANN_JCTC13_4089}. In particular, as was exploited by \citeauthor{Laestadius2021}~\cite{Laestadius2021}, we can connect the Grayce--Harris functional to the 
previously introduced paramagnetic functional(s) $\FCSpure[\rho, \jpara]$ (and $\FDM[\rho,\jpara]$)
\begin{equation}\label{eqGHfun}
\begin{aligned}
&\FGH[\rho,\A]  = \inf_{\Gamma \mapsto \rho} \trace(H[0,\A]\Gamma)\\
&\quad  =  \langle \tfrac{1}{2}  \vert \A \vert^2 , \rho \rangle  + \inf_{\jpara} \left\{ \FCSpure[\rho,\jpara] +\langle \A , \jpara \rangle   \right\}.
\end{aligned}
\end{equation}
It is interesting to note that $\FGH[\rho,\A]$ is nonconvex in $\A$ (Proposition~1 in Ref.~\citenum{Laestadius2021}), such that it can describe not just diamagnetic systems.

Equation~\eqref{eq:GHenergy} is the BDFT variational principle. 
In our lingo, we can note that HK1 is available to us through this semiuniversal nature of $\FGH[\rho,\A]$: Suppose for a given $\A$, we have a density $\rho$ that comes from two different Hamiltonians $H[v_1,\A]$ and $H[v_2,\A]$, then
\begin{equation*}
    \begin{aligned}
        \FGH[\rhogs,\A]  &= \inf_{\Gamma \mapsto \rhogs} \trace(H[0,\A]\Gamma), \\
        E[v,\A]  &=   \FGH[\rhogs,\A]  + \langle v ,\rhogs   \rangle 
    \end{aligned}
\end{equation*}
imply that the two Hamiltonians $H[v_1,\A]$ and $H[v_2,\A]$ must share a ground state.

Furthermore, we also have a type of full HK result: For every fixed $\A$, a positive ground-state density $\rhogs(\rr) > 0$ (which follows from a magnetic UCP almost everywhere, see Section~\ref{sec:mUCP}) determines $v$ up to a constant~\cite{GRAYCE}. Again, in our lingo, this can be seen through the next step of a HK2. Simply use the common ground state from HK1 and subtract the two Schr\"odinger equations (recall that $\A$ is fixed and the same). After multiplication with $\psi_\mathrm{gs}^*$, integrating out all particle positions $\rr_2,\ldots,\rr_N$ and dividing by $\rhogs$ then establishes that $v_1-v_2 = \text{constant}$.

\section{Total CDFT}
\label{sec:tCDFT}

Based on the gauge invariance and the fact that the total (physical) current is used as a basic variable in time-dependent CDFT~\cite{VIGNALE_PRB70_201102,VIGNALE_PRL77_2037}, it seems a natural approach to also use this current (and not the paramagnetic current) density for the theory without time dependence. 
Moreover, as will be discussed in Section~\ref{sec-MDFT}, the Maxwell--Schr\"odinger energy minimization principle also leads to a DFT formulated with the total current density. We therefore  
now turn to the question of formulating CDFT using the total current density.  
Recall that, for given wave function $\psi$ (or a density matrix $\Gamma$) and a vector potential $\A$, we define the total current density $\jtot = \jpara_\psi + \A \rho_\psi$ (or with $\jpara_\Gamma$ and $\rho_\Gamma$). We will investigate two different routes of formulating the theory, that is

\begin{enumerate} [label=(\roman*)]
    \item varying only $\psi$ (or $\Gamma$) which then requires that $\A$ is fixed and known, and 
    \item having $\jtot$ as an entirely free parameter, however, still assuming that there exists some $(\tilde v,\tilde \A)$ that has the given density pair $(\rho,\jtot)$ as ground-state densities. 
\end{enumerate}

We shall here see that both formulations run into problems. For simplicity we will restrict the discussion to pure states (but the reader can freely replace $\psi$ by $\Gamma$ and the proper adjustments). We shall also look at what results on HK-type theorems using the total current density can be obtained using a different methodology than the partitioning into HK1 and HK2. These results are unfortunately {\it quite restrictive}. 
Again, the UCP will play a role, i.e., a ground state $\psi_0$ of the given Hamiltonian is almost everywhere nonzero such that we can divide by it and still make statements true for the full domain considered (almost everywhere). (See Section~\ref{sec:mUCP} above for further details) 

To begin our study, 
if the Hamiltonian $H[v,\A]$ has a ground state $\psi_0$, then the total current density is given by 
\begin{equation} \label{eq:j-to-a}
\jtot = \jpara_{\psi_0} + \rho_0\A,
\end{equation}
with $\rho_0 = \rho_{\psi_0}$. 
To make the connection between a $(\rho,\jtot)$-density functional and the expectation value of the energy for a discussion on the HK1 and HK2 structure, we write for a free $\jtot$
\begin{equation}\label{eq:a-D}
\begin{aligned}
    \jtot = \jpara_\psi + \frac{\jtot - \jpara_\psi}{\rho_\psi} \rho_\psi  =: \jpara_\psi + \a[\psi; \jtot] \rho_\psi.
\end{aligned}
\end{equation}
Note that the last equality defines a vector $\a = \a[\psi;\jtot]$ that is well-defined as long as $\rho_\psi \neq 0 $, which is guaranteed by the UCP (Section~\ref{sec:mUCP}).  
Then, it holds using Eq.~\eqref{eq:E-cdft}
\begin{equation} \label{eq:EjCDFT0}
\begin{aligned}
    E[v,\A] = \inf_{\psi}\big\{ \langle \psi|H_0|\psi \rangle - \langle \a[\psi; \jtot] \cdot \A , \rho_\psi \rangle &  \\
     + \langle \A,   \jtot \rangle  + \langle  v + \tfrac 1 2 \vert \A \vert^2, \rho_\psi\rangle     \big\} &.
\end{aligned}
\end{equation}
This equation will be the starting point of our analysis here since it realizes the desired linear coupling between the total current and the vector potential.

For an approach where $\A$ is fixed in $\jtot = \jpara_\psi + \rho_\psi \A$ (i.e., only varying $\psi$), we obviously can take $\a[\psi; \jtot] = \A$ for all considered $\psi$'s and Eq.~\eqref{eq:EjCDFT0} reduces to
\begin{equation} \label{eq:EjCDFT1}
\begin{aligned}
E[v,\A] &= \inf_{\psi}\left\{ \langle \psi |H_0 |\psi \rangle   + \langle \A,   \jtot \rangle  + \langle  v - \tfrac 1 2 \vert \A \vert^2, \rho_\psi\rangle     \right\},\\
\jtot &= \jpara_\psi + \rho_\psi \A .
\end{aligned}
\end{equation}
Now, in an attempt to obtain a HK1 result, assume that $\rho$ and $\jtot$ are fixed such that the r.h.s.\ in Eq.~\eqref{eq:EjCDFT1} becomes
\begin{equation*}
\begin{aligned}
     \inf_{\psi\mapsto (\rho,\jtot)} \left\{ \langle \psi | H_0 |\psi \rangle  \right\}  + \langle \A, \jtot \rangle  + \langle v - \tfrac 1 2 \vert \A \vert^2, \rho \rangle & \\
     =: F_1[\rho,\jtot]  + \langle \A,   \jtot \rangle  + \langle v - \tfrac 1 2 \vert \A \vert^2, \rho \rangle &,  
\end{aligned}
\end{equation*}
where the last equality defines the functional $F_1$. (The idea would then be to vary over $(\rho,\jtot)$ to obtain $E[v,\A]$.) 
However, a more careful consideration of $F_1[\rho,\jtot] = \inf_{\psi\mapsto (\rho,\jtot)}\left\{ \langle \psi| H_0 |\psi \rangle  \right\} $ is needed. The notation $\psi\mapsto (\rho,\jtot)$ here assumes that the admissible set of wave functions satisfies
\begin{equation}\label{eq:jtotjpara}
\rho_\psi = \rho, \quad \jpara_\psi + \rho\A = \jtot,
\end{equation}
i.e., the functional has a parametric dependence on $\A$ through the minimization domain, and we (must) write $F_1[\rho,\jtot] = F_1[\rho,\jtot;\A]$. Thus, while it holds
\[
E[v,\A] = \inf\left\{F_1[\rho,\jtot;\A]  + \langle \A,   \jtot \rangle  + \langle v - \tfrac 1 2 \vert \A \vert^2, \rho \rangle  \right\},
\]
$F_1$ is not a universal functional of $(\rho,\jtot)$ since its search domain over $\psi$ depends on $\A$. The observant reader might already have noted that in this case
\begin{align*}
    \langle \A,   \jtot \rangle  + \langle v - \tfrac 1 2 \vert \A \vert^2, \rho \rangle
    =
    \langle \A,   \jpara \rangle  + \langle v + \tfrac 1 2 \vert \A \vert^2, \rho \rangle
\end{align*}
by Eq.~\eqref{eq:jtotjpara}. Moreover, $F_1[\rho,\jtot;\A] = \FLL[\rho,\jpara]$ and the energy in this formulation is simply reduced to the corresponding one of paramagnetic CDFT (see Section~\ref{sec:pCDFT}). Consequently, such a total current formulation that just has been presented is nothing but a more or less obvious reformulation of paramagnetic CDFT. This observation has not gone unnoticed in the literature and we refer to Refs.~\citenum{Vignale-comment-2013} and \citenum{Tellgren2018} for a further pedagogical discussion of this fact.

Let us now continue and attempt to obtain a HK1 result. Suppose that $\rho$ and $\jtot$ are fixed such that the rhs. in Eq.~\eqref{eq:EjCDFT0} becomes
\begin{equation} \label{eq:EjCDFT}
\begin{aligned}
         \inf_{\psi\mapsto \rho} & \left\{ \langle \psi | H_0 | \psi \rangle - \langle \a[\psi; \jtot] \cdot \A , \rho \rangle \right\}  \\
         &+ \langle \A, \jtot \rangle 
         + \langle v + \tfrac 1 2 \vert \A \vert^2, \rho \rangle . 
\end{aligned} 
\end{equation}
Note that, for the first term in Eq.~\eqref{eq:EjCDFT}, we only need to restrict $\psi$ such that $\rho_\psi = \rho$. 
Note, in particular, the term $\langle \a[\psi; \jtot] \cdot \A , \rho \rangle $ inside the constrained search $\inf_{\psi\mapsto (\rho,\jtot)}\left\{ \langle \psi |H_0 |\psi \rangle - \langle \a[\psi; \jtot] \cdot \A , \rho \rangle \right\}$ (the would-be $F$-functional). 
Consequently, we have failed to obtain the form (see Section~\ref{sec:Prel})
\begin{align*} 
E[\vv] = \inf_\xx \{ F[\xx] + f[\vv,\xx]  \}, \quad \vv = (v,\A), \quad \xx = (\rho,\jtot) .
\end{align*}
Rather we have obtained $F[\rho,\jtot; \A]$, which is not universal in the sense that it depends on the vector potential. 
Although $(\rho,\jtot)$ is fixed, different potentials $\A$ might alter the selection of $\psi$ in the constrained-search functional and 
based on this partitioning alone, it is not clear that if two potential pairs share $(\rho,\jtot)$ then they also share a ground state. 
We will come back to this matter below when discussing Diener's approach.


In \citeauthor{Diener}~\cite{Diener} an unorthodox formulation of total CDFT was undertaken, including an attempted HK theorem for the total current density based on a suggested new Rayleigh--Ritz variational principle. 
In \citeauthor{Tellgren2012}~\cite{Tellgren2012} it was pointed out that a crucial step of the argument was left unmotivated: The strict inequality in Diener's generalized variational principle was not motivated (see next section). 
Moreover, further technical issues were raised in \citeauthor{Laestadius2014}~\cite{LaestadiusBenedicks2014}. 
Diener's approach is interesting for the reason that it comes very close to succeeding. 
Nonetheless, in \citeauthor{Laestadius2021}~\cite{Laestadius2021} it was finally proven that Diener's approach is 
unfortunately irreparably wrong. We will give a brief summary in the next section.

\subsection{Diener's formulation}
\label{sec:Diener}

\citeauthor{Diener}~\cite{Diener} gave a very interesting attempt to achieve 
a formulation using the total current density.
In particular, he tried to establish a
ground-state DFT of the total current density as well as a 
HK-like result. 
In Ref.~\citenum{Laestadius2021} Diener's attempt was reinterpreted based on a maximin variational principle and, using elementary facts about convexity, it was proven that Diener's approach does not give the correct ground-state energy. Further, it was shown that the suggestion of a HK result is irreparably flawed. We will here outline parts of the argument in Ref.~\citenum{Laestadius2021}. 

Diener's formalism can be simplified by algebraically manipulating the ground-state energy formula until we obtain a variational expression that can be related to his working equations. To give a brief outline, we first recall Section~\ref{sec:B-DFT}
and rewrite the Grayce--Harris functional in Eq.~\eqref{eqGHfun} with $\jarb$ denoting an arbitrary current density, 
\begin{equation*}
  \begin{aligned}
    &G[\rho,\A] =  \langle \tfrac{1}{2}  \vert\A \vert^2, \rho \rangle \\
    &+ \inf_{\jpara} \left\{  \FCSpure[\rho,\jpara] +\langle \A , \jpara \rangle  - \inf_{\jarb} \int \frac{|\jpara + \rho \A - \jarb|^2}{2\rho} \d\rr \right\}
    \\
    & = \inf_{\jpara} \sup_{\jarb} \left\{ \FCSpure[\rho,\jpara]  + \langle \A , \jarb \rangle - \int \frac{|\jpara-\jarb|^2}{2\rho} \d\rr \right\}  .
  \end{aligned}
\end{equation*}
The total current density is reproduced when $\mathbf{k} = \jpara + \rho \A$, also 
solving the minimax problem. We can remark that the issues related to that the correct energy cannot be obtained from a minimization principle for the total current density is mitigated through the above manipulations.
Now, it is a general fact that $\inf_x \sup_y f(x,y) \geq \sup_y \inf_x f(x,y)$, such that  
we next obtain
\begin{equation}
  \label{eqDinmaximin}
  \begin{aligned}
    & G[\rho,\A] \\
    &\geq \sup_{\jarb} \inf_{\jpara} \left\{\FCSpure[\rho,\jpara] + \langle \A , \jarb \rangle  - \int \frac{|\jpara-\jarb|^2}{2\rho} \d\rr \right\}
    \\
    & = \sup_{\jarb} \left\{ \FD[\rho,\jarb] + \langle \A , \jarb \rangle  \right\}
     =: \GD[\rho,\A].
  \end{aligned}
\end{equation}
The last equality is a definition that defines 
$\GD[\rho,\A]$. This, furthermore, identifies Diener's proposed total current-density functional 
\begin{equation}
  \label{eqDienerFun}
  \begin{aligned}
    \FD[\rho,\jarb] & = \inf_{\jpara}  \left\{ \FCSpure[\rho,\jpara] - \int \frac{|\jpara-\jarb|^2}{2\rho} \d\rr \right\}
    \\
      & = \inf_{\Gamma \mapsto \rho} \left\{ \trace( H_0\Gamma ) - \int \frac{|\jpara_{\Gamma}-\jarb|^2}{2\rho} \d\rr \right\}.
  \end{aligned}
\end{equation}
The functional (defined in the right-hand side of Eq.~\eqref{eqDinmaximin}) $\GD$ is convex in $\A$, i.e., the map $\A \mapsto \GD[\rho,\A]$ for fixed $\rho$ is convex. 
Consequently, $\GD$ can \emph{only} describe diamagnetic systems, whereas the Grayce--Harris functional $\FGH[\rho,\A]$ is nonconvex in $\A$. This leads to the fact that (Proposition~2 in Ref.~\citenum{Laestadius2021}): 
  For some $(\rho,\A)$, we have a strict inequality $\FGH[\rho,\A] > \GD[\rho,\A]$.

A question is then (notwithstanding the above) whether Diener's functional $\FD[\rho,\jarb]$ and the variational principle for $\GD[\rho,\A]$ are useful for reconstructing the correct external vector potential from an input pair $(\rho,\jtot = \jpara + \rho\A)$. 
This, together with the (Grayce and Harris) BDFT extension of the HK theorem to determine $v$ (see Section~\ref{sec:B-DFT}), would 
establish a HK-type mapping, i.e.: $(v,\jtot)$ determines $(\rho,\A)$ up to a gauge. 
 

Using Eq.~\eqref{eq:a-D}, we can express a relation between a state $\qmstate$ and an arbitrary vector field $\jarb$ through the effective vector potential
\begin{equation*} 
  \aD(\Gamma,\jarb) = \frac{\jarb - \jpara_{\qmstate}}{\rho_{\qmstate}}.
\end{equation*}
Similar to Eq.~\eqref{eq:j-to-a} we have $\jarb = \jpara_{\qmstate} + \rho_{\qmstate} \aD(\qmstate,\jarb)$, imitating the standard relationship between the total current density, the paramagnetic current, and the actual external vector potential of the system. Suppose now that $\jtot = \jpara + \rho \A$ is the correct ground-state total current density
of a magnetic system described by $\A$. Built into Diener's construction is a possible HK-type mapping, which is clear if we make the following observation: if $\FD(\rho,\jtot)$ always yields a minimizer $\qmstate_{\mathrm{m}}$ in Eq.~\eqref{eqDienerFun} such that $\aD(\qmstate_{\mathrm{m}},\jarb) = \A$, then $(\rho,\jtot)$ determines $\A$. The main result of BDFT (as described in Section~\ref{sec:B-DFT}) would then imply that also the scalar potential would be determined up to an additive constant. 
Elaborating a little further, since the input to the functional $\FD$ is gauge invariant, the external vector potential can at best be determined up to a gauge. Thus, we allow for multiple gauge dependent minimizers $\jparaDm$ in Eq.~\eqref{eqDienerFun} (each coming from a $\qmstate_\mathrm{m}$ with $\aD(\qmstate_{\mathrm{m}},\jarb) = \A + \nabla f$) and where one corresponds to a gauge in which $\aD(\qmstate_{\mathrm{m}},\jarb) = \A$. This {\it would} then be the HK-type mapping resulting from Diener's functional. Alas, the next proposition shows that such an $\FD$-based mapping does not exist. 
\begin{proposition}[Proposition~3 in Ref.~\citenum{Laestadius2021}] \label{prop:EITmasterwork}
  For some $(\rho,\A)$, Diener's current density functional $\FD$ fails to reconstruct the external potential. That is, for any minimizer $\jparaDm$ in Eq.~\eqref{eqDienerFun} we have
  \begin{equation*}
        \frac{\jtot - \jparaDm}{\rho} \neq \A.
  \end{equation*}
\end{proposition}

\subsection{Partial HK results} 

We will finish our discussion about total CDFT considering when a HK result can actually be proven. As will be evident, these are quite restricted results. In the one-electron case, a HK theorem follows from $N$-representability constraints and no assumption that the density is a ground-state density is even necessary. Wherever $\rho(\rr) \neq 0$, we can directly reconstruct the external magnetic field as the vorticity 
\begin{equation*}
  \pmb{\nu} = \nabla\times\frac{\mathbf{j}}{\rho} = \nabla\times\frac{\jpara + \rho\A}{\rho} = \mathbf{0} + \B,
\end{equation*}
since $\nabla \times (\jpara/\rho)=\mathbf{0}$ in the one-particle case.
The above HK result for one-electron systems leaves open what happens if the density vanishes on finite volume of space (see Section~\ref{sec:mUCP} for conditions when this cannot happen). Idealized model cases where this happens have been discussed in connection with the Aharonov--Bohm effect. When the density vanishes on an infinitely long cylindrical or tube-shaped region, it follows from the Byers--Yang theorem~\cite{Byers_1961} that the total current density is a periodic function of the flux inside the tube. Hence, magnetic fields that differ in zero-density regions can produce the same total current density. This type of counterexample works for both, one-electron systems and many-electron systems.

\begin{theorem}\label{HK:N1}
   For one-electron systems in a magnetic field, the total current $\jtot$ and the particle density $\rho(\rr) \neq 0$ a.e.\ determine $(v,\A)$ up to a gauge transformation.
\end{theorem}

In (the very restricted) case of $\jpara=0$ we have the even stronger result that the vector potential gets fully determined. This can be stated also for many-electron systems if additionally to $\rho$ and $\jtot$ also $\jpara$ is given.

\begin{theorem}\label{HK:jpnull}
   The triple $(\rho,\jpara,\jtot)$, with $\rho(\rr) \neq 0$ a.e., determines $\A$ and $v$ up to an additive constant. 
\end{theorem}

\begin{proof}
By $\A = (\jtot-\jpara)/\rho$
the vector potential already gets fixed. Then Section~\ref{sec:B-DFT} describes how to determine $v$ up to a constant.
\end{proof}


To the best of our understanding, and besides the above two results, all known attempts in the literature fall short of a general HK result for the total current density.

\section{Maxwell--Schr\"odinger DFT}
\label{sec-MDFT}

\newcommand \Aind {\A_{\mathrm{ind}}}
\newcommand \Bind {\B_{\mathrm{ind}}}
\newcommand \Atot {\A_{\mathrm{tot}}}
\newcommand \Atotix[1] {\A_{\mathrm{tot};{#1}}}
\newcommand \Btot {\B_{\mathrm{tot}}}

\newcommand \Aext {\mathbf{A}_{\mathrm{ext}}}
\newcommand \Aextix[1] {\mathbf{A}_{\mathrm{ext};{#1}}}
\newcommand \Bext {\mathbf{B}_{\mathrm{ext}}}
\newcommand \Bextix[1] {\mathbf{B}_{\mathrm{ext};{#1}}}
\newcommand \Jext {\mathbf{J}_{\mathrm{ext}}} 
\newcommand \Jextix[1] {\mathbf{J}_{\mathrm{ext};{#1}}}

We have seen above that the total current density is not suitable as a variational parameter, at least not in the conventional variational principle. We here consider a modification of the conventional variational principle that also takes into account the energy of the induced magnetic field.

An external magnetic field induces an electric current density $-\mathbf{j}$ in a molecule (recall that the charge of an electron is $-e=-1$ in our units). In accordance with Biot--Savart's law, $\nabla\times\B_{\mathrm{ind}}(\rr) = -\mu_0 \mathbf{j}(\rr)$, this current density in turn induces an internal magnetic field. For a system with a non-degenerate ground state, there is no permanent current density, and in a weak uniform magnetic field $\Bext$, one therefore has
\begin{equation*}
  \Bind(\rr) = \boldsymbol{\sigma}(\rr) \, \Bext + \mathcal{O}(B^2),
\end{equation*}
where $\boldsymbol{\sigma}(\rr)$ is a dimensionless nuclear shielding tensor~\cite{HelgakerJaszunskiRuud}. Its value at the nuclear positions is important in nuclear magnetic resonance spectroscopy and it is sometimes, as with the nucleus-independent chemical shift method~\cite{SCHLEYER_JACS118_6317}, studied at other selected locations within a molecule. The eigenvalues of $\boldsymbol{\sigma}(\rr)$ are typically on the order of 100 ppm or $10^{-4}$. Hence, the induced field tends to be much weaker than the external field. Nonetheless, the induced field has an energy that is typically neglected in standard electronic structure theory, but is accounted for in the Maxwell--Schr\"odinger model of quantum electrons coupled self-consistently to a classical electromagnetic field.
Remarkably, taking into account the energy of the induced magnetic magnetic field in what will then be called Maxwell--Schr\"odinger DFT (or MDFT for short) has a substantial qualitative impact on current-density functional theory~\cite{TellgrenSolo2018}. It allows for a natural formulation using the total current density or, equivalently, its induced magnetic field. Moreover, the central functional turns out to be a version of the Grayce--Harris functional (see Section~\ref{sec:B-DFT}), which now appears as a universal functional, rather than the Vignale--Rasolt functional. In general, the magnetic field induced by the current density can be described by the vector potential
\begin{equation*}
  \Aind(\rr) = -\frac{\mu_0}{4\pi} \int \frac{\mathbf{j}(\rr')}{|\rr-\rr'|} \d\rr'.
\end{equation*}
The energy of the field is
\begin{equation*}
  \mathcal{E}_{\mathrm{ind}} = \frac{1}{2\mu_0} \int |\Bind(\rr)|^2 \d\rr = \frac{\mu_0}{8\pi} \int \frac{\mathbf{j}(\rr)\cdot\mathbf{j}(\rr')}{|\rr-\rr'|} \d\rr \d\rr'.
\end{equation*}
For simplicity, we now demand that both the external and the induced magnetic fields have finite energy, i.e., we take all magnetic fields to belong to the function space
\begin{equation*}
  L^2_{\mathrm{div}}(\mathbb{R}^3)= \{  \mathbf{u} \in L^2(\mathbb{R}^3) \, | \, \nabla\cdot\mathbf{u} = 0 \}.
\end{equation*}
We require vector potentials to satisfy $\nabla\times\A \in L^2_{\mathrm{div}}(\mathbb{R}^3)$.
The Maxwell--Schr\"odinger energy functional is 
\begin{equation*}
\begin{aligned}
  &E_{\mathrm{M}}[v,\Aext] \\
  &= \inf_{\Aind} \left( E[v,\Aext+\Aind] + \frac{1}{2\mu_0} \|\nabla\times\Aind\|_2^2 \right).
\end{aligned}
\end{equation*}
For the external field $\Bext$, we regard not only its vector potentials $\Aext$, but also the associated current density $-\mu_0 \Jext = \nabla\times\Bext$ as an alternate representation of $\Bext$.  For example, the ground-state energy $E[v,\Aext]$ can equally well be regarded as a functional of $\Bext$ or $\Jext$~\cite{TellgrenSolo2018}. Exploiting the gauge invariance of the ground-state energy functional $E[v,\Aext]$, we can now write the Maxwell--Schr\"odinger energy functional as
\begin{equation*}
  E_{\mathrm{M}}[v,\Bext] = \inf_{\Bind} \left( E[v,\Bext+\Bind] + \frac{1}{2\mu_0} \|\Bind\|_2^2 \right).
\end{equation*}
Here, at the outset, the induced magnetostatic field $\Bind$ is treated as an independent variational parameter, which does not necessarily satisfy Biot--Savart's law. However, this relation is satisfied by a minimizer since a form of Biot--Savart's law is just the stationarity condition for the above minimization~\cite{TellgrenSolo2018,Garrigue2020-magneticHK}.
The infimum in the above equation can just as well be taken over $\Btot = \Bext+\Bind$. Then one sees that $E_{\mathrm{M}}[v,\Bext]$ is the Moreau--Yosida regularization (already discussed in Section~\ref{sec:para-MY} for the density functional of paramagnetic CDFT) of the conventional energy $E[v,\Bext]$. This has the immediate consequence of imposing an upper limit on how diamagnetic a system can be in the sense that~\cite{TellgrenSolo2018}
\begin{equation*}
  E_{\mathrm{M}}[v,\Bext] \leq E[v,\mathbf{0}] + \frac{1}{2\mu_0} \langle \Bext, \Bext \rangle.
\end{equation*}
Moreover, expressing the energy $E[v,\Atot]$ in terms of the Grayce--Harris functional gives
\begin{align*}
  E_{\mathrm{M}}[v,\Aext] = \frac{\|\Bext\|_2^2}{2\mu_0} & + \inf_{\rho,\Atot} \left\{ \langle v, \rho \rangle - \frac{\langle \Bext, \Btot \rangle}{\mu_0} \right. \\
  &\left. + \frac{\|\Btot\|_2^2}{2\mu_0}  + G[\rho,\Atot]\right\}.
\end{align*}
or, exploiting gauge invariance,
\begin{align*}
  E_{\mathrm{M}}[v,\Bext] = \frac{\|\Bext\|_2^2}{2\mu_0} & + \inf_{\rho,\Btot} \left\{ \langle v, \rho \rangle - \frac{\langle \Bext, \Btot \rangle}{\mu_0} \right. \\
  &\left. + \frac{\| \Btot\|_2^2}{2\mu_0}  + G[\rho,\Btot]\right\}.
\end{align*}
From this expression it follows that $\bar{E}_{\mathrm{M}}[v,\Bext] = E_{\mathrm{M}}[v,\Bext] - \frac{1}{2\mu_0} \| \Bext\|_2^2$ is jointly concave, and it is, to within a reparametrization eliminating the factor $2\mu_0$, a Legendre--Fenchel transform of the shifted Grayce--Harris functional $\bar{G}[\rho,\Btot] = \frac{1}{2\mu_0} \|\Btot\|_2^2 + G[\rho,\Btot]$:
\begin{equation*}
\begin{aligned}
    &\bar{E}_{\mathrm{M}}[v,\Bext] \\
    &=  \inf_{\rho,\Btot} \left\{ \langle v, \rho \rangle - \frac{1}{\mu_0} \langle \Bext, \Btot \rangle + \bar{G}[\rho,\Btot]\right\}.
\end{aligned}
\end{equation*}

The Maxwell--Schr\"odinger ground-state energy can also be expressed in terms of the paramagnetic current density and the Vignale--Rasolt functional,
\begin{align*}
   E_{\mathrm{M}}[v,\Aext] =& \inf_{\rho,\Atot,\jpara} \bigg\{ \langle v+\tfrac{1}{2} |\Atot|^2, \rho\rangle + \langle \Atot,\jpara\rangle \\
   &\left. + \frac{\|\nabla\times\Atot-\Bext\|_2^2}{2\mu_0} + \FCSpure[\rho,\jpara] \right\}.
\end{align*}
When minimizers $\rho,\Atot,\jpara$ are available, we have
\begin{equation*}
 \begin{split}
  & \langle \tfrac{1}{2} |\Atot|^2, \rho\rangle + \langle \Atot,\jpara\rangle + \frac{\|\Btot-\Bext\|_2^2}{2\mu_0}
  \\
  & = \inf_{\Atot'} \left\{ \langle \tfrac{1}{2} |\Atot'|^2, \rho\rangle + \langle \Atot',\jpara\rangle + \frac{\|\Btot'-\Bext\|_2^2}{2\mu_0} \right\}
  \end{split}
\end{equation*}
and therefore also, with $\Atot' = \Atot + \epsilon \boldsymbol{\zeta}$,
\begin{equation*}
  \langle \boldsymbol{\zeta}, \jpara + \rho \Atot + \tfrac{1}{\mu_0} \nabla\times(\nabla\times\Atot-\Bext)\rangle \geq 0,
\end{equation*}
for all $\boldsymbol{\zeta}$. Hence, we define the total current density in the Maxwell--Schr\"odinger model to be
\begin{equation*}
  \mathbf{j} = \jpara + \rho \Atot = \jpara + \rho \Aext + \rho \Aind
\end{equation*}
and for minimizers we recover Biot--Savart's law, $-\mu_0 \mathbf{j} = \nabla\times(\Btot-\Bext)$, which is now a self-consistent condition where the induced field appears on both the left- and right-hand side.

The convex structure of the outlined theory automatically yields a type of HK1 result~\cite{TellgrenSolo2018}:

\begin{theorem}[HK1 in MDFT]
\label{thmHK1inMDFT}
Suppose that the pairs $(\Gamma_1,\Atotix{1})$ and $(\Gamma_2,\Atotix{2})$ are both Maxwell--Schr\"odinger ground states for $(v_1,\Aextix{1})$ and $(v_2,\Aextix{2})$, respectively. Suppose further that $\Gamma_1,\Gamma_2\mapsto \rho$ and that $\Atotix{1}$ and $\Atotix{2}$ yield the same magnetic field $\nabla\times\Atotix{1} = \nabla\times\Atotix{2} = \Btot$. Then $(\Gamma_1,\Atotix{1})$ is also a ground state for $(v_2,\Aextix{2})$ and vice versa.
\end{theorem}

\begin{proof}
Let us divide the proof into two cases, where in the first case we make an additional assumption. Case I: The total vector potentials are equal, $\Atotix{1} = \Atotix{2} = \Atot$, then
\begin{equation}
 \label{eqMDFT_HK1proof}
 \begin{split}
  &E_{\mathrm{M}}[v_1,\Aextix{1}] \\
  &= \langle v_1, \rho \rangle + \frac{\|\Bextix{1}-\Btot\|_2^2}{2\mu_0} +\trace( H[0,\Atot] \Gamma_1 )
  \\
  &\leq \langle v_1, \rho \rangle + \frac{\|\Bextix{1}-\Btot\|_2^2}{2\mu_0} + \trace( H[0,\Atot] \Gamma_2)
  \\
  &= E_{\mathrm{M}}[v_2,\Aextix{2}] + \langle v_1-v_2, \rho \rangle \\
  &\quad + \frac{\|\Bextix{1}-\Btot\|_2^2-\|\Bextix{2}-\Btot\|_2^2}{2\mu_0},
  \end{split}
\end{equation}
and the same holds if the indices 1 and 2 are exchanged.
Adding the two resulting inequalities yields
\begin{equation*}
 \begin{split}
  &E_{\mathrm{M}}[v_1,\Aextix{1}] + E_{\mathrm{M}}[v_2,\Aextix{2}] \\
   &\leq E_{\mathrm{M}}[v_1,\Aextix{1}]+ E_{\mathrm{M}}[v_2,\Aextix{2}].
  \end{split}
\end{equation*}
If Eq.~\eqref{eqMDFT_HK1proof} is a strict inequality for either of the index combinations, one would obtain the contradiction $E_{\mathrm{M}}[v_1,\Aextix{1}] + E_{\mathrm{M}}[v_2,\Aextix{2}] 
< E_{\mathrm{M}}[v_1,\Aextix{1}]+ E_{\mathrm{M}}[v_2,\Aextix{2}]$. Hence, Eq.~\eqref{eqMDFT_HK1proof} must hold with equality.

Case II: The vector potentials are not equal, $\Atotix{1} \neq \Atotix{2}$. Since the vector potentials share the same magnetic field, they at most differ by a gauge function, $\Atotix{2} = \Atotix{1} + \nabla \chi$. Defining
\begin{equation*}
  \begin{split}
    \Gamma'_{1}(\mathbf{r}_1,\ldots,\mathbf{r}_{N}; \mathbf{r}'_1,\ldots,\mathbf{r}'_{N})  =\; & \Gamma_{1}(\mathbf{r}_1,\ldots,\mathbf{r}_{N}; \mathbf{r}'_1,\ldots,\mathbf{r}'_{N})
    \\
    & \times \prod_{k=1}^{N} \rme^{\i\left(\chi(\mathbf{r}_k) - \chi(\mathbf{r}'_k)\right)}
  \end{split}
\end{equation*}
we note that, by gauge invariance, $\trace( H[0,\Atotix{1}] \Gamma_1) = \trace( H[0,\Atotix{2}] \Gamma'_1)$, so $(\Gamma'_1,\Atotix{2})$ is still a Maxwell--Schr\"odinger ground state for $(v_1,\Aextix{1})$. Since also $\Gamma'_1\mapsto\rho$, we can consider $(\Gamma'_1,\Atotix{2})$ and $(\Gamma_2,\Atotix{2})$ instead of $(\Gamma_1,\Atotix{1})$ and $(\Gamma_2,\Atotix{2})$. This reduces Case II to Case I.
\end{proof}

The fact that it is the total magnetic field that enters in the HK1 result has the surprising consequence that the current density required to generate the external field, $-\Jextix{i}(\rr) = \mu_0^{-1} \nabla\times\Bextix{i}(\rr)$, comes into play. Specifically, the shared current density relevant to the HK1 result is
\begin{equation*}
 -\frac{\nabla\times\Btot}{\mu_0} =  \Jextix{i} + \jpara_i + \rho \Atotix{i}.
\end{equation*}

\begin{theorem}[HK2 in MDFT]
  Suppose two different external potentials $(v_1,\Aextix{1})$ and $(v_2,\Aextix{2})$ share the same ground-state density $\rho$ and total magnetic field $\Btot$ with $\rho(\rr)>0$ (almost everywhere). Then (a) $v_1$ and $v_2$ are equal up to a constant, and (b) the external magnetic fields are equal, $\nabla\times\Aextix{1} = \nabla\times\Aextix{2}$.
\end{theorem}
\begin{proof}
  Part (a): By Theorem~\ref{thmHK1inMDFT}, there exists a shared ground state $\Gamma$ and vector potential $\Atot$ such that $(\Gamma,\Atot)$ is a ground state of both $H[v_1,\Atot]$ and $H[v_2,\Atot]$. That $v_1 = v_2 + \text{constant}$ now follows from the HK result in BDFT in Section~\ref{sec:B-DFT}.  \\
  Part (b): Biot--Savart's law yields
  \begin{equation*}
    \frac{\nabla\times\Btot}{\mu_0} =  -(\Jextix{i} + \jpara_{\Gamma} + \rho \Atot).
  \end{equation*}
  Since $\Jextix{i} = -\mu_0^{-1} \nabla\times\Bextix{i} = -\mu_0^{-1} \nabla\times\nabla\times\Aextix{i}$ is the only term that depends on $i$, it follows that
  \begin{equation*}
    \nabla\times\Bextix{i} = \nabla\times\Bextix{j}.
  \end{equation*}
Finally, under the condition $\Bextix{i} \in L^2_{\mathrm{div}}$, the curl is invertible. Hence, $\Bextix{i} = \nabla\times\Aextix{i} = \nabla\times\Aextix{j} = \Bextix{j}$.
\end{proof}

The convexity of the outlined theory as well as the above HK result are both results of the introduction of an internal magnetic field as an additional variational degree of freedom. While the vacuum magnetic permeability has an empirical value, $\mu_0 = 1.2566 \times 10^{-6}$~NA$^{-2}$, one could try to connect the above model and its HK results to the conventional Schr\"odinger model from before by considering the limit $\mu_0 \to 0^{+}$, though to our knowledge this has not yet been done. This might be one avenue for deriving a type of HK result for total current densities. Finally, we note the work by Garrigue~\cite{Garrigue2020-magneticHK} where the Maxwell--Schr\"odinger model is also analyzed and Theorem~2.7 of that work establishes a HK2 result involving the current density $\jpara + \rho \Aind$. Hence, the counterexamples that prevent a full HK2 result for the paramagnetic current density within the conventional Vignale--Rasolt CDFT formulation are circumvented in the Maxwell--Schr\"odinger model.

\section{Quantum-electrodynamical DFT}
\label{sec:QEDFT}

If we want to understand where the Schr\"odinger equations in their various forms encountered in this review come from, we can find the answer in the theory of QED. This theory arises from representing the energy-momentum relation of special relativity $E^2 = p^2c^2 + m^2 c^4$ in terms of first-order differential equations~\cite{ryder_1996,Greiner_1996}. If we do so for massive spin-$1/2$ particles we end up with the single-particle Dirac equation, while for massless spin-$1$ particles we arrive at the Riemann--Silberstein equations~\cite{silberstein1907,oppenheimer1931,bialynicki1994wave,gersten1999maxwell}. The Riemann--Silberstein equations are one of many equivalent ways to express the Maxwell equations in vacuum. The equations for matter and for light are coupled by making the local conservation of charges (charges are not destroyed but can only be moved around in space and time) explicit~\cite{Greiner_1996,ryder_1996}. This leads to the ``minimal-coupling prescription'', which is commonly expressed by the simple rule to replace the momentum operator $-\i\nabla$ by $-\i \nabla + \A$. The first thing that is problematic in these equations, however, is that since they are first order, they allow for negative-energy solutions which are nonphysical. One therefore performs a ``second quantization step'', where the equations are expressed in terms of field operators for light as well as for charged particles, and the negative energy solutions are assigned a positive value and interpreted physically as the corresponding anti-particles~\cite{Greiner_1996,ryder_1996}. The resulting quantum field theory is, however, mathematically notoriously badly behaved~\cite{baez2014introduction,scharf2014finite}, since it rests on the ill-defined concept of multiplying distribution-valued operators~\cite{thirring2013quantum}. This is the origin of the regularization and renormalization issues in quantum field theories~\cite{Greiner_1996,ryder_1996}. A second problem is encountered for the quantized light field, where in general we have four polarization directions, while physically only two transverse polarizations exist. This problem arises due to the gauge freedom of the Maxwell equation in vector-potential formulation, and in general implies quite intricate technical solutions~\cite{keller2012quantum,Greiner_1996}. However, if we decide to work in Coulomb gauge, i.e.,
\begin{equation*}
    \nabla \cdot \A = 0,
\end{equation*}
then in vacuum it holds $-\nabla^2 \phi = 0$. This implies that the zero component of the electromagnetic vector potential is  $\phi= 0$. Thus only the two physical transverse degrees of freedom of the light field are left that need to be second quantized. Yet, upon coupling to the charged matter degrees of freedom, the Coulomb gauge condition implies that the total longitudinal electric and interaction energy that arises from the charged particles is expressed directly by the expectation value of~\cite{Greiner_1996} 
\begin{equation}\label{eq:CoulombMaxwell}
\begin{aligned}
   W_{\rm C} &= \frac{1}{2}\sum_{k \neq l}^{N} \sum_{\bn \in \Z^3} \frac{4 \pi}{\kn^2} \frac{\rme^{\i \kn \cdot(\rr_k-\rr_l)}}{L^3}
   \\
   &\overset{L\rightarrow \infty}{\longrightarrow} \frac{1}{2} \sum_{k \neq l}^{N} \frac{1}{|\rr_k - \rr_l|}
\end{aligned}
\end{equation}
in atomic units. For simplicity, we have here assumed a finite but arbitrarily large quantization volume $L^3$ with periodic boundary conditions which implies a Fourier expansion with the wave vector $\kn = \tfrac{2 \pi}{L} \bn$. For $L\rightarrow \infty$ the longitudinal Maxwell energy becomes the usual Coulomb interaction.
In just the same way the external scalar potential $v$, that acts as the binding potential for the system, arises from the coupling to electrons and to external charges like nuclei.
Thus we see that by including the Coulomb interaction and the external scalar potential, which was already present in Eq.~\eqref{eq:CDFT-Ham-general}, we have taken into account the full longitudinal Maxwell energy together with the backreaction of matter on the longitudinal light field. Consequently, for only scalar external potentials, we have also already taken into account the corresponding (purely longitudinal) photon-field energy. Considering the issues that we encountered throughout this review when trying to establish a HK2 result for CDFT, a simple physical explanation is at hand: We also need to take into account the energy contribution of the transverse photon field (induced magnetic field). Indeed, Section~\ref{sec-MDFT} highlights that this idea works and a self-consistent treatment of light and matter allows to establish also a HK2 result in the context of CDFT. Let us see whether we can also find a similar HK2 result if we keep both, light and matter fully quantized in the next step.

Note that we have assumed first-quantized charged particles in Eq.~\eqref{eq:CoulombMaxwell}, i.e., no electron-positron pair creation/annihilation is possible and the number of particles (electrons) is therefore conserved~\cite{spohn2004dynamics}. We have thus avoided one potential regularization/renormalization problem of fully relativistic QED~\cite{Greiner_1996,ryder_1996}. For the quantized field modes we have then
\begin{equation}\label{eq:quantizedvectorpot}
\begin{aligned}
\Ah(\rr) = \sqrt{\frac{4 \pi}{L^3}} \sum_{\bn \in \Z^3}\sum_{\lambda=1}^{2} \frac{\be(\bn,\lambda)}{\sqrt{2 \omega_\bn}}\left(\ah_{\bn,\lambda} \rme^{\i \kn \cdot \rr} +\ah^{\dagger}_{\bn,\lambda} \rme^{-\i \kn\cdot \rr}  \right),
\end{aligned}
\end{equation}
where $\be(\bn,\lambda)$ are two orthogonal transverse (with respect to $\bn$) polarization vectors and $\omega_\bn = c |\kn|$. If we assume $L\rightarrow \infty$, the sum in Eq.~\eqref{eq:quantizedvectorpot} becomes an integral and the creation $\ah^{\dagger}_{\bn,\lambda}$ and annihilation operators $\ah_{\bn,\lambda}$ that obey
\begin{equation*}
\left[\ah_{\bn',\lambda'},\ah_{\bn,\lambda}^{\dagger}\right] = \delta_{\bn,\bn'}\delta_{\lambda', \lambda},
\end{equation*}
turn into genuine field operators~\cite{thirring2013quantum}. For notational simplicity and to avoid further discussions about the properties of these field operators we keep a finite but arbitrarily large volume. The free quantized electromagnetic Hamiltonian is then simply
\begin{equation*}
H_{\rm ph} = \sum_{\bn ,\lambda} \omega_{\bn} \ah^{\dagger}_{\bn,\lambda} \ah_{\bn,\lambda},
\end{equation*}
and the coupling of the free photon field to a classical transverse external \textit{charge} current 
\begin{equation}\label{eq:externalcurrentexpansion}
\begin{aligned}
\Jext(\rr) = \sqrt{\frac{1}{4 \pi L^3}} \sum_{\bn ,\lambda} \sqrt{\frac{\omega_{\bn}^3}{2}} \be(\bn,\lambda) (J_{\bn,\lambda} \rme^{\i \kn \cdot \rr} + J_{\bn,\lambda}^{*}\rme^{-\i \kn\cdot \rr})
\end{aligned}
\end{equation}
is then
\begin{equation*}
\int \Jext(\rr) \cdot \Ah(\rr) \d \rr = \sum_{\bn ,\lambda} \omega_{\bn} \left(J^{*}_{\bn,\lambda} \ah_{\bn,\lambda} + J_{\bn,\lambda} \ah_{\bn,\lambda}^\dagger \right).
\end{equation*}
The fully coupled Pauli--Fierz Hamiltonian is then~\cite{spohn2004dynamics}
\begin{equation}\label{eq:PauliFierzHamiltonian}
\begin{aligned}
H &= \sum_{k=1}^{N}\left\{ \frac{1}{2 m}\left[\boldsymbol{\sigma}_{k}\cdot\left(-\i \nabla_k + \Ah(\rr_k)  \right)  \right]^2 + v(\rr_k)\right\} \\ 
& + W_{\rm C} + \sum_{\bn, \lambda} \omega_{\bn} \left( \ah^{\dagger}_{\bn,\lambda} \ah_{\bn,\lambda} - J^{*}_{\bn,\lambda} \ah_{\bn,\lambda} - J_{\bn,\lambda} \ah_{\bn,\lambda}^\dagger \right),
\end{aligned}
\end{equation}
where $\boldsymbol{\sigma}_k$ is the standard vector of Pauli spin matrices. We note first that the expectation value of the operator $\Ah(\rr)$ corresponds to the induced transverse field, i.e., if we compare to the previous section, it is the induced magnetic field. Yet instead of denoting the internal field $\mathbf{A}(\rr)$ with a subindex as done before, we here denote external magnetic fields with $\mathbf{A}_{\rm ext}(\rr)$. We further note here that the coupling to any external transverse vector potential $\mathbf{A}_{\rm ext}(\rr)$ can be taken into account by merely a coherent shift (vacuum polarization) of the photon modes. This means $\ah_{\bn,\lambda} \rightarrow \ah_{\bn,\lambda} - A^{\rm ext}_{\bn,\lambda}$ and accordingly for the creation operator, where the $A^{\rm ext}_{\bn,\lambda}$ are the Fourier expansion coefficients of the vector potential~\cite{Ruggenthaler2015}. That is, in Eq.~\eqref{eq:quantizedvectorpot} we get $\Ah(\rr) \rightarrow \Ah(\rr) + \mathbf{A}_{\rm ext}(\rr)$ upon such a coherent shift. This implies that we can represent any external magnetic field acting on the electronic system by taking the corresponding external transverse charge current that generates this field via the static Maxwell equation
\begin{equation*}
-\nabla^2 \mathbf{A}_{\rm ext}(\rr) = \mu_0  \Jext(\rr).
\end{equation*}
This equivalence of external classical transverse currents and external classical transverse vector potentials, which correspond uniquely to an external classical magnetic field (as also discussed in Section~\ref{sec-MDFT}), is of significance for a density-functional reformulation of the Pauli--Fierz quantum theory. This further consequence of the gauge principle implies that there are two ways of generating the same physical equilibrium situation. We note that for a time-dependent situation this is no longer the case, since we then have different initial states and potentially different dynamics. Thus if we want to achieve a HK2 result we need to make a choice. In the following we will choose to describe all the physically different magnetic fields by external classical transverse currents. Thus we have two classical external fields that we can adapt to generate physically different situations, the usual external classical scalar potential $v(\rr)$ of standard electronic DFT and the external classical transverse charge currents $\Jext(\rr)$, i.e., an external pair $(v,\Jext)$.

Before we come to the formulation of quantum-electrodynamical DFT (QEDFT), let us make some final yet important remarks with respect to the mathematical properties of the Pauli--Fierz Hamiltonian. Firstly, to have a well-defined self-adjoint Hamiltonian, one usually employs a form factor that regularizes how the modes couple in the ultraviolet regime~\cite{spohn2004dynamics}. The simplest version of this is to have an ultraviolet cutoff. We will therefore assume some highest momentum cutoff $\Lambda$ in $\kn$ in Eq.~\eqref{eq:PauliFierzHamiltonian}, which also implies that the allowed $\Jext(\rr)$ have a highest allowed momentum in the expansion of Eq.~\eqref{eq:externalcurrentexpansion}. Also, depending on the chosen cutoff $|\kn| \leq \Lambda$, one needs to use a bare mass for the electrons, since the observable mass $m=m_{\rm e} = 1$ (in atomic units) does contain already all the photon contributions. Now, with having the photon modes explicit, the free dispersion of the electron will change without modifying the observable mass to a (cutoff-dependent) bare mass of the electron. Thus in Eq.~\eqref{eq:PauliFierzHamiltonian} we have $m_{\rm e} \geq m=m(\Lambda)>0$~\cite{spohn2004dynamics,hainzl2002mass}. We note that we here assume a finite volume $L^3 \subset \mathbb{R}^3$ and hence for any scalar potential we will have a ground state by construction. Nevertheless, for the Pauli--Fierz Hamiltonian defined on all of $\mathbb{R}^3$ it can be shown that every scalar potential that has a ground state without coupling to the photon field also has a ground state with the coupling to the photon field~\cite{spohn2004dynamics}. This gives a nice consistency with standard electronic DFT and the question of $v$-representable ground states.

Let us next, following the structure proposed in this review, first define the HK1 for QEDFT. For this we re-express the Hamiltonian of Eq.~\eqref{eq:PauliFierzHamiltonian} in terms of
\begin{equation*}
H[v,\Jext] = H_0 + \sum_{k=1}^N v(\rr_k) -  \int \Jext(\rr) \cdot \Ah(\rr) \d \rr.
\end{equation*}
In this way the (ensemble) constrained search functional for QEDFT is then
\begin{equation*}
F_{\rm CS, ens}[\rho, \A] = \inf_{\Gamma \mapsto (\rho, \A)} \trace(H_0 \Gamma), 
\end{equation*}
such that 
\begin{equation*}
E[v,\Jext] = \inf_{\rho,\A} \{F_{\rm CS, ens}[\rho, \A] + \langle v, \rho \rangle - \langle \Jext,\A \rangle  \}.
\end{equation*}
This follows exactly the structure proposed in Section~X of Part~I of this review. With respect to previous examples, e.g., the Maxwell--Schr\"odinger DFT, we now have, however, density matrices that contain electronic and photonic degrees of freedom.

\pagebreak 

\begin{theorem}[HK1 in QEDFT]
\label{thmHK1inQEDFT}
Let $\Gamma_1$ be a ground state of $H[v_1,\mathbf{J}_{\rm ext,1}]$ and $\Gamma_2$ a ground state of $H[v_2,\mathbf{J}_{\rm ext,2}]$. If $\Gamma_1, \Gamma_2 \rightarrow (\rho,\A)$, i.e., if these states share the same density and vector potential, then $\Gamma_1$ is also a ground state of $H[v_2,\mathbf{J}_{\rm ext,2}]$ and $\Gamma_2$ is also a ground state of $H[v_1,\mathbf{J}_{\rm ext,1}]$. 
\end{theorem}

For the proof of this statement we refer to Theorem~1 of Part~I of this review. Let us next turn to the more important question of the HK2 in QEDFT. To do so we first note that the total (physical) \textit{charge} current density of the Pauli--Fierz Hamiltonian (also compare with Section~\ref{sec-MDFT}) is
\begin{equation*}
\bj_{\Gamma}(\rr)= -\frac{1}{m}\bj_{\Gamma}^{\rm p}(\rr) - \bj^{\rm m}_{\Gamma}(\rr) - \bj^{\rm dm} _{\Gamma}(\rr),
\end{equation*}
where $\bj^{\rm m}_{\Gamma} = \tfrac{1}{2m}\nabla \times \sum_{k=1}^{N} \trace(\Gamma \boldsymbol{\sigma}_k \delta(\rr-\rr_k) )$ is the magnetization current and $\bj^{\rm dm} _{\Gamma}(\rr) = \frac{1}{m} \sum_{k=1}^N \trace(\Gamma \delta(\rr-\rr_k) \Ah(\rr))$ the quantized diamagnetic current. By using the Heisenberg equation of motion for $\Ah(\rr)$ twice~\cite{Ruggenthaler2015} we find that any eigenstate of the coupled matter-photon system obeys the static inhomogeneous Maxwell equation in Coulomb gauge
\begin{equation}\label{eq:staticMaxwell}
- \nabla^2 \A_{\Gamma}(\rr) = \mu_0  \left(\bj_{\Gamma, \perp}(\rr) + \Jext(\rr)  \right), 
\end{equation}
where $\A_{\Gamma}(\rr) = \trace(\Gamma \Ah(\rr))$ and $\bj_{\Gamma, \perp}(\rr)$ is the transverse part of the total charge current. This allows us to show the following theorem.

\begin{theorem}[HK2 in QEDFT]
If two external pairs $(v_1,\mathbf{J}_{\rm ext,1})$ and $(v_2,\mathbf{J}_{\rm ext,2})$ share a common eigenstate $\psi$ and if $\psi$ is non-zero almost everywhere, then these two pairs are the same. The equivalent statement holds for density matrices.
\end{theorem}

\begin{proof}
First we note that if both Hamiltonians $H[v_1,\mathbf{J}_{\rm ext,1}]$ and $H[v_2,\mathbf{J}_{\rm ext,2}]$ share a common eigenstate then due to Eq.~\eqref{eq:staticMaxwell} we have $\mathbf{J}_{\rm ext,1}=\mathbf{J}_{\rm ext,2}$. Thus we are left with the two equations
\begin{align*}
\left(H_0 + V[v_1]\right)\psi &= E[v_1,\Jext]\psi,
\\
\left(H_0 + V[v_2]\right)\psi &= E[v_2,\Jext]\psi,
\end{align*}
and we can follow the standard HK2 proof of Theorem~2 of Part~I. We can further use Corollary~3 of Part~I to find the equivalent statement for the density matrices.
\end{proof}
At that point we see again how powerful the abstract formulation of HK1 and HK2 as presented at the end of Part~I and then repeated at the beginning of Part~II of this review is. It allows to re-use many results of standard electronic DFT for other settings as well. We finally note that for the corresponding KS system in QEDFT one commonly uses non-interacting electrons and photons, which leads to electronic Pauli--Kohn--Sham equations coupled to a static inhomogeneous Maxwell equation of the form of Eq.~\eqref{eq:staticMaxwell}~\cite{Ruggenthaler2015}.

\section{Summary}
\label{sec:summary}

\begin{table*}
 \begin{tabular}{l|cccc} 
    \hline \hline
    flavor & density variables & potential variables & HK1 & HK2 \\   \hline
    DFT & $\rho$ & $v$ & yes & yes \\
    LDFT & $\rho,\pp,\mathbf{L}$ & $v,\A=\a+\tfrac{1}{2}\B\times\rr$ & yes & yes, if $\rho$ asym. \\
    SDFT (col.) & $\rho, m_z = \rho_{\uparrow}-\rho_{\downarrow}$ & $v,B'_z$ & yes & debated \\
    SDFT (noncol.) & $\rho, \mathbf{m}$ & $v,\B'$ & yes & no \\
    CDFT (no spin) & $\rho,\jpara$ & $v,\A$ & yes & no \\
    CDFT (with spin curr.) & $\rho,\jmag = \jpara + \nabla\times\mathbf{m}$ & $v,\A$ & yes & no \\
    CDFT (with spin dens.) & $\rho,\mathbf{m},\jpara$ & $v,\A$ & yes & no \\
    CDFT (with spin dens.) & $\rho,\mathbf{m},\jpara$ & $v,\B',\A$ & yes & no \\
    MDFT & $\rho,\Btot$ & $v,\B$ & yes & yes \\
    QEDFT & $\rho,\boldsymbol{A}$ & $v,\Jext$ & yes & yes  \\
    \hline \hline
 \end{tabular}
 \caption{The status of HK1 and HK2 within some flavors of DFT.}\label{table:sum-dfts}
\end{table*}

Many flavors of density-functional theory exist besides standard DFT. All flavors considered here capture some aspect of spin and orbital magnetism.
 They can be characterized in terms of constraints on the, at the outset, very general Hamiltonian given in Eq.~\eqref{eq:CDFT-Ham-general}
\begin{equation*}
\begin{aligned}
  H[v,\B',\A] &= \frac{1}{2} \sum_{k=1}^N (-\i\nabla_k + \A(\rr_k))^2 + \sum_{k=1}^N \B'(\rr_k)\cdot\mathbf{S}_k\\
  &\quad + \sum_k v(\rr_k) + \lambda\sum_{k<l} \frac{1}{r_{kl}}.
\end{aligned}
\end{equation*}
As noted there, we allow for the case that the magnetic field appearing in the spin-Zeeman term is unrelated to the vector potential that affects the orbital degrees of freedom, i.e., $\B' \neq \nabla\times\A$.

Noncollinear SDFT is obtained when orbital effects are neglected by setting $\A=\mathbf{0}$. The external magnetic field $\B'$ is then paired with the spin density $\mathbf{m}$. Both of these are general, noncollinear vector fields.
Yet, most practical approximate functionals are constructed for the collinear case when only collinear magnetic fields $\B' = (0,0,B'_z)$ along, say, the $z$-axis are allowed. There is a global spin quantization axis and only the component $m_z = \rho_{\uparrow}-\rho_{\downarrow}$ of the vector field $\mathbf{m}$ is needed. The non-uniqueness of potentials (i.e., the lack of a HK2 result in our terminology) in collinear SDFT has been discussed by several authors, with different conclusions. The situation is summarized in \citeauthor{AYERS_JCP124_224108}~\cite{AYERS_JCP124_224108}.

In paramagnetic CDFT, the orbital effects are retained. Different flavors of CDFT are possible depending on how the spin-Zeeman term is treated. The simplest flavor, treated here in great detail, simply neglects it ($\B'=\mathbf{0}$). Alternatively, in the physically natural case where $\B' = \nabla\times\A$, a partial integration turns the spin-Zeeman term into an interaction between $\A$ and the spin current density. The latter is absorbed into the paramagnetic current density to form $\jmag = \jpara + \nabla\times\mathbf{m}$. Retaining $\B'$ as an independent variable, unrelated to $\nabla\times\A$, yields the most flexible setting with the triple $(\rho,\mathbf{m},\jpara)$ as the basic density variables. Loosely speaking, in a CDFT formulation analogous to Lieb's formulation of standard DFT, the triple of independent density variables must have a triple of independent potential variables. Hence, $\B'$ needs to be retained as an independent variable if $\mathbf{m}$ is to be an independent density. However, when Lieb-like formulation is not required, nothing prevents the introduction of additional constraints in a constrained-search formulation. In this sense, a CDFT formulation with a triple of density variables $(\rho,\mathbf{m},\jpara)$ and a pair potential variables $(v,\A)$, with $\B'=\nabla\times\A$, also exists.

With regard to the Hohenberg--Kohn theorem in CDFT, the inclusion or exclusion of spin effects makes no difference: HK1 holds and HK2 does not.
As already noted in Part~I, the HK1 result does not only hold for standard DFT, but holds for all variants of extended DFTs that offer the required structure. Paramagnetic CDFT has this structure and is arguably the most natural CDFT formulation as far as the mathematical results are concerned. At the same time, this theory is not invariant under gauge transformations and a HK2 cannot hold.
On the other hand, for the formulation of CDFT that uses the total (physical) current it is 
unfortunate that in general
\begin{equation*}
\begin{aligned}
  &\langle \psi | H[v,\A] | \psi \rangle \\
  &\neq \inf_{ \psi \mapsto (\rho,\jtot) }\{ \langle \psi | H_0 | \psi \rangle   \}
 + \langle \A,  \jtot \rangle + \langle  v - \tfrac 1 2 \vert \A \vert^2, \rho\rangle .  
\end{aligned}
\end{equation*}
Equality only holds for $\psi$ such that $\a[\psi; \jtot] = \A$ (where $\a[\psi;\jtot]$ was defined in Eq.~\eqref{eq:a-D}), then a HK1 result follows. As can be seen by Eq.~\eqref{eq:EjCDFT}, it is not evident how to obtain an HK1 result since  minimization of just $H_0$ over wave functions then leaves out $\psi$-dependent terms. Note that HK2 does not hold, since we know that a shared eigenstate of magnetic Hamiltonians does not imply that the potentials are equal (up to a gauge).  
Also note that, if one fixes $\a= \A$, then one effectively has a paramagnetic formulation of CDFT again. 

Furthermore, what could be stressed from the above discussion is that, regardless of the status of a full HK result, we have 
no {\it HK variational principle} in total CDFT~\cite{LaestadiusBenedicks2015}. 
Thus, even if the question of a HK result for the total current could be answered in the positive, the formulation would be restricted to $v$-representable densities and thereby excluding the usual approach of utilizing constrained-search functionals on $N$-representable densities. This has stopped the mathematical development of total CDFT as compared to the paramagnetic variant. 

We have seen that by going beyond the usual density-functional setting, when new density and corresponding potential variables are included, problems arise mostly with respect to HK2. This is compactly highlighted in Table~\ref{table:sum-dfts}.  
While the mathematical reasons have been discussed in detail in the preceding sections, there are often also simple physical reasons. This holds specifically in the context when magnetic fields are included and associated densities are considered. The non-uniqueness results, discussed in Section~\ref{sec:cond-count}, arise because the back-reaction of the current on the external field and the change in Maxwell energy is not taken into account. Doing so by also including the induced Maxwell field in a self-consistent manner, as discussed in Section~\ref{sec-MDFT}, avoids some of these issues and a HK2 theorem becomes available.
Hence, the density-functional theory based on the Maxwell--Schr\"odinger model (MDFT) is a type of total CDFT with a full HK result.
This intuitive result, however, raises the question why we do not need to include the Maxwell field energy also in the usual (standard) DFT of only scalar external potentials. The answer to this question is given in Section~\ref{sec:QEDFT} with the help of QED. We saw that the Coulomb interaction of the usual Schr\"odinger equation is actually taking the self-consistent longitudinal photon energy into account upon interaction with matter. Therefore it seems natural to also take the transverse photon energy into account. In the context of low-energy QED, where both contributions are considered self-consistently, we therefore again find a HK2 result.

This formal discussion has shown that promoting the Maxwell field to a quantized system allows to recover a DFT formulation that is very close to the original electronic DFT. And by approximating the Pauli--Fierz theory we obtain, in the mean-field coupling limit, the Maxwell--Schr\"odinger equation, and by discarding the transverse part of the Maxwell field altogether, we find standard electronic DFT. Yet, apart from this nice consistency and the simple form of a DFT, is there any other reason to consider QEDFT and the Pauli--Fierz theory? The answer is `yes' and lies in the emerging fields of polaritonic chemistry and materials science as well as ab initio QED~\cite{garcia2021manipulating,ruggenthaler2022understanding}. In these fields, photonic structures, such as optical cavities, change locally the vacuum modes that couple to the matter subsystem and hence present a novel control knob to influence chemical and material properties. There are by now many seminal experimental results that show that upon reaching strong matter-photon coupling in photonic structures, we can indeed modify and control such properties. Consequently, methods that can approximately solve the Pauli--Fierz field theory become increasingly important.


\section*{Acknowledgement}

We are indebted to our two referees for numerous comments and careful corrections that helped to greatly improve the paper.
EIT, MAC and AL thank the  Research Council of Norway (RCN) under CoE (Hylleraas Centre) Grant No.~262695, for AL and MAC also CCerror Grant No.~287906 and for EIT also ``Magnetic Chemistry'' Grant No.~287950, and MR acknowledges the Cluster of Excellence ``CUI: Advanced Imaging of Matter'' of the Deutsche Forschungsgemeinschaft (DFG), EXC 2056, project ID 390715994. AL and MAC were also supported by the ERC through StG REGAL under agreement No.~101041487.
The authors thank Centre for Advanced Studies (CAS) in Oslo, since this work includes insights gathered at the YoungCAS workshop ``Do Electron Current Densities Determine All There Is to Know?'', held July 9-13, 2018, in Oslo, Norway.

\pagebreak
\section*{Bibliography}

%


\begin{thebibliography}{76}%
\makeatletter
\providecommand \@ifxundefined [1]{%
 \@ifx{#1\undefined}
}%
\providecommand \@ifnum [1]{%
 \ifnum #1\expandafter \@firstoftwo
 \else \expandafter \@secondoftwo
 \fi
}%
\providecommand \@ifx [1]{%
 \ifx #1\expandafter \@firstoftwo
 \else \expandafter \@secondoftwo
 \fi
}%
\providecommand \natexlab [1]{#1}%
\providecommand \enquote  [1]{``#1''}%
\providecommand \bibnamefont  [1]{#1}%
\providecommand \bibfnamefont [1]{#1}%
\providecommand \citenamefont [1]{#1}%
\providecommand \href@noop [0]{\@secondoftwo}%
\providecommand \href [0]{\begingroup \@sanitize@url \@href}%
\providecommand \@href[1]{\@@startlink{#1}\@@href}%
\providecommand \@@href[1]{\endgroup#1\@@endlink}%
\providecommand \@sanitize@url [0]{\catcode `\\12\catcode `\$12\catcode
  `\&12\catcode `\#12\catcode `\^12\catcode `\_12\catcode `\%12\relax}%
\providecommand \@@startlink[1]{}%
\providecommand \@@endlink[0]{}%
\providecommand \url  [0]{\begingroup\@sanitize@url \@url }%
\providecommand \@url [1]{\endgroup\@href {#1}{\urlprefix }}%
\providecommand \urlprefix  [0]{URL }%
\providecommand \Eprint [0]{\href }%
\providecommand \doibase [0]{https://doi.org/}%
\providecommand \selectlanguage [0]{\@gobble}%
\providecommand \bibinfo  [0]{\@secondoftwo}%
\providecommand \bibfield  [0]{\@secondoftwo}%
\providecommand \translation [1]{[#1]}%
\providecommand \BibitemOpen [0]{}%
\providecommand \bibitemStop [0]{}%
\providecommand \bibitemNoStop [0]{.\EOS\space}%
\providecommand \EOS [0]{\spacefactor3000\relax}%
\providecommand \BibitemShut  [1]{\csname bibitem#1\endcsname}%
\let\auto@bib@innerbib\@empty
\bibitem [{\citenamefont {Vignale}\ and\ \citenamefont
  {Rasolt}(1987)}]{Vignale1987}%
  \BibitemOpen
  \bibfield  {author} {\bibinfo {author} {\bibfnamefont {G.}~\bibnamefont
  {Vignale}}\ and\ \bibinfo {author} {\bibfnamefont {M.}~\bibnamefont
  {Rasolt}},\ }\bibfield  {title} {\enquote {\bibinfo {title}
  {Density-functional theory in strong magnetic fields},}\ }\href
  {https://doi.org/10.1103/PhysRevLett.59.2360} {\bibfield  {journal} {\bibinfo
   {journal} {Phys. Rev. Lett.}\ }\textbf {\bibinfo {volume} {59}},\ \bibinfo
  {pages} {2360--2363} (\bibinfo {year} {1987})}\BibitemShut {NoStop}%
\bibitem [{\citenamefont {Vignale}\ and\ \citenamefont
  {Rasolt}(1988)}]{Vignale1988}%
  \BibitemOpen
  \bibfield  {author} {\bibinfo {author} {\bibfnamefont {G.}~\bibnamefont
  {Vignale}}\ and\ \bibinfo {author} {\bibfnamefont {M.}~\bibnamefont
  {Rasolt}},\ }\bibfield  {title} {\enquote {\bibinfo {title} {Current- and
  spin-density-functional theory for inhomogeneous electronic systems in strong
  magnetic fields},}\ }\href {https://doi.org/10.1103/PhysRevB.37.10685}
  {\bibfield  {journal} {\bibinfo  {journal} {Phys. Rev. B}\ }\textbf {\bibinfo
  {volume} {37}},\ \bibinfo {pages} {10685--10696} (\bibinfo {year}
  {1988})}\BibitemShut {NoStop}%
\bibitem [{\citenamefont {Vignale}, \citenamefont {Rasolt},\ and\ \citenamefont
  {Geldart}(1990)}]{vignale-rasolt-geldart1990}%
  \BibitemOpen
  \bibfield  {author} {\bibinfo {author} {\bibfnamefont {G.}~\bibnamefont
  {Vignale}}, \bibinfo {author} {\bibfnamefont {M.}~\bibnamefont {Rasolt}},\
  and\ \bibinfo {author} {\bibfnamefont {D.}~\bibnamefont {Geldart}},\
  }\bibfield  {title} {\enquote {\bibinfo {title} {Magnetic fields and density
  functional theory},}\ }\href {https://doi.org/10.1016/S0065-3276(08)60599-7}
  {\bibfield  {journal} {\bibinfo  {journal} {Adv. Quantum Chem.}\ }\textbf
  {\bibinfo {volume} {21}},\ \bibinfo {pages} {235--253} (\bibinfo {year}
  {1990})}\BibitemShut {NoStop}%
\bibitem [{\citenamefont {Diener}(1991)}]{Diener}%
  \BibitemOpen
  \bibfield  {author} {\bibinfo {author} {\bibfnamefont {G.}~\bibnamefont
  {Diener}},\ }\bibfield  {title} {\enquote {\bibinfo {title}
  {Current-density-functional theory for a nonrelativistic electron gas in a
  strong magnetic field},}\ }\href {https://doi.org/10.1088/0953-8984/3/47/014}
  {\bibfield  {journal} {\bibinfo  {journal} {J. Phys.: Condens. Matter}\
  }\textbf {\bibinfo {volume} {3}},\ \bibinfo {pages} {9417--9428} (\bibinfo
  {year} {1991})}\BibitemShut {NoStop}%
\bibitem [{\citenamefont {Capelle}\ and\ \citenamefont
  {Vignale}(2002)}]{Capelle2002}%
  \BibitemOpen
  \bibfield  {author} {\bibinfo {author} {\bibfnamefont {K.}~\bibnamefont
  {Capelle}}\ and\ \bibinfo {author} {\bibfnamefont {G.}~\bibnamefont
  {Vignale}},\ }\bibfield  {title} {\enquote {\bibinfo {title} {Nonuniqueness
  and derivative discontinuities in density-functional theories for
  current-carrying and superconducting systems},}\ }\href
  {https://doi.org/10.1103/PhysRevB.65.113106} {\bibfield  {journal} {\bibinfo
  {journal} {Phys. Rev. B}\ }\textbf {\bibinfo {volume} {65}},\ \bibinfo
  {pages} {113106} (\bibinfo {year} {2002})}\BibitemShut {NoStop}%
\bibitem [{\citenamefont {Penz}\ \emph {et~al.}(2023)\citenamefont {Penz},
  \citenamefont {Tellgren}, \citenamefont {Csirik}, \citenamefont
  {Ruggenthaler},\ and\ \citenamefont {Laestadius}}]{PartI}%
  \BibitemOpen
  \bibfield  {author} {\bibinfo {author} {\bibfnamefont {M.}~\bibnamefont
  {Penz}}, \bibinfo {author} {\bibfnamefont {E.~I.}\ \bibnamefont {Tellgren}},
  \bibinfo {author} {\bibfnamefont {M.~A.}\ \bibnamefont {Csirik}}, \bibinfo
  {author} {\bibfnamefont {M.}~\bibnamefont {Ruggenthaler}},\ and\ \bibinfo
  {author} {\bibfnamefont {A.}~\bibnamefont {Laestadius}},\ }\bibfield  {title}
  {\enquote {\bibinfo {title} {The structure of the density-potential mapping.
  {P}art {I}: Standard density-functional theory},}\ }\href@noop {} {\bibfield
  {journal} {\bibinfo  {journal} {ACS Phys. Chem. Au}\ } (\bibinfo {year}
  {2023})}\BibitemShut {NoStop}%
\bibitem [{\citenamefont {von Barth}(2004)}]{vonBarth2004basic}%
  \BibitemOpen
  \bibfield  {author} {\bibinfo {author} {\bibfnamefont {U.}~\bibnamefont {von
  Barth}},\ }\bibfield  {title} {\enquote {\bibinfo {title} {Basic
  density-functional theory—an overview},}\ }\href
  {https://doi.org/10.1238/Physica.Topical.109a00009} {\bibfield  {journal}
  {\bibinfo  {journal} {Phys. Scr.}\ }\textbf {\bibinfo {volume} {2004}},\
  \bibinfo {pages} {9} (\bibinfo {year} {2004})}\BibitemShut {NoStop}%
\bibitem [{\citenamefont {Burke}\ and\ \citenamefont
  {friends}(2007)}]{burke2007abc}%
  \BibitemOpen
  \bibfield  {author} {\bibinfo {author} {\bibfnamefont {K.}~\bibnamefont
  {Burke}}\ and\ \bibinfo {author} {\bibnamefont {friends}},\ }\href
  {https://dft.uci.edu/doc/g1.pdf} {\enquote {\bibinfo {title} {The {ABC} of
  {DFT}},}\ } (\bibinfo {year} {2007}),\ \bibinfo {note} {accessed
  2023-01-31}\BibitemShut {NoStop}%
\bibitem [{\citenamefont {Burke}(2012)}]{burke2012perspective}%
  \BibitemOpen
  \bibfield  {author} {\bibinfo {author} {\bibfnamefont {K.}~\bibnamefont
  {Burke}},\ }\bibfield  {title} {\enquote {\bibinfo {title} {Perspective on
  density functional theory},}\ }\href {https://doi.org/10.1063/1.4704546}
  {\bibfield  {journal} {\bibinfo  {journal} {J. Chem. Phys.}\ }\textbf
  {\bibinfo {volume} {136}},\ \bibinfo {pages} {150901} (\bibinfo {year}
  {2012})}\BibitemShut {NoStop}%
\bibitem [{\citenamefont {Dreizler}\ and\ \citenamefont
  {Gross}(2012)}]{dreizler2012-book}%
  \BibitemOpen
  \bibfield  {author} {\bibinfo {author} {\bibfnamefont {R.~M.}\ \bibnamefont
  {Dreizler}}\ and\ \bibinfo {author} {\bibfnamefont {E.~K.}\ \bibnamefont
  {Gross}},\ }\href@noop {} {\emph {\bibinfo {title} {Density Functional
  Theory: An Approach to the Quantum Many-body Problem}}}\ (\bibinfo
  {publisher} {Springer},\ \bibinfo {year} {2012})\BibitemShut {NoStop}%
\bibitem [{\citenamefont {Eschrig}(2003)}]{eschrig2003-book}%
  \BibitemOpen
  \bibfield  {author} {\bibinfo {author} {\bibfnamefont {H.}~\bibnamefont
  {Eschrig}},\ }\href@noop {} {\emph {\bibinfo {title} {The Fundamentals of
  Density Functional Theory}}},\ \bibinfo {edition} {2nd}\ ed.\ (\bibinfo
  {publisher} {Springer},\ \bibinfo {year} {2003})\BibitemShut {NoStop}%
\bibitem [{\citenamefont {Parr}\ and\ \citenamefont {Yang}(1989)}]{parr}%
  \BibitemOpen
  \bibfield  {author} {\bibinfo {author} {\bibfnamefont {R.}~\bibnamefont
  {Parr}}\ and\ \bibinfo {author} {\bibfnamefont {W.}~\bibnamefont {Yang}},\
  }\href@noop {} {\emph {\bibinfo {title} {Density Functional Theory of Atoms
  and Molecules}}}\ (\bibinfo  {publisher} {Oxford University Press},\ \bibinfo
  {year} {1989})\BibitemShut {NoStop}%
\bibitem [{\citenamefont {Teale}\ \emph {et~al.}(2022)\citenamefont {Teale},
  \citenamefont {Helgaker}, \citenamefont {Savin}, \citenamefont {Adamo},
  \citenamefont {Aradi}, \citenamefont {Arbuznikov}, \citenamefont {Ayers},
  \citenamefont {Baerends}, \citenamefont {Barone},\ and\ \citenamefont
  {Calaminici}}]{teale2022round-table}%
  \BibitemOpen
  \bibfield  {author} {\bibinfo {author} {\bibfnamefont {A.~M.}\ \bibnamefont
  {Teale}}, \bibinfo {author} {\bibfnamefont {T.}~\bibnamefont {Helgaker}},
  \bibinfo {author} {\bibfnamefont {A.}~\bibnamefont {Savin}}, \bibinfo
  {author} {\bibfnamefont {C.}~\bibnamefont {Adamo}}, \bibinfo {author}
  {\bibfnamefont {B.}~\bibnamefont {Aradi}}, \bibinfo {author} {\bibfnamefont
  {A.~V.}\ \bibnamefont {Arbuznikov}}, \bibinfo {author} {\bibfnamefont
  {P.~W.}\ \bibnamefont {Ayers}}, \bibinfo {author} {\bibfnamefont {E.~J.}\
  \bibnamefont {Baerends}}, \bibinfo {author} {\bibfnamefont {V.}~\bibnamefont
  {Barone}},\ and\ \bibinfo {author} {\bibfnamefont {P.}~\bibnamefont
  {Calaminici}},\ }\bibfield  {title} {\enquote {\bibinfo {title} {{DFT}
  exchange: Sharing perspectives on the workhorse of quantum chemistry and
  materials science},}\ }\href {https://doi.org/10.1039/D2CP02827A} {\bibfield
  {journal} {\bibinfo  {journal} {Phys. Chem. Chem. Phys.}\ }\textbf {\bibinfo
  {volume} {24}},\ \bibinfo {pages} {28700--28781} (\bibinfo {year}
  {2022})}\BibitemShut {NoStop}%
\bibitem [{\citenamefont {Lai}(2001)}]{LAI_RMP73_629}%
  \BibitemOpen
  \bibfield  {author} {\bibinfo {author} {\bibfnamefont {D.}~\bibnamefont
  {Lai}},\ }\bibfield  {title} {\enquote {\bibinfo {title} {Matter in strong
  magnetic fields},}\ }\href {https://doi.org/10.1103/RevModPhys.73.629}
  {\bibfield  {journal} {\bibinfo  {journal} {Rev. Mod. Phys.}\ }\textbf
  {\bibinfo {volume} {73}},\ \bibinfo {pages} {629--662} (\bibinfo {year}
  {2001})}\BibitemShut {NoStop}%
\bibitem [{\citenamefont {Helgaker}, \citenamefont {Jaszunski},\ and\
  \citenamefont {Ruud}(1999)}]{HelgakerJaszunskiRuud}%
  \BibitemOpen
  \bibfield  {author} {\bibinfo {author} {\bibfnamefont {T.}~\bibnamefont
  {Helgaker}}, \bibinfo {author} {\bibfnamefont {M.}~\bibnamefont
  {Jaszunski}},\ and\ \bibinfo {author} {\bibfnamefont {K.}~\bibnamefont
  {Ruud}},\ }\bibfield  {title} {\enquote {\bibinfo {title} {Ab initio methods
  for the calculation of {NMR} shielding and indirect spinminus signspin
  coupling constants},}\ }\href {https://doi.org/10.1021/cr960017t} {\bibfield
  {journal} {\bibinfo  {journal} {Chem. Rev.}\ }\textbf {\bibinfo {volume}
  {99}},\ \bibinfo {pages} {293--352} (\bibinfo {year} {1999})}\BibitemShut
  {NoStop}%
\bibitem [{\citenamefont {Gomes}\ and\ \citenamefont
  {Mallion}(2001)}]{GOMES_CR101_1349}%
  \BibitemOpen
  \bibfield  {author} {\bibinfo {author} {\bibfnamefont {J.~A. N.~F.}\
  \bibnamefont {Gomes}}\ and\ \bibinfo {author} {\bibfnamefont {R.~B.}\
  \bibnamefont {Mallion}},\ }\bibfield  {title} {\enquote {\bibinfo {title}
  {Aromaticity and ring currents},}\ }\href {https://doi.org/10.1021/cr990323h}
  {\bibfield  {journal} {\bibinfo  {journal} {Chem. Rev.}\ }\textbf {\bibinfo
  {volume} {101}},\ \bibinfo {pages} {1349--1384} (\bibinfo {year}
  {2001})}\BibitemShut {NoStop}%
\bibitem [{\citenamefont {Vaara}\ and\ \citenamefont
  {Pyykk{\"o}}(2001)}]{VAARA_PRL86_3268}%
  \BibitemOpen
  \bibfield  {author} {\bibinfo {author} {\bibfnamefont {J.}~\bibnamefont
  {Vaara}}\ and\ \bibinfo {author} {\bibfnamefont {P.}~\bibnamefont
  {Pyykk{\"o}}},\ }\bibfield  {title} {\enquote {\bibinfo {title}
  {Magnetic-field-induced quadrupole splitting in gaseous and liquid {Xe-131}
  {NMR}: Quadratic and quartic field dependence},}\ }\href
  {https://doi.org/10.1103/PhysRevLett.86.3268} {\bibfield  {journal} {\bibinfo
   {journal} {Phys. Rev. Lett.}\ }\textbf {\bibinfo {volume} {86}},\ \bibinfo
  {pages} {3268--3271} (\bibinfo {year} {2001})}\BibitemShut {NoStop}%
\bibitem [{\citenamefont {Pagola}\ \emph {et~al.}(2005)\citenamefont {Pagola},
  \citenamefont {Pelloni}, \citenamefont {Caputo}, \citenamefont {Ferraro},\
  and\ \citenamefont {Lazzeretti}}]{PAGOLA_PRA72_033401}%
  \BibitemOpen
  \bibfield  {author} {\bibinfo {author} {\bibfnamefont {G.~I.}\ \bibnamefont
  {Pagola}}, \bibinfo {author} {\bibfnamefont {S.}~\bibnamefont {Pelloni}},
  \bibinfo {author} {\bibfnamefont {M.~C.}\ \bibnamefont {Caputo}}, \bibinfo
  {author} {\bibfnamefont {M.~B.}\ \bibnamefont {Ferraro}},\ and\ \bibinfo
  {author} {\bibfnamefont {P.}~\bibnamefont {Lazzeretti}},\ }\bibfield  {title}
  {\enquote {\bibinfo {title} {Fourth-rank hypermagnetizability of medium-size
  planar conjugated molecules and fullerene},}\ }\href
  {https://doi.org/10.1103/PhysRevA.72.033401} {\bibfield  {journal} {\bibinfo
  {journal} {Phys. Rev. A}\ }\textbf {\bibinfo {volume} {72}},\ \bibinfo
  {pages} {033401} (\bibinfo {year} {2005})}\BibitemShut {NoStop}%
\bibitem [{\citenamefont {Caputo}\ \emph {et~al.}(2007)\citenamefont {Caputo},
  \citenamefont {Ferraro}, \citenamefont {Pagola},\ and\ \citenamefont
  {Lazzeretti}}]{CAPUTO_JCP126_154103}%
  \BibitemOpen
  \bibfield  {author} {\bibinfo {author} {\bibfnamefont {M.~C.}\ \bibnamefont
  {Caputo}}, \bibinfo {author} {\bibfnamefont {M.~B.}\ \bibnamefont {Ferraro}},
  \bibinfo {author} {\bibfnamefont {G.~I.}\ \bibnamefont {Pagola}},\ and\
  \bibinfo {author} {\bibfnamefont {P.}~\bibnamefont {Lazzeretti}},\ }\bibfield
   {title} {\enquote {\bibinfo {title} {Calculation of the electric
  hypershielding at the nuclei of molecules in a strong magnetic field},}\
  }\href {https://doi.org/10.1063/1.2716666} {\bibfield  {journal} {\bibinfo
  {journal} {J. Chem. Phys.}\ }\textbf {\bibinfo {volume} {126}},\ \bibinfo
  {eid} {154103} (\bibinfo {year} {2007})}\BibitemShut {NoStop}%
\bibitem [{\citenamefont {Tellgren}\ \emph {et~al.}(2012)\citenamefont
  {Tellgren}, \citenamefont {Kvaal}, \citenamefont {Sagvolden}, \citenamefont
  {Ekstr\"om}, \citenamefont {Teale},\ and\ \citenamefont
  {Helgaker}}]{Tellgren2012}%
  \BibitemOpen
  \bibfield  {author} {\bibinfo {author} {\bibfnamefont {E.~I.}\ \bibnamefont
  {Tellgren}}, \bibinfo {author} {\bibfnamefont {S.}~\bibnamefont {Kvaal}},
  \bibinfo {author} {\bibfnamefont {E.}~\bibnamefont {Sagvolden}}, \bibinfo
  {author} {\bibfnamefont {U.}~\bibnamefont {Ekstr\"om}}, \bibinfo {author}
  {\bibfnamefont {A.~M.}\ \bibnamefont {Teale}},\ and\ \bibinfo {author}
  {\bibfnamefont {T.}~\bibnamefont {Helgaker}},\ }\bibfield  {title} {\enquote
  {\bibinfo {title} {Choice of basic variables in current-density-functional
  theory},}\ }\href {https://doi.org/10.1103/PhysRevA.86.062506} {\bibfield
  {journal} {\bibinfo  {journal} {Phys. Rev. A}\ }\textbf {\bibinfo {volume}
  {86}},\ \bibinfo {pages} {062506} (\bibinfo {year} {2012})}\BibitemShut
  {NoStop}%
\bibitem [{\citenamefont {Laestadius}\ and\ \citenamefont
  {Benedicks}(2014)}]{LaestadiusBenedicks2014}%
  \BibitemOpen
  \bibfield  {author} {\bibinfo {author} {\bibfnamefont {A.}~\bibnamefont
  {Laestadius}}\ and\ \bibinfo {author} {\bibfnamefont {M.}~\bibnamefont
  {Benedicks}},\ }\bibfield  {title} {\enquote {\bibinfo {title}
  {{H}ohenberg--{K}ohn theorems in the presence of magnetic field},}\ }\href
  {https://doi.org/10.1002/qua.24668} {\bibfield  {journal} {\bibinfo
  {journal} {Int. J. Quantum Chem.}\ }\textbf {\bibinfo {volume} {114}},\
  \bibinfo {pages} {782--795} (\bibinfo {year} {2014})}\BibitemShut {NoStop}%
\bibitem [{\citenamefont {Grayce}\ and\ \citenamefont {Harris}(1994)}]{GRAYCE}%
  \BibitemOpen
  \bibfield  {author} {\bibinfo {author} {\bibfnamefont {C.~J.}\ \bibnamefont
  {Grayce}}\ and\ \bibinfo {author} {\bibfnamefont {R.~A.}\ \bibnamefont
  {Harris}},\ }\bibfield  {title} {\enquote {\bibinfo {title} {Magnetic-field
  density-functional theory},}\ }\href
  {https://doi.org/10.1103/PhysRevA.50.3089} {\bibfield  {journal} {\bibinfo
  {journal} {Phys. Rev. A}\ }\textbf {\bibinfo {volume} {50}},\ \bibinfo
  {pages} {3089--3095} (\bibinfo {year} {1994})}\BibitemShut {NoStop}%
\bibitem [{\citenamefont {Laestadius}, \citenamefont {Penz},\ and\
  \citenamefont {Tellgren}(2021)}]{Laestadius2021}%
  \BibitemOpen
  \bibfield  {author} {\bibinfo {author} {\bibfnamefont {A.}~\bibnamefont
  {Laestadius}}, \bibinfo {author} {\bibfnamefont {M.}~\bibnamefont {Penz}},\
  and\ \bibinfo {author} {\bibfnamefont {E.}~\bibnamefont {Tellgren}},\
  }\bibfield  {title} {\enquote {\bibinfo {title} {Revisiting
  density-functional theory of the total current density},}\ }\href
  {https://doi.org/10.1088/1361-648X/abf784} {\bibfield  {journal} {\bibinfo
  {journal} {J. Phys.: Cond. Matter}\ }\textbf {\bibinfo {volume} {33}},\
  \bibinfo {pages} {295504} (\bibinfo {year} {2021})}\BibitemShut {NoStop}%
\bibitem [{\citenamefont {Laestadius}\ and\ \citenamefont
  {Benedicks}(2015)}]{LaestadiusBenedicks2015}%
  \BibitemOpen
  \bibfield  {author} {\bibinfo {author} {\bibfnamefont {A.}~\bibnamefont
  {Laestadius}}\ and\ \bibinfo {author} {\bibfnamefont {M.}~\bibnamefont
  {Benedicks}},\ }\bibfield  {title} {\enquote {\bibinfo {title} {Nonexistence
  of a hohenberg-kohn variational principle in total current-density-functional
  theory},}\ }\href {https://doi.org/10.1103/PhysRevA.91.032508} {\bibfield
  {journal} {\bibinfo  {journal} {Phys. Rev. A}\ }\textbf {\bibinfo {volume}
  {91}},\ \bibinfo {pages} {032508} (\bibinfo {year} {2015})}\BibitemShut
  {NoStop}%
\bibitem [{\citenamefont {Tellgren}(2018)}]{TellgrenSolo2018}%
  \BibitemOpen
  \bibfield  {author} {\bibinfo {author} {\bibfnamefont {E.~I.}\ \bibnamefont
  {Tellgren}},\ }\bibfield  {title} {\enquote {\bibinfo {title}
  {Density-functional theory for internal magnetic fields},}\ }\href
  {https://doi.org/10.1103/PhysRevA.97.012504} {\bibfield  {journal} {\bibinfo
  {journal} {Phys. Rev. A}\ }\textbf {\bibinfo {volume} {97}},\ \bibinfo
  {pages} {012504} (\bibinfo {year} {2018})}\BibitemShut {NoStop}%
\bibitem [{\citenamefont {Ayers}\ and\ \citenamefont
  {Yang}(2006)}]{AYERS_JCP124_224108}%
  \BibitemOpen
  \bibfield  {author} {\bibinfo {author} {\bibfnamefont {P.~W.}\ \bibnamefont
  {Ayers}}\ and\ \bibinfo {author} {\bibfnamefont {W.}~\bibnamefont {Yang}},\
  }\bibfield  {title} {\enquote {\bibinfo {title} {Legendre-transform
  functionals for spin-density-functional theory},}\ }\href
  {https://doi.org/10.1063/1.2200884} {\bibfield  {journal} {\bibinfo
  {journal} {J. Chem. Phys.}\ }\textbf {\bibinfo {volume} {124}},\ \bibinfo
  {eid} {224108} (\bibinfo {year} {2006})}\BibitemShut {NoStop}%
\bibitem [{\citenamefont {Capelle}\ and\ \citenamefont
  {Gross}(1997)}]{CAPELLE_PRL78_1872}%
  \BibitemOpen
  \bibfield  {author} {\bibinfo {author} {\bibfnamefont {K.}~\bibnamefont
  {Capelle}}\ and\ \bibinfo {author} {\bibfnamefont {E.~K.~U.}\ \bibnamefont
  {Gross}},\ }\bibfield  {title} {\enquote {\bibinfo {title} {Spin-density
  functionals from current-density functional theory and vice versa: {A} road
  towards new approximations},}\ }\href
  {https://doi.org/10.1103/PhysRevLett.78.1872} {\bibfield  {journal} {\bibinfo
   {journal} {Phys. Rev. Lett.}\ }\textbf {\bibinfo {volume} {78}},\ \bibinfo
  {pages} {1872--1875} (\bibinfo {year} {1997})}\BibitemShut {NoStop}%
\bibitem [{\citenamefont {Eschrig}\ and\ \citenamefont
  {Servedio}(1999)}]{ESCHRIG_JCC20_23}%
  \BibitemOpen
  \bibfield  {author} {\bibinfo {author} {\bibfnamefont {H.}~\bibnamefont
  {Eschrig}}\ and\ \bibinfo {author} {\bibfnamefont {V.~D.~P.}\ \bibnamefont
  {Servedio}},\ }\bibfield  {title} {\enquote {\bibinfo {title} {Relativistic
  density functional approach to open shells},}\ }\href
  {https://doi.org/10.1002/(SICI)1096-987X(19990115)20:1<23::AID-JCC5>3.0.CO;2-N}
  {\bibfield  {journal} {\bibinfo  {journal} {J. Comput. Chem.}\ }\textbf
  {\bibinfo {volume} {20}},\ \bibinfo {pages} {23--30} (\bibinfo {year}
  {1999})}\BibitemShut {NoStop}%
\bibitem [{\citenamefont {Gontier}(2013)}]{GONTIER_PRL111_153001}%
  \BibitemOpen
  \bibfield  {author} {\bibinfo {author} {\bibfnamefont {D.}~\bibnamefont
  {Gontier}},\ }\bibfield  {title} {\enquote {\bibinfo {title}
  {$n$-representability in noncollinear spin-polarized density-functional
  theory},}\ }\href {https://doi.org/10.1103/PhysRevLett.111.153001} {\bibfield
   {journal} {\bibinfo  {journal} {Phys. Rev. Lett.}\ }\textbf {\bibinfo
  {volume} {111}},\ \bibinfo {pages} {153001} (\bibinfo {year}
  {2013})}\BibitemShut {NoStop}%
\bibitem [{\citenamefont {Eich}\ and\ \citenamefont
  {Gross}(2013)}]{EICH_PRL111_156401}%
  \BibitemOpen
  \bibfield  {author} {\bibinfo {author} {\bibfnamefont {F.~G.}\ \bibnamefont
  {Eich}}\ and\ \bibinfo {author} {\bibfnamefont {E.~K.~U.}\ \bibnamefont
  {Gross}},\ }\bibfield  {title} {\enquote {\bibinfo {title} {Transverse
  spin-gradient functional for noncollinear spin-density-functional theory},}\
  }\href {https://doi.org/10.1103/PhysRevLett.111.156401} {\bibfield  {journal}
  {\bibinfo  {journal} {Phys. Rev. Lett.}\ }\textbf {\bibinfo {volume} {111}},\
  \bibinfo {pages} {156401} (\bibinfo {year} {2013})}\BibitemShut {NoStop}%
\bibitem [{\citenamefont {Laestadius}\ \emph {et~al.}(2019)\citenamefont
  {Laestadius}, \citenamefont {Penz}, \citenamefont {Tellgren}, \citenamefont
  {Ruggenthaler}, \citenamefont {Kvaal},\ and\ \citenamefont
  {Helgaker}}]{MY-CDFTpaper2019}%
  \BibitemOpen
  \bibfield  {author} {\bibinfo {author} {\bibfnamefont {A.}~\bibnamefont
  {Laestadius}}, \bibinfo {author} {\bibfnamefont {M.}~\bibnamefont {Penz}},
  \bibinfo {author} {\bibfnamefont {E.~I.}\ \bibnamefont {Tellgren}}, \bibinfo
  {author} {\bibfnamefont {M.}~\bibnamefont {Ruggenthaler}}, \bibinfo {author}
  {\bibfnamefont {S.}~\bibnamefont {Kvaal}},\ and\ \bibinfo {author}
  {\bibfnamefont {T.}~\bibnamefont {Helgaker}},\ }\bibfield  {title} {\enquote
  {\bibinfo {title} {{Kohn--Sham theory with paramagnetic currents:
  Compatibility and functional differentiability}},}\ }\href
  {https://doi.org/10.1021/acs.jctc.9b00141} {\bibfield  {journal} {\bibinfo
  {journal} {J. Chem. Theory Comput.}\ }\textbf {\bibinfo {volume} {15}},\
  \bibinfo {pages} {4003--4020} (\bibinfo {year} {2019})}\BibitemShut {NoStop}%
\bibitem [{\citenamefont {Lieb}\ and\ \citenamefont
  {Schrader}(2013)}]{LiebSchrader}%
  \BibitemOpen
  \bibfield  {author} {\bibinfo {author} {\bibfnamefont {E.~H.}\ \bibnamefont
  {Lieb}}\ and\ \bibinfo {author} {\bibfnamefont {R.}~\bibnamefont
  {Schrader}},\ }\bibfield  {title} {\enquote {\bibinfo {title} {Current
  densities in density-functional theory},}\ }\href
  {https://doi.org/10.1103/PhysRevA.88.032516} {\bibfield  {journal} {\bibinfo
  {journal} {Phys. Rev. A}\ }\textbf {\bibinfo {volume} {88}},\ \bibinfo
  {pages} {032516} (\bibinfo {year} {2013})}\BibitemShut {NoStop}%
\bibitem [{\citenamefont {Tellgren}, \citenamefont {Kvaal},\ and\ \citenamefont
  {Helgaker}(2014)}]{TellgrenNrep}%
  \BibitemOpen
  \bibfield  {author} {\bibinfo {author} {\bibfnamefont {E.~I.}\ \bibnamefont
  {Tellgren}}, \bibinfo {author} {\bibfnamefont {S.}~\bibnamefont {Kvaal}},\
  and\ \bibinfo {author} {\bibfnamefont {T.}~\bibnamefont {Helgaker}},\
  }\bibfield  {title} {\enquote {\bibinfo {title} {Fermion $n$-representability
  for prescribed density and paramagnetic current density},}\ }\href
  {https://doi.org/10.1103/PhysRevA.89.012515} {\bibfield  {journal} {\bibinfo
  {journal} {Phys. Rev. A}\ }\textbf {\bibinfo {volume} {89}},\ \bibinfo
  {pages} {012515} (\bibinfo {year} {2014})}\BibitemShut {NoStop}%
\bibitem [{\citenamefont {Laestadius}\ and\ \citenamefont
  {Tellgren}(2018)}]{LaestadiusTellgren2018}%
  \BibitemOpen
  \bibfield  {author} {\bibinfo {author} {\bibfnamefont {A.}~\bibnamefont
  {Laestadius}}\ and\ \bibinfo {author} {\bibfnamefont {E.~I.}\ \bibnamefont
  {Tellgren}},\ }\bibfield  {title} {\enquote {\bibinfo {title}
  {Density--wave-function mapping in degenerate current-density-functional
  theory},}\ }\href {https://doi.org/10.1103/PhysRevA.97.022514} {\bibfield
  {journal} {\bibinfo  {journal} {Phys. Rev. A}\ }\textbf {\bibinfo {volume}
  {97}},\ \bibinfo {pages} {022514} (\bibinfo {year} {2018})}\BibitemShut
  {NoStop}%
\bibitem [{\citenamefont {Lieb}\ and\ \citenamefont {Loss}(2001)}]{LiebLoss}%
  \BibitemOpen
  \bibfield  {author} {\bibinfo {author} {\bibfnamefont {E.~H.}\ \bibnamefont
  {Lieb}}\ and\ \bibinfo {author} {\bibfnamefont {M.}~\bibnamefont {Loss}},\
  }\href@noop {} {\emph {\bibinfo {title} {Analysis}}}\ (\bibinfo  {publisher}
  {American Mathematical Society, Providence, Rhode Island},\ \bibinfo {year}
  {2001})\BibitemShut {NoStop}%
\bibitem [{\citenamefont {Avron}, \citenamefont {Herbst},\ and\ \citenamefont
  {Simon}(1981)}]{Avron1981}%
  \BibitemOpen
  \bibfield  {author} {\bibinfo {author} {\bibfnamefont {J.~E.}\ \bibnamefont
  {Avron}}, \bibinfo {author} {\bibfnamefont {I.~W.}\ \bibnamefont {Herbst}},\
  and\ \bibinfo {author} {\bibfnamefont {B.}~\bibnamefont {Simon}},\ }\bibfield
   {title} {\enquote {\bibinfo {title} {Schrödinger operators with magnetic
  fields},}\ }\href {https://doi.org/10.1007/BF01209311} {\bibfield  {journal}
  {\bibinfo  {journal} {Commun. Math. Phys.}\ }\textbf {\bibinfo {volume}
  {79}},\ \bibinfo {pages} {529–572} (\bibinfo {year} {1981})}\BibitemShut
  {NoStop}%
\bibitem [{\citenamefont {Laestadius}(2014)}]{Laestadius2014}%
  \BibitemOpen
  \bibfield  {author} {\bibinfo {author} {\bibfnamefont {A.}~\bibnamefont
  {Laestadius}},\ }\bibfield  {title} {\enquote {\bibinfo {title} {Density
  functionals in the presence of magnetic field},}\ }\href
  {https://doi.org/10.1002/qua.24707} {\bibfield  {journal} {\bibinfo
  {journal} {Int. J. Quantum Chem.}\ }\textbf {\bibinfo {volume} {114}},\
  \bibinfo {pages} {1445--1456} (\bibinfo {year} {2014})}\BibitemShut {NoStop}%
\bibitem [{\citenamefont {Giesbertz}(2016)}]{giesbertz2016invertibility}%
  \BibitemOpen
  \bibfield  {author} {\bibinfo {author} {\bibfnamefont {K.~J.}\ \bibnamefont
  {Giesbertz}},\ }\bibfield  {title} {\enquote {\bibinfo {title} {Invertibility
  of the retarded response functions for initial mixed states: application to
  one-body reduced density matrix functional theory},}\ }\href
  {https://doi.org/10.1080/00268976.2016.1141253} {\bibfield  {journal}
  {\bibinfo  {journal} {Mol. Phys.}\ }\textbf {\bibinfo {volume} {114}},\
  \bibinfo {pages} {1128--1134} (\bibinfo {year} {2016})}\BibitemShut {NoStop}%
\bibitem [{\citenamefont {Valone}(1980)}]{Valone80}%
  \BibitemOpen
  \bibfield  {author} {\bibinfo {author} {\bibfnamefont {S.~M.}\ \bibnamefont
  {Valone}},\ }\bibfield  {title} {\enquote {\bibinfo {title} {Consequences of
  extending 1‐matrix energy functionals from pure–state representable to
  all ensemble representable 1 matrices},}\ }\href
  {https://doi.org/10.1063/1.440249} {\bibfield  {journal} {\bibinfo  {journal}
  {J. Chem. Phys.}\ }\textbf {\bibinfo {volume} {73}},\ \bibinfo {pages}
  {1344--1349} (\bibinfo {year} {1980})}\BibitemShut {NoStop}%
\bibitem [{\citenamefont {Penz}\ and\ \citenamefont {van
  Leeuwen}(2021)}]{penz-DFT-graphs}%
  \BibitemOpen
  \bibfield  {author} {\bibinfo {author} {\bibfnamefont {M.}~\bibnamefont
  {Penz}}\ and\ \bibinfo {author} {\bibfnamefont {R.}~\bibnamefont {van
  Leeuwen}},\ }\bibfield  {title} {\enquote {\bibinfo {title}
  {Density-functional theory on graphs},}\ }\href
  {https://doi.org/10.1063/5.0074249} {\bibfield  {journal} {\bibinfo
  {journal} {J. Chem. Phys.}\ }\textbf {\bibinfo {volume} {155}},\ \bibinfo
  {pages} {244111} (\bibinfo {year} {2021})}\BibitemShut {NoStop}%
\bibitem [{\citenamefont {Kvaal}\ \emph {et~al.}(2021)\citenamefont {Kvaal},
  \citenamefont {Laestadius}, \citenamefont {Tellgren},\ and\ \citenamefont
  {Helgaker}}]{Kvaal2021}%
  \BibitemOpen
  \bibfield  {author} {\bibinfo {author} {\bibfnamefont {S.}~\bibnamefont
  {Kvaal}}, \bibinfo {author} {\bibfnamefont {A.}~\bibnamefont {Laestadius}},
  \bibinfo {author} {\bibfnamefont {E.}~\bibnamefont {Tellgren}},\ and\
  \bibinfo {author} {\bibfnamefont {T.}~\bibnamefont {Helgaker}},\ }\bibfield
  {title} {\enquote {\bibinfo {title} {Lower semicontinuity of the universal
  functional in paramagnetic current–density functional theory},}\ }\href
  {https://doi.org/10.1021/acs.jpclett.0c03422} {\bibfield  {journal} {\bibinfo
   {journal} {J. Phys. Chem. Lett.}\ }\textbf {\bibinfo {volume} {12}},\
  \bibinfo {pages} {1421--1425} (\bibinfo {year} {2021})},\ \bibinfo {note}
  {pMID: 33522817}\BibitemShut {NoStop}%
\bibitem [{\citenamefont {Laestadius}\ \emph {et~al.}(2018)\citenamefont
  {Laestadius}, \citenamefont {Penz}, \citenamefont {Tellgren}, \citenamefont
  {Ruggenthaler}, \citenamefont {Kvaal},\ and\ \citenamefont
  {Helgaker}}]{KSpaper2018}%
  \BibitemOpen
  \bibfield  {author} {\bibinfo {author} {\bibfnamefont {A.}~\bibnamefont
  {Laestadius}}, \bibinfo {author} {\bibfnamefont {M.}~\bibnamefont {Penz}},
  \bibinfo {author} {\bibfnamefont {E.~I.}\ \bibnamefont {Tellgren}}, \bibinfo
  {author} {\bibfnamefont {M.}~\bibnamefont {Ruggenthaler}}, \bibinfo {author}
  {\bibfnamefont {S.}~\bibnamefont {Kvaal}},\ and\ \bibinfo {author}
  {\bibfnamefont {T.}~\bibnamefont {Helgaker}},\ }\bibfield  {title} {\enquote
  {\bibinfo {title} {{Generalized Kohn--Sham iteration on Banach spaces}},}\
  }\href {https://doi.org/10.1063/1.5037790} {\bibfield  {journal} {\bibinfo
  {journal} {J. Chem. Phys.}\ }\textbf {\bibinfo {volume} {149}},\ \bibinfo
  {pages} {164103} (\bibinfo {year} {2018})}\BibitemShut {NoStop}%
\bibitem [{\citenamefont {Tchenkoue}\ \emph {et~al.}(2019)\citenamefont
  {Tchenkoue}, \citenamefont {Penz}, \citenamefont {Theophilou}, \citenamefont
  {Ruggenthaler},\ and\ \citenamefont {Rubio}}]{tchenkoue2019force}%
  \BibitemOpen
  \bibfield  {author} {\bibinfo {author} {\bibfnamefont {M.-L.~M.}\
  \bibnamefont {Tchenkoue}}, \bibinfo {author} {\bibfnamefont {M.}~\bibnamefont
  {Penz}}, \bibinfo {author} {\bibfnamefont {I.}~\bibnamefont {Theophilou}},
  \bibinfo {author} {\bibfnamefont {M.}~\bibnamefont {Ruggenthaler}},\ and\
  \bibinfo {author} {\bibfnamefont {A.}~\bibnamefont {Rubio}},\ }\bibfield
  {title} {\enquote {\bibinfo {title} {Force balance approach for advanced
  approximations in density functional theories},}\ }\href
  {https://doi.org/10.1063/1.5123608} {\bibfield  {journal} {\bibinfo
  {journal} {J. Chem. Phys.}\ }\textbf {\bibinfo {volume} {151}},\ \bibinfo
  {pages} {154107} (\bibinfo {year} {2019})}\BibitemShut {NoStop}%
\bibitem [{\citenamefont {Kvaal}\ \emph {et~al.}(2014)\citenamefont {Kvaal},
  \citenamefont {Ekstr{\"o}m}, \citenamefont {Teale},\ and\ \citenamefont
  {Helgaker}}]{Kvaal2014}%
  \BibitemOpen
  \bibfield  {author} {\bibinfo {author} {\bibfnamefont {S.}~\bibnamefont
  {Kvaal}}, \bibinfo {author} {\bibfnamefont {U.}~\bibnamefont {Ekstr{\"o}m}},
  \bibinfo {author} {\bibfnamefont {A.~M.}\ \bibnamefont {Teale}},\ and\
  \bibinfo {author} {\bibfnamefont {T.}~\bibnamefont {Helgaker}},\ }\bibfield
  {title} {\enquote {\bibinfo {title} {{Differentiable but exact formulation of
  density-functional theory}},}\ }\href {https://doi.org/10.1063/1.4867005}
  {\bibfield  {journal} {\bibinfo  {journal} {J. Chem. Phys.}\ }\textbf
  {\bibinfo {volume} {140}},\ \bibinfo {pages} {18A518} (\bibinfo {year}
  {2014})}\BibitemShut {NoStop}%
\bibitem [{\citenamefont {Kvaal}()}]{Kvaal2022-MY}%
  \BibitemOpen
  \bibfield  {author} {\bibinfo {author} {\bibfnamefont {S.}~\bibnamefont
  {Kvaal}},\ }\href@noop {} {\enquote {\bibinfo {title} {{M}oreau--{Y}osida
  regularization in {DFT}},}\ }\bibinfo {howpublished} {(10 Aug 2022) arXiv
  e-prints [math.NA] 2208.05268},\ \bibinfo {note} {accessed
  2023-01-31}\BibitemShut {NoStop}%
\bibitem [{\citenamefont {Zhao}, \citenamefont {Morrison},\ and\ \citenamefont
  {Parr}(1994)}]{ZMP1994}%
  \BibitemOpen
  \bibfield  {author} {\bibinfo {author} {\bibfnamefont {Q.}~\bibnamefont
  {Zhao}}, \bibinfo {author} {\bibfnamefont {R.~C.}\ \bibnamefont {Morrison}},\
  and\ \bibinfo {author} {\bibfnamefont {R.~G.}\ \bibnamefont {Parr}},\
  }\bibfield  {title} {\enquote {\bibinfo {title} {From electron densities to
  {K}ohn--{S}ham kinetic energies, orbital energies, exchange-correlation
  potentials, and exchange-correlation energies},}\ }\href
  {https://doi.org/10.1103/PhysRevA.50.2138} {\bibfield  {journal} {\bibinfo
  {journal} {Phys. Rev. A}\ }\textbf {\bibinfo {volume} {50}},\ \bibinfo
  {pages} {2138} (\bibinfo {year} {1994})}\BibitemShut {NoStop}%
\bibitem [{\citenamefont {Penz}, \citenamefont {Csirik},\ and\ \citenamefont
  {Laestadius}(2023)}]{Penz2022ZMP}%
  \BibitemOpen
  \bibfield  {author} {\bibinfo {author} {\bibfnamefont {M.}~\bibnamefont
  {Penz}}, \bibinfo {author} {\bibfnamefont {M.~A.}\ \bibnamefont {Csirik}},\
  and\ \bibinfo {author} {\bibfnamefont {A.}~\bibnamefont {Laestadius}},\
  }\bibfield  {title} {\enquote {\bibinfo {title} {Density-potential inversion
  from {M}oreau--{Y}osida regularization},}\ }\href
  {https://doi.org/10.1088/2516-1075/acc626} {\bibfield  {journal} {\bibinfo
  {journal} {Electron. Struct.}\ }\textbf {\bibinfo {volume} {5}},\ \bibinfo
  {pages} {014009} (\bibinfo {year} {2023})}\BibitemShut {NoStop}%
\bibitem [{\citenamefont {Tellgren}\ \emph {et~al.}(2018)\citenamefont
  {Tellgren}, \citenamefont {Laestadius}, \citenamefont {Helgaker},
  \citenamefont {Kvaal},\ and\ \citenamefont {Teale}}]{Tellgren2018}%
  \BibitemOpen
  \bibfield  {author} {\bibinfo {author} {\bibfnamefont {E.~I.}\ \bibnamefont
  {Tellgren}}, \bibinfo {author} {\bibfnamefont {A.}~\bibnamefont
  {Laestadius}}, \bibinfo {author} {\bibfnamefont {T.}~\bibnamefont
  {Helgaker}}, \bibinfo {author} {\bibfnamefont {S.}~\bibnamefont {Kvaal}},\
  and\ \bibinfo {author} {\bibfnamefont {A.~M.}\ \bibnamefont {Teale}},\
  }\bibfield  {title} {\enquote {\bibinfo {title} {Uniform magnetic fields in
  density-functional theory},}\ }\href {https://doi.org/10.1063/1.5007300}
  {\bibfield  {journal} {\bibinfo  {journal} {J. Chem. Phys.}\ }\textbf
  {\bibinfo {volume} {148}},\ \bibinfo {pages} {024101} (\bibinfo {year}
  {2018})}\BibitemShut {NoStop}%
\bibitem [{\citenamefont {Barcelo}\ \emph {et~al.}(1988)\citenamefont
  {Barcelo}, \citenamefont {Kenig}, \citenamefont {Ruiz},\ and\ \citenamefont
  {Sogge}}]{BKRS}%
  \BibitemOpen
  \bibfield  {author} {\bibinfo {author} {\bibfnamefont {B.}~\bibnamefont
  {Barcelo}}, \bibinfo {author} {\bibfnamefont {C.~E.}\ \bibnamefont {Kenig}},
  \bibinfo {author} {\bibfnamefont {A.}~\bibnamefont {Ruiz}},\ and\ \bibinfo
  {author} {\bibfnamefont {C.~D.}\ \bibnamefont {Sogge}},\ }\bibfield  {title}
  {\enquote {\bibinfo {title} {Weighted {S}obolev inequalities and unique
  continuation for the {L}aplacian plus lower order terms},}\ }\href
  {https://doi.org/10.1215/ijm/1255989128} {\bibfield  {journal} {\bibinfo
  {journal} {Illinois J. Math.}\ }\textbf {\bibinfo {volume} {32}},\ \bibinfo
  {pages} {230--245} (\bibinfo {year} {1988})}\BibitemShut {NoStop}%
\bibitem [{\citenamefont {Wolff}(1992)}]{Wolff1992}%
  \BibitemOpen
  \bibfield  {author} {\bibinfo {author} {\bibfnamefont {T.~H.}\ \bibnamefont
  {Wolff}},\ }\bibfield  {title} {\enquote {\bibinfo {title} {A property of
  measures in $\mathbb{R}^n$ and an application to unique continuation},}\
  }\href {https://doi.org/10.1007/BF01896975} {\bibfield  {journal} {\bibinfo
  {journal} {Geom. Funct. Anal.}\ }\textbf {\bibinfo {volume} {2}},\ \bibinfo
  {pages} {225--284} (\bibinfo {year} {1992})}\BibitemShut {NoStop}%
\bibitem [{\citenamefont {Kurata}(1993)}]{Kurata1}%
  \BibitemOpen
  \bibfield  {author} {\bibinfo {author} {\bibfnamefont {K.}~\bibnamefont
  {Kurata}},\ }\bibfield  {title} {\enquote {\bibinfo {title} {A unique
  continuation theorem for uniformly elliptic equations with strongly singular
  potentials},}\ }\href {https://doi.org/10.1080/03605309308820968} {\bibfield
  {journal} {\bibinfo  {journal} {Comm. in P.D.E.}\ }\textbf {\bibinfo {volume}
  {18}},\ \bibinfo {pages} {1161--1189} (\bibinfo {year} {1993})}\BibitemShut
  {NoStop}%
\bibitem [{\citenamefont {Kurata}(1997)}]{Kurata2}%
  \BibitemOpen
  \bibfield  {author} {\bibinfo {author} {\bibfnamefont {K.}~\bibnamefont
  {Kurata}},\ }\bibfield  {title} {\enquote {\bibinfo {title} {A unique
  continuation theorem for the {S}chr\"odinger equation with singular magnetic
  field},}\ }\href {https://doi.org/10.1090/S0002-9939-97-03672-1} {\bibfield
  {journal} {\bibinfo  {journal} {Proc. Amer. Math. Soc.}\ }\textbf {\bibinfo
  {volume} {125}},\ \bibinfo {pages} {853--860} (\bibinfo {year}
  {1997})}\BibitemShut {NoStop}%
\bibitem [{\citenamefont {Regbaoui}(2001)}]{Regbaoui}%
  \BibitemOpen
  \bibfield  {author} {\bibinfo {author} {\bibfnamefont {R.}~\bibnamefont
  {Regbaoui}},\ }\bibfield  {title} {\enquote {\bibinfo {title} {Unique
  continuation from sets of positive measure},}\ }in\ \href
  {https://doi.org/10.1007/978-1-4612-0203-5_13} {\emph {\bibinfo {booktitle}
  {Carleman Estimates and Applications to Uniqueness and Control Theory}}},\
  \bibinfo {series} {Progress in Nonlinear Differential Equations and Their
  Applications}, Vol.~\bibinfo {volume} {46},\ \bibinfo {editor} {edited by\
  \bibinfo {editor} {\bibfnamefont {F.}~\bibnamefont {Colombini}}\ and\
  \bibinfo {editor} {\bibfnamefont {C.}~\bibnamefont {Zuily}}}\ (\bibinfo
  {publisher} {Birkh\"auser, Basel},\ \bibinfo {year} {2001})\ pp.\ \bibinfo
  {pages} {179--190}\BibitemShut {NoStop}%
\bibitem [{\citenamefont {Laestadius}, \citenamefont {Benedicks},\ and\
  \citenamefont {Penz}(2020)}]{LaestadiusBenedicksPenz}%
  \BibitemOpen
  \bibfield  {author} {\bibinfo {author} {\bibfnamefont {A.}~\bibnamefont
  {Laestadius}}, \bibinfo {author} {\bibfnamefont {M.}~\bibnamefont
  {Benedicks}},\ and\ \bibinfo {author} {\bibfnamefont {M.}~\bibnamefont
  {Penz}},\ }\bibfield  {title} {\enquote {\bibinfo {title} {Unique
  continuation for the magnetic {S}chrödinger equation},}\ }\href
  {https://doi.org/10.1002/qua.26149} {\bibfield  {journal} {\bibinfo
  {journal} {Int. J. Quantum Chem.}\ }\textbf {\bibinfo {volume} {120}},\
  \bibinfo {pages} {e26149} (\bibinfo {year} {2020})}\BibitemShut {NoStop}%
\bibitem [{\citenamefont {Garrigue}(2020)}]{Garrigue2020-magneticHK}%
  \BibitemOpen
  \bibfield  {author} {\bibinfo {author} {\bibfnamefont {L.}~\bibnamefont
  {Garrigue}},\ }\bibfield  {title} {\enquote {\bibinfo {title} {{Unique
  continuation for many-body {S}chr{\"o}dinger operators and the
  {H}ohenberg--{K}ohn theorem. II. The Pauli Hamiltonian}},}\ }\href
  {https://doi.org/10.4171/DM/765} {\bibfield  {journal} {\bibinfo  {journal}
  {Doc. Math.}\ }\textbf {\bibinfo {volume} {25}},\ \bibinfo {pages} {869--898}
  (\bibinfo {year} {2020})}\BibitemShut {NoStop}%
\bibitem [{\citenamefont {Reimann}\ \emph {et~al.}(2017)\citenamefont
  {Reimann}, \citenamefont {Borgoo}, \citenamefont {Tellgren}, \citenamefont
  {Teale},\ and\ \citenamefont {Helgaker}}]{REIMANN_JCTC13_4089}%
  \BibitemOpen
  \bibfield  {author} {\bibinfo {author} {\bibfnamefont {S.}~\bibnamefont
  {Reimann}}, \bibinfo {author} {\bibfnamefont {A.}~\bibnamefont {Borgoo}},
  \bibinfo {author} {\bibfnamefont {E.~I.}\ \bibnamefont {Tellgren}}, \bibinfo
  {author} {\bibfnamefont {A.~M.}\ \bibnamefont {Teale}},\ and\ \bibinfo
  {author} {\bibfnamefont {T.}~\bibnamefont {Helgaker}},\ }\bibfield  {title}
  {\enquote {\bibinfo {title} {Magnetic-field density-functional theory
  ({BDFT}): Lessons from the adiabatic connection},}\ }\href
  {https://doi.org/10.1021/acs.jctc.7b00295} {\bibfield  {journal} {\bibinfo
  {journal} {J. Chem. Theory Comput.}\ }\textbf {\bibinfo {volume} {13}},\
  \bibinfo {pages} {4089--4100} (\bibinfo {year} {2017})}\BibitemShut {NoStop}%
\bibitem [{\citenamefont {Vignale}(2004)}]{VIGNALE_PRB70_201102}%
  \BibitemOpen
  \bibfield  {author} {\bibinfo {author} {\bibfnamefont {G.}~\bibnamefont
  {Vignale}},\ }\bibfield  {title} {\enquote {\bibinfo {title} {Mapping from
  current densities to vector potentials in time-dependent current density
  functional theory},}\ }\href
  {http://link.aps.org/doi/10.1103/PhysRevB.70.201102} {\bibfield  {journal}
  {\bibinfo  {journal} {Phys. Rev. B}\ }\textbf {\bibinfo {volume} {70}},\
  \bibinfo {pages} {201102} (\bibinfo {year} {2004})}\BibitemShut {NoStop}%
\bibitem [{\citenamefont {Vignale}\ and\ \citenamefont
  {Kohn}(1996)}]{VIGNALE_PRL77_2037}%
  \BibitemOpen
  \bibfield  {author} {\bibinfo {author} {\bibfnamefont {G.}~\bibnamefont
  {Vignale}}\ and\ \bibinfo {author} {\bibfnamefont {W.}~\bibnamefont {Kohn}},\
  }\bibfield  {title} {\enquote {\bibinfo {title} {Current-dependent
  exchange-correlation potential for dynamical linear response theory},}\
  }\href {https://doi.org/10.1103/PhysRevLett.77.2037} {\bibfield  {journal}
  {\bibinfo  {journal} {Phys. Rev. Lett.}\ }\textbf {\bibinfo {volume} {77}},\
  \bibinfo {pages} {2037--2040} (\bibinfo {year} {1996})}\BibitemShut {NoStop}%
\bibitem [{\citenamefont {Vignale}, \citenamefont {Ullrich},\ and\
  \citenamefont {Capelle}(2013)}]{Vignale-comment-2013}%
  \BibitemOpen
  \bibfield  {author} {\bibinfo {author} {\bibfnamefont {G.}~\bibnamefont
  {Vignale}}, \bibinfo {author} {\bibfnamefont {C.~A.}\ \bibnamefont
  {Ullrich}},\ and\ \bibinfo {author} {\bibfnamefont {K.}~\bibnamefont
  {Capelle}},\ }\bibfield  {title} {\enquote {\bibinfo {title} {Comment on
  `{D}ensity and physical current density functional theory' by {X}iao-{Y}in
  {P}an and {V}iraht {S}ahni},}\ }\href {https://doi.org/10.1002/qua.24327}
  {\bibfield  {journal} {\bibinfo  {journal} {Int. J. Quantum Chem.}\ }\textbf
  {\bibinfo {volume} {113}},\ \bibinfo {pages} {1422--1423} (\bibinfo {year}
  {2013})}\BibitemShut {NoStop}%
\bibitem [{\citenamefont {Byers}\ and\ \citenamefont
  {Yang}(1961)}]{Byers_1961}%
  \BibitemOpen
  \bibfield  {author} {\bibinfo {author} {\bibfnamefont {N.}~\bibnamefont
  {Byers}}\ and\ \bibinfo {author} {\bibfnamefont {C.~N.}\ \bibnamefont
  {Yang}},\ }\bibfield  {title} {\enquote {\bibinfo {title} {Theoretical
  considerations concerning quantized magnetic flux in superconducting
  cylinders},}\ }\href {https://doi.org/10.1103/physrevlett.7.46} {\bibfield
  {journal} {\bibinfo  {journal} {Phys. Rev. Lett.}\ }\textbf {\bibinfo
  {volume} {7}},\ \bibinfo {pages} {46--49} (\bibinfo {year}
  {1961})}\BibitemShut {NoStop}%
\bibitem [{\citenamefont {von Ragu{\'{e}}~Schleyer}\ \emph
  {et~al.}(1996)\citenamefont {von Ragu{\'{e}}~Schleyer}, \citenamefont
  {Maerker}, \citenamefont {Dransfeld}, \citenamefont {Jiao},\ and\
  \citenamefont {van Eikema~Hommes}}]{SCHLEYER_JACS118_6317}%
  \BibitemOpen
  \bibfield  {author} {\bibinfo {author} {\bibfnamefont {P.}~\bibnamefont {von
  Ragu{\'{e}}~Schleyer}}, \bibinfo {author} {\bibfnamefont {C.}~\bibnamefont
  {Maerker}}, \bibinfo {author} {\bibfnamefont {A.}~\bibnamefont {Dransfeld}},
  \bibinfo {author} {\bibfnamefont {H.}~\bibnamefont {Jiao}},\ and\ \bibinfo
  {author} {\bibfnamefont {N.~J.~R.}\ \bibnamefont {van Eikema~Hommes}},\
  }\bibfield  {title} {\enquote {\bibinfo {title} {Nucleus-independent chemical
  shifts: A simple and efficient aromaticity probe},}\ }\href
  {https://doi.org/10.1021/ja960582d} {\bibfield  {journal} {\bibinfo
  {journal} {J. Am. Chem. Soc.}\ }\textbf {\bibinfo {volume} {118}},\ \bibinfo
  {pages} {6317--6318} (\bibinfo {year} {1996})}\BibitemShut {NoStop}%
\bibitem [{\citenamefont {Ryder}(1996)}]{ryder_1996}%
  \BibitemOpen
  \bibfield  {author} {\bibinfo {author} {\bibfnamefont {L.~H.}\ \bibnamefont
  {Ryder}},\ }\href@noop {} {\emph {\bibinfo {title} {Quantum Field Theory}}},\
  \bibinfo {edition} {2nd}\ ed.\ (\bibinfo  {publisher} {Cambridge University
  Press},\ \bibinfo {year} {1996})\BibitemShut {NoStop}%
\bibitem [{\citenamefont {Greiner}\ and\ \citenamefont
  {Reinhardt}(1996)}]{Greiner_1996}%
  \BibitemOpen
  \bibfield  {author} {\bibinfo {author} {\bibfnamefont {W.}~\bibnamefont
  {Greiner}}\ and\ \bibinfo {author} {\bibfnamefont {J.}~\bibnamefont
  {Reinhardt}},\ }\href@noop {} {\emph {\bibinfo {title} {Field
  Quantization}}}\ (\bibinfo  {publisher} {Springer Berlin Heidelberg},\
  \bibinfo {year} {1996})\BibitemShut {NoStop}%
\bibitem [{\citenamefont {Silberstein}(1907)}]{silberstein1907}%
  \BibitemOpen
  \bibfield  {author} {\bibinfo {author} {\bibfnamefont {L.}~\bibnamefont
  {Silberstein}},\ }\bibfield  {title} {\enquote {\bibinfo {title}
  {Elektromagnetische {G}rundgleichungen in bivektorieller {B}ehandlung},}\
  }\href {https://doi.org/10.1002/andp.19073270313} {\bibfield  {journal}
  {\bibinfo  {journal} {Ann. Phys.}\ }\textbf {\bibinfo {volume} {327}},\
  \bibinfo {pages} {579--586} (\bibinfo {year} {1907})}\BibitemShut {NoStop}%
\bibitem [{\citenamefont {Oppenheimer}(1931)}]{oppenheimer1931}%
  \BibitemOpen
  \bibfield  {author} {\bibinfo {author} {\bibfnamefont {J.~R.}\ \bibnamefont
  {Oppenheimer}},\ }\bibfield  {title} {\enquote {\bibinfo {title} {Note on
  light quanta and the electromagnetic field},}\ }\href
  {https://doi.org/10.1103/PhysRev.38.725} {\bibfield  {journal} {\bibinfo
  {journal} {Phys. Rev.}\ }\textbf {\bibinfo {volume} {38}},\ \bibinfo {pages}
  {725--746} (\bibinfo {year} {1931})}\BibitemShut {NoStop}%
\bibitem [{\citenamefont {Bia{\l}ynicki-Birula}(1994)}]{bialynicki1994wave}%
  \BibitemOpen
  \bibfield  {author} {\bibinfo {author} {\bibfnamefont {I.}~\bibnamefont
  {Bia{\l}ynicki-Birula}},\ }\bibfield  {title} {\enquote {\bibinfo {title} {On
  the wave function of the photon},}\ }\href
  {https://doi.org/10.12693/APHYSPOLA.86.97} {\bibfield  {journal} {\bibinfo
  {journal} {Acta Phys. Pol. A}\ }\textbf {\bibinfo {volume} {1}},\ \bibinfo
  {pages} {97--116} (\bibinfo {year} {1994})}\BibitemShut {NoStop}%
\bibitem [{\citenamefont {Gersten}(1999)}]{gersten1999maxwell}%
  \BibitemOpen
  \bibfield  {author} {\bibinfo {author} {\bibfnamefont {A.}~\bibnamefont
  {Gersten}},\ }\bibfield  {title} {\enquote {\bibinfo {title} {Maxwell’s
  equations as the one-photon quantum equation},}\ }\href
  {https://doi.org/10.1023/A:1017551920941} {\bibfield  {journal} {\bibinfo
  {journal} {Found. Phys. Lett.}\ }\textbf {\bibinfo {volume} {12}},\ \bibinfo
  {pages} {291--298} (\bibinfo {year} {1999})}\BibitemShut {NoStop}%
\bibitem [{\citenamefont {Baez}, \citenamefont {Segal},\ and\ \citenamefont
  {Zhou}(2014)}]{baez2014introduction}%
  \BibitemOpen
  \bibfield  {author} {\bibinfo {author} {\bibfnamefont {J.~C.}\ \bibnamefont
  {Baez}}, \bibinfo {author} {\bibfnamefont {I.~E.}\ \bibnamefont {Segal}},\
  and\ \bibinfo {author} {\bibfnamefont {Z.}~\bibnamefont {Zhou}},\ }\bibfield
  {title} {\enquote {\bibinfo {title} {Introduction to algebraic and
  constructive quantum field theory},}\ }in\ \href@noop {} {\emph {\bibinfo
  {booktitle} {Introduction to Algebraic and Constructive Quantum Field
  Theory}}}\ (\bibinfo  {publisher} {Princeton University Press},\ \bibinfo
  {year} {2014})\BibitemShut {NoStop}%
\bibitem [{\citenamefont {Scharf}(2014)}]{scharf2014finite}%
  \BibitemOpen
  \bibfield  {author} {\bibinfo {author} {\bibfnamefont {G.}~\bibnamefont
  {Scharf}},\ }\href@noop {} {\emph {\bibinfo {title} {Finite quantum
  electrodynamics: the causal approach}}}\ (\bibinfo  {publisher} {Courier
  Corporation},\ \bibinfo {year} {2014})\BibitemShut {NoStop}%
\bibitem [{\citenamefont {Thirring}(2013)}]{thirring2013quantum}%
  \BibitemOpen
  \bibfield  {author} {\bibinfo {author} {\bibfnamefont {W.}~\bibnamefont
  {Thirring}},\ }\href@noop {} {\emph {\bibinfo {title} {Quantum mathematical
  physics: Atoms, molecules and large systems}}}\ (\bibinfo  {publisher}
  {Springer Science \& Business Media},\ \bibinfo {year} {2013})\BibitemShut
  {NoStop}%
\bibitem [{\citenamefont {Keller}(2012)}]{keller2012quantum}%
  \BibitemOpen
  \bibfield  {author} {\bibinfo {author} {\bibfnamefont {O.}~\bibnamefont
  {Keller}},\ }\href@noop {} {\emph {\bibinfo {title} {Quantum theory of
  near-field electrodynamics}}}\ (\bibinfo  {publisher} {Springer Science \&
  Business Media},\ \bibinfo {year} {2012})\BibitemShut {NoStop}%
\bibitem [{\citenamefont {Spohn}(2004)}]{spohn2004dynamics}%
  \BibitemOpen
  \bibfield  {author} {\bibinfo {author} {\bibfnamefont {H.}~\bibnamefont
  {Spohn}},\ }\href@noop {} {\emph {\bibinfo {title} {Dynamics of charged
  particles and their radiation field}}}\ (\bibinfo  {publisher} {Cambridge
  university press},\ \bibinfo {year} {2004})\BibitemShut {NoStop}%
\bibitem [{\citenamefont {Ruggenthaler}()}]{Ruggenthaler2015}%
  \BibitemOpen
  \bibfield  {author} {\bibinfo {author} {\bibfnamefont {M.}~\bibnamefont
  {Ruggenthaler}},\ }\href@noop {} {\enquote {\bibinfo {title} {Ground-state
  quantum-electro\-dynamical density-functional theory},}\ }\bibinfo
  {howpublished} {(4 Aug 2017) arXiv e-prints [quant-ph] 1509.01417},\ \bibinfo
  {note} {accessed 2023-01-31}\BibitemShut {NoStop}%
\bibitem [{\citenamefont {Hainzl}\ and\ \citenamefont
  {Seiringer}(2002)}]{hainzl2002mass}%
  \BibitemOpen
  \bibfield  {author} {\bibinfo {author} {\bibfnamefont {C.}~\bibnamefont
  {Hainzl}}\ and\ \bibinfo {author} {\bibfnamefont {R.}~\bibnamefont
  {Seiringer}},\ }\bibfield  {title} {\enquote {\bibinfo {title} {Mass
  renormalization and energy level shift in non-relativistic {QED}},}\ }\href
  {https://doi.org/10.48550/arXiv.math-ph/0205044} {\bibfield  {journal}
  {\bibinfo  {journal} {Adv. Theor. Math. Phys}\ }\textbf {\bibinfo {volume}
  {6}},\ \bibinfo {pages} {847--871} (\bibinfo {year} {2002})}\BibitemShut
  {NoStop}%
\bibitem [{\citenamefont {Garcia-Vidal}, \citenamefont {Ciuti},\ and\
  \citenamefont {Ebbesen}(2021)}]{garcia2021manipulating}%
  \BibitemOpen
  \bibfield  {author} {\bibinfo {author} {\bibfnamefont {F.~J.}\ \bibnamefont
  {Garcia-Vidal}}, \bibinfo {author} {\bibfnamefont {C.}~\bibnamefont
  {Ciuti}},\ and\ \bibinfo {author} {\bibfnamefont {T.~W.}\ \bibnamefont
  {Ebbesen}},\ }\bibfield  {title} {\enquote {\bibinfo {title} {Manipulating
  matter by strong coupling to vacuum fields},}\ }\href
  {https://doi.org/10.1126/science.abd0336} {\bibfield  {journal} {\bibinfo
  {journal} {Science}\ }\textbf {\bibinfo {volume} {373}},\ \bibinfo {pages}
  {eabd0336} (\bibinfo {year} {2021})}\BibitemShut {NoStop}%
\bibitem [{\citenamefont {Ruggenthaler}, \citenamefont {Sidler},\ and\
  \citenamefont {Rubio}()}]{ruggenthaler2022understanding}%
  \BibitemOpen
  \bibfield  {author} {\bibinfo {author} {\bibfnamefont {M.}~\bibnamefont
  {Ruggenthaler}}, \bibinfo {author} {\bibfnamefont {D.}~\bibnamefont
  {Sidler}},\ and\ \bibinfo {author} {\bibfnamefont {A.}~\bibnamefont
  {Rubio}},\ }\href@noop {} {\enquote {\bibinfo {title} {Understanding
  polaritonic chemistry from ab initio quantum electrodynamics},}\ }\bibinfo
  {howpublished} {(8 Nov 2022) arXiv e-prints [quant-ph] 2211.04241},\ \bibinfo
  {note} {accessed 2023-03-01}\BibitemShut {NoStop}%
\end{thebibliography}
\end{document}